\documentclass[11pt]{article}

\usepackage{amsmath}
\usepackage[sc]{mathpazo}

\usepackage[colorlinks]{hyperref}
\hypersetup{linkcolor=[rgb]{0,0.4,0.7},filecolor=[rgb]{0,0.4,0.7},citecolor=[rgb]{0,0.4,0.7},urlcolor=[rgb]{0,0,0.7}}
\usepackage{xspace}
\usepackage{fullpage}
\usepackage{boxedminipage}
\usepackage[boxed]{algorithm}
\usepackage{epigraph}
\usepackage{framed}
\usepackage[framemethod=tikz]{mdframed}
\usepackage{titlesec}
\usepackage{lipsum}%

\usepackage{tikz}
\usetikzlibrary{shapes.geometric}
\usetikzlibrary{arrows}
\usetikzlibrary{arrows.meta}
\usetikzlibrary{patterns}
\usetikzlibrary{shapes.misc}
\newcommand{\OptCCVal}{\mathsf{OPT}_{\mathsf{C}}}
\def\cAOPT{{\cal A}_{\mathsf{OPT}}}
\newcommand{\local}{${\mathsf{LOCAL}}$}
\newcommand{\congest}{${\mathsf{CONGEST}}$}

\newcommand{\dist}{\mbox{\rm dist}}
\newcommand{\OptimalCycleCover}{\mathsf{OptimalCycleCover}}
\newcommand{\OptimalEdgeCycleCover}{\mathsf{OptimalEdgeCycleCover}}

\newcommand{\dilation}{\mbox{\tt d}}
\newcommand{\congestion}{\mbox{\tt c}}
\newcommand{\densitythreshold}{\mbox{\tt b}}

\newcommand{\density}{\phi}
\newcommand{\NonTreeMinorClosed}{\mathsf{NonTreeMinorClosed}}

\newcommand{\supergraph}{\widetilde{G}}

\newcommand{\NonTreeCover}{\mathsf{NonTreeCover}}
\newcommand{\TreeCover}{\mathsf{TreeCover}}
\newcommand{\Partition}{\mathsf{Partition}}

\newcommand{\DistTreeCover}{\mathsf{DistTreeCover}}
\newcommand{\EdgeDisjointPath}{\mathsf{TreeEdgeDisjointPath}}
\newcommand{\GraphCover}{\mathsf{CycleCover}}

\newcommand{\SimplifyCycles}{\mathsf{SimplifyCycles}}

\newcommand{\swap}{\mathsf{Swap}}
\newcommand{\Diam}{\mathsf{Diam}}
\newcommand{\depth}{\mathsf{Depth}}

\def\cA{{\cal A}}
\def\cC{{\cal C}}
\def\cN{{\cal N}}

\renewcommand{\paragraph}[1]{\vspace{0.15cm}\noindent {\bf #1}}
\titlespacing*{\section}{0pt}{1.1\baselineskip}{\baselineskip}
\titlespacing*{\subsection}{-5pt}{\baselineskip}{\baselineskip}

\usepackage{amsthm,amsmath,amssymb}
\usepackage{cleveref,aliascnt}

\newtheorem{theorem}{Theorem}
\newtheorem{fact}{Fact}
\newtheorem{observation}{Observation}
\newtheorem{definition}{Definition}

\newtheorem{lemma}{Lemma}

\newtheorem{claim}{Claim}

\usepackage{tikz}
\usepackage{relsize}
\usepackage{ctable}

\crefname{claim}{Claim}{Claims}
\crefname{observation}{Observation}{Observations}

\newcommand{\A}{{\mathcal A}}

\newcommand{\bit}{\{0,1\}}

\newcommand{\ie}  {i.e.,\ }
\newcommand{\eg}  {e.g.,\ }

\newcommand{\poly}{\mathsf{poly}}

\newcommand{\ith}[1]{{#1}\textsuperscript{th}}

\newcommand{\ignore}[1]{}

\newcommand{\OPT}{{\mathsf{OPT}}}

\title{Low Congestion Cycle Covers and Their Applications}
\author{Merav Parter\thanks{Department of Computer Science and
		Applied Mathematics, Weizmann Institute of Science, Israel.  Supported 
		in part by grants from the Israel Science Foundation (no.\ 2084/18).}
	\and Eylon Yogev\footnotemark[1]
}

\date{}

\begin{document}
\maketitle
\begin{abstract}
A cycle cover of a bridgeless graph $G$ is a collection of simple cycles in $G$ such that each edge $e$ appears on at least one cycle. The common objective in cycle cover computation is to minimize the total lengths of all cycles.
Motivated by applications to distributed computation, we introduce the notion of \emph{low-congestion} cycle covers, in which all cycles in the cycle collection are both \emph{short} and nearly \emph{edge-disjoint}. Formally, a $(\dilation,\congestion)$-cycle cover of a graph $G$ is a collection of cycles in $G$ in which each cycle is of length at most $\dilation$ and each edge participates in at least one cycle and at most $\congestion$ cycles. 

A-priori, it is not clear that cycle covers that enjoy both a small overlap and 
a short cycle length even exist, nor if it is possible to efficiently find them. 
Perhaps quite surprisingly, we 
prove the following: Every bridgeless graph of diameter $D$ admits a $(\dilation,\congestion)$-cycle cover where $\dilation = \tilde{O}(D)$ and $\congestion=\tilde{O}(1)$. That is, the edges of $G$ can be covered by cycles such that each cycle is of 
length at most $\widetilde{O}(D)$ and each edge participates in at most 
$\widetilde{O}(1)$ cycles. These parameters are existentially tight up to polylogarithmic terms.

Furthermore, we show how to extend our result to achieve universally optimal cycle covers. Let $C_e$ is the length of the shortest cycle that covers $e$, and let $\OPT(G)= \max_{e \in G} C_e$. We show that every bridgeless graph admits a $(\dilation,\congestion)$-cycle cover where $\dilation = \tilde{O}(\OPT(G))$ and $\congestion=\tilde{O}(1)$.

We demonstrate the usefulness of low congestion cycle covers in different settings of resilient computation. For instance, we consider a Byzantine fault model where in each round, the adversary chooses a single message and corrupt in an arbitrarily manner. 
We provide a compiler that turns any $r$-round distributed algorithm for a 
graph $G$ with diameter $D$, into an equivalent fault tolerant algorithm with 
$r\cdot \poly(D)$ rounds.
\end{abstract}

\thispagestyle{empty}
\newpage
\tableofcontents
\thispagestyle{empty}
\newpage

\setcounter{page}{1}
\section{Introduction}
A cycle cover of a graph $G$ is a collection of cycles such that each edge of $G$ appears in at least one of the cycles. Cycle covers were introduced by Itai and Rodeh \cite{itai1978covering} in 1978 with the objective to cover all edges of a bridgeless\footnote{A graph $G$ is bridgeless, if any single edge removal keeps the graph connected.} graph with cycles of \emph{total} minimum length. 
This objective finds applications in centralized routing, robot navigation and fault-tolerant optical networks \cite{hochbaum2001bounded}. For instance, in the related Chinese Postman Problem, introduced in 1962 by Guan \cite{guan1962graphic,edmonds1973matching}, the objective is to compute the shortest tour that covers each edge by a cycle.  Szekeresand \cite{szekeres1973polyhedral} and Seymour \cite{Seymour79} independently have conjectured that every bridgeless graph has a cycle cover in which each edge is covered by exactly two cycles, this is known as the \emph{double cycle cover} conjecture. 
Many variants of cycle covers have been studied throughout the years from the combinatorial and the optimization point of views
\cite{fan1992integer,thomassen1997complexity,hochbaum2001bounded,immorlica2005cycle,blaser2005approximating,krivelevich2005approximation,manthey2009minimum,khachay2016approximability}. 


\subsection{Low Congestion Cycle Covers}
Motivated by various applications for \emph{resilient distributed computing}, we introduce a new notion of \emph{low-congestion} cycle covers: a collection of cycles that cover all graph edges by cycles that are both \emph{short} and almost \emph{edge-disjoint}. The efficiency of our low-congestion cover is measured by the key parameters of packet routing \cite{leighton1994packet}: dilation (length of largest cycle) and congestion (maximum edge overlap of cycles).
Formally, a $(\dilation,\congestion)$-cycle cover of a graph $G$ is a collection of cycles in $G$ in which each cycle is of length at most $\dilation$, and each edge participates in at least one cycle and at most $\congestion$ cycles. 
Using the beautiful result of Leighton, Maggs and Rao  \cite{leighton1994packet} and the follow-up of \cite{Ghaffari15}, a $(\dilation,\congestion)$-cycle cover allows one to route information on all cycles simultaneously in time $\widetilde{O}(\dilation+\congestion)$.

Since $n$-vertex graphs with at least $2n$ edges have girth $O(\log n)$, one can cover all but $2n$ edges in $G$, by 
edge-disjoint cycles of length $O(\log n)$ (e.g., by repeatedly omitting short cycles from $G$). 
For a bridgeless graph with diameter $D$, it is easy to cover the remaining graph edges with cycles of length $O(D)$, which is optimal (e.g., the cycle graph).
This can be done by covering each edge $e=(u,v)$ using the alternative $u$-$v$ shortest path in $G\setminus \{e\}$. Although providing short cycles, such an approach might create cycles with a large overlap, e.g., where a single edge appears on many (\eg $\Omega(n)$) of the cycles. Indeed, a-priori, it is not clear that cycle covers that enjoy both low congestion and short lengths, say $O(D)$, even exist, nor if it is possible to efficiently find them. 
Perhaps surprisingly, our main result shows that such covers exist and in 
particular, one can enjoy a dilation of $O(D\log n)$ while incurring only a 
poly-logarithmic congestion.
\begin{theorem}[Low Congestion Cycle Cover]\label{thm:cyclecover_upper}
Every bridgeless graph with diameter $D$ has a 
$(\dilation,\congestion)$-cycle cover where $\dilation=\widetilde{O}(D)$ and 
$\congestion = \widetilde{O}(1)$.
That is, the edges of $G$ can be covered by cycles such that each cycle is of 
length at most $\widetilde{O}(D)$ and each edge participates in at most 
$\widetilde{O}(1)$ cycles. 
\end{theorem}

\Cref{thm:cyclecover_upper} is \emph{existentially} optimal up to 
poly-logarithmic factors, e.g., the cycle graph. 
We also study cycle covers that are {\em universally-optimal} with respect to the input 
graph $G$ (up to log-factors). By using neighborhood covers \cite{awerbuch1998near}, we show how to convert 
the existentially optimal construction into a universally optimal one: 
\begin{theorem}[Optimal Cycle Cover, Informal]\label{thm:cyclecover_opt}
	There exists a construction of (nearly) universally optimal 
	$(\dilation,\congestion)$-cycle covers with 
	$\dilation=\widetilde{O}(\OPT(G))$ and $\congestion=\widetilde{O}(1)$, where
	$\OPT(G)$ is the best possible cycle length (i.e., even 
	without the congestion constraint).
\end{theorem}
In fact, our algorithm can be made nearly optimal with respect to each 
individual edge. That is, we can construct a cycle cover that covers each edge 
$e$ by a cycle whose length is $\widetilde{O}(|C_e|)$ where $C_e$ is the 
shortest cycle in $G$ that goes through $e$. The congestion for any edge remains
$\widetilde{O}(1)$. 

Turning to the distributed setting, we also provide a construction of cycle covers for the family of \emph{minor-closed} graphs. Our construction is (nearly) optimal in terms of both its run-time and in the parameters of the cycle cover. Minor-closed graphs have recently attracted a lot of attention in the setting of distributed network optimization \cite{ghaffari2016distributed,haeupler2016low,ghaffari2017near,haeupler2018minor,leviplanar18}. 
\begin{theorem}[Optimal Cycle Cover Construction for Minor Close Graphs, 
Informal]\label{thm:optimal-minor-closed}
For the family of minor closed graphs, 
there exists an $\widetilde{O}(OPT(G))$-round algorithm that constructs  
$(\dilation,\congestion)$-cycle cover with 
$\dilation=\widetilde{O}(OPT(G)), \congestion=\widetilde{O}(1)$, where
$\OPT(G)$ is equal to the best possible cycle length (i.e., even without
the constraint on the congestion).
\end{theorem}

\paragraph{Natural Generalizations of Low-Congestion Cycle Covers.}
Interestingly, our cycle cover constructions are quite flexible and naturally generalize to other related graph structures. For example, a $(\dilation,\congestion)$-\emph{two-edge-disjoint} cycle cover of a $3$-edge connected graphs is a collection of cycles such that each edge is covered by at least two edge disjoint cycles in $G$, each edge appears on at most $\congestion$ cycles, and each cycle is of length at most $\dilation$. In other words, such cycle cover provides $3$ edge disjoint paths between every neighboring nodes, these paths are short and nearly edge disjoint.

As we will describe next, we use this notation of two-edge-disjoint cycle covers in the context of fault tolerant algorithms. Towards this end, we show:
\begin{theorem}\label{thm:twoedge}[Two-edge-Disjoint Cycle Covers, 
Informal]\label{lem:twoedgecyclecover}
Every $3$-edge connected $n$-vertex graph with diameter $D$ has a $(\dilation,\congestion)$-\emph{two-edge-disjoint} cycle cover with $\dilation=\widetilde{O}(D^3)$ and $\congestion=\widetilde{O}(D^2)$. 
\end{theorem}
It is also quite straightforward to adapt the construction of Theorem 
\ref{thm:twoedge} to yield $k$-edge disjoint covers which cover every edge by 
$k$ edge disjoint cycles. These variants are also related to the notions of 
length-bounded cuts and flows \cite{baier2010length}, and find applications in 
fault tolerant computation. 

Cycle covers can be extended even further, one interesting example is ``$P_k$ covers'', where it is required to cover all paths of length at most $k$ in $G$ by simple cycles. The cost of such an extension has an overhead of $O((\Delta \cdot D)^k)$ in the dilation and congestion, where $\Delta$ is the maximum degree in $G$. These variant might find applications in secure computation.

Finally, in a companion work~\cite{parteryogev17} cycle cover are used to 
construct a new graph structure called \emph{private neighborhood trees} which 
serve the basis of a compiler for secure distributed algorithms.

\paragraph{Low-Congestion Covers as a Backbone in Distributed Algorithms.} Many of 
the underlying graph algorithms in the \congest\ model are based (either 
directly or indirectly) on low-congestion communication backbones. Ghaffari and 
Haeupler introduced the notion of low-congestion shortcuts for planar graphs 
\cite{ghaffari2016distributed}. These shortcuts have been shown to be useful 
for a wide range of problems, including MST, Min-Cut \cite{haeupler2018round}, 
shortest path computation \cite{haeupler2018faster} and other problems 
\cite{ghaffari2017near,li2018distributed}. Low congestion shortcuts have been 
studied also for bounded genus graphs \cite{haeupler2016low}, bounded treewidth 
graphs \cite{haeupler2016near} and recently also for general graphs 
\cite{haeupler2018round}. Ghaffari considered \emph{shallow-tree packing}, 
collection of small depth and nearly edge disjoint trees for the purpose 
of distributed broadcast \cite{ghaffari2015distributed}. 

Our low-congestion cycle covers join this wide family of low-congestion covers -- the graph theoretical infrastructures that underlay efficient algorithms in the \congest\ model. It is noteworthy that our cycle cover constructions are based on novel and independent ideas and are technically \emph{not} related to any of the existing low congestion graph structures.

\subsection{Distributed Compiler for Resilient Computation}
Our motivation for defining low-congestion cycle covers is rooted in the 
setting of distributed computation in a faulty or non-trusted environment. 
In this work, we consider two types of malicious adversaries, a byzantine adversary that can corrupt messages, and an eavesdropper adversary that listens on graph edges and show how to compile an algorithm to be resilient to such adversaries.

We present a new general framework for resilient computation in 
the \congest\ model \cite{Peleg:2000} of distributed computing. In this model, 
execution proceeds in synchronous rounds and in each round, each node can send a message of size $O(\log n)$ to each of its neighbors. 

The low-congestion cycle covers give raise to a 
\emph{simulation} methodology (in the spirit of synchronizes 
\cite{awerbuch1985complexity}) that can take any $r$-round distributed 
algorithm $\cA$ and compile it into a \emph{resilient} one, while incurring 
a blowup in the round complexity as a function of network's diameter.  
In the high-level, omitting many technicalities, our applications use the fact that the cycle cover provides each edge $e$ \emph{two}-edge-disjoint paths: a direct one using the edge $e$ and an indirect one using the cycle $C_e$ that covers $e$. Our low-congestion covers allows one to send information on all cycles in essentially the same round complexity as sending a message on a \emph{single} cycle.

\paragraph{Compiler for Byzantine Adversary.}
Fault tolerant computation \cite{ben1988completeness,gartner1999fundamentals} 
concerns with the efficient information exchange in communication networks 
whose nodes or edges are subject to Byzantine faults. 
The three cornerstone problems in the area of fault tolerant computation are: 
\emph{consensus} 
\cite{DworkPPU88,fischer1983consensus,fischer1985impossibility,keidar2001cost,leblanc2013resilient},
\emph{broadcasting} (i.e., one to all) 
\cite{peleg1989time,pelc1996fault,berman1997reliable,kranakis2001fault,pelc2005broadcasting}
 and \emph{gossiping} (i.e., all to all) 
\cite{blough1993optimal,bagchi1994information,censor2017fast}. A plentiful list 
of fault tolerant algorithms have been devised for these problems and various 
fault and communication models have been considered, see \cite{pelc1996fault} 
for a survey on this topic. 

In the area of interactive coding, the common model considers an adversary that can corrupt at most a fraction, known as \emph{error rate}, of the messages sent throughout the entire protocol. Hoza and Schulman \cite{hoza2016adversarial} showed a general compiler for synchronous distributed algorithms that handles an  adversarial error rate of $O(1/|E|)$ while incurring a constant communication overhead.
\cite{censor2018making} extended this result for the asynchronous setting, see \cite{Gelles17} for additional error models in interactive coding. 

In our applications, we consider a Byzantine adversary that can corrupt a \emph{single} message in each round, regardless of the number of messages sent over all. This is different, and in some sense incomparable, to the adversarial model in interactive coding, where the adversary is limited to corrupt only a bounded \emph{fraction} of all message. On the one hand, the latter adversary is stronger than ours as it allows to corrupt potentially many messages in a given round. On the other hand, in the case where the original protocol sends a linear number (in the number of vertices) of messages in a given round, our adversary is \emph{stronger} as the interactive coding adversary which handles only error rate of $O(1/n)$ cannot corrupt a single edge in each and every round. 
As will be elaborated more in the technical sections, this 
adversarial setting calls for the stronger variant of cycle covers in which 
each edge is covered by \emph{two} edge-disjoint cycles (as discussed in 
Theorem \ref{thm:twoedge}). 
\begin{theorem}\label{thm:bcomp}
	(Compiler for Byzantine Adversary, Informal)
Assume that a $(\dilation_1, \congestion_1)$ cycle cover and a $(\dilation_2,\congestion_2)$ two-edge-disjoint cycle cover are computed in a (fault-free) preprocessing phase. 
Then any  distributed algorithm $\cA$ can be compiled into an equivalent algorithm $\cA'$ that is \emph{resilient} to a Byzantine adversary while incurring an overhead of $\widetilde{O}((\congestion_1+\dilation_1)^2\cdot\dilation_2)$ in the number of rounds.
\end{theorem}

\paragraph{Compiler Against Eavesdropping.}
Our second application considers an eavesdropper adversary that in each round can listen on one of the graph edges of his choice. The goal is to take an algorithm $\cA$ and compile it to an equivalent algorithm $\cA'$ with the guarantee that the adversary learns nothing (in the information theoretic sense) regarding the messages of $\cA$. 
This application perfectly fits the cycle cover infrastructure. We show: 

\begin{theorem}[Compiler for Eavesdropping, Informal]\label{thm:ecomp}
Assume that $(\dilation, \congestion)$-cycle cover is computed in a preprocessing phase. Then any  distributed algorithm $\cA$ can be compiled into an $\cA'$ algorithm that is \emph{resilient} to an eavesdropping adversary while incurring an overhead of $\widetilde{O}(\dilation+\congestion)$ in the number of rounds.
\end{theorem}

In a companion work~\cite{parteryogev17}, low-congestion cycle covers are used 
to build up a more massive infrastructure that provides much stronger security 
guarantees. In the setting of \cite{parteryogev17}, the adversary takes 
over a single \emph{node} in the network and the goal is for all nodes to learn 
nothing on inputs and outputs of other nodes. This calls for combining the 
graph theory with cryptographic tools to get a compiler that is both efficient 
and secure.

\paragraph{Our Focus.} We note that the main focus in this paper is to study 
low-congestion cycle covers from an algorithmic and combinatorial perspective, as well as to demonstrate their applications for resilient computation. 
In these distributed applications, it is assumed that the cycle covers are constructed in a \emph{preprocessing} phase and are given in a distributed manner (e.g., each edge $e$ knows the cycles that go through it). Such preprocessing should be done only once per graph.

Though our focus is not in the distributed implementation of constructing these cycle covers, we do address this setting to some extent by: (1) providing a preprocessing algorithm with $\widetilde{O}(n)$ rounds that constructs the covers for general graphs; (2) providing a (nearly) optimal construction for the family of minor closed graphs. A sublinear distributed construction of cycle covers for general graphs requires considerably extra work and appears in a follow-up work \cite{parteryogev17}. 


\subsection{Preliminaries}\label{sec:prelims}
\paragraph{Graph Notations.} For a rooted tree $T \subseteq G$, and $z \in V$, let $T(z)$ be the 
subtree of $T$ rooted at $z$, and let $\pi(u,v,T)$ be the tree path between $u$ 
and $v$, when $T$ is clear from the context, we may omit it and simply write 
$\pi(u,v)$. For a vertex $u$, let $p(u)$ be the parent of $u$ in $T$. Let $P_1$ be a $u$-$v$ path (possibly $u=v$) and $P_2$ be a $v$-$z$ path, we denote by $P_1 \circ P_2$ to be the concatenation of the two paths. 

The fundamental cycle $C$ of an edge $e=(u,v) \notin T$ is the cycle formed by taking $e$ and the tree path between $u$ and $v$ in $T_0$, i.e., $C=e \circ \pi(u,v,T)$.
For $u,v \in G$, let $\dist(u,v,G)$ be the length (in edges) of the shortest $u-v$ path in $G$.

For every integer $i\geq 1$, let $\Gamma_i(u,G)=\{v ~\mid~ \dist_G(u,v)\leq i\}$. When $i=1$, we simply write $\Gamma(u,G)$.
Let $\deg(u,G)=|\Gamma(u,G)|$ be the degree of $u$ in $G$. For a subset of edges 
$E' \subseteq E(G)$, let $\deg(u,E')=|\{v : (u,v) \in E'\}|$ be the number 
of edges incident to $u$ in $E'$.
For a subset of nodes $U$, let $\deg(U,E')=\sum_{u \in U}\deg(u,E')$. For a 
subset of vertices $S_i\subseteq V(G)$, let $G[S_i]$ be the induced subgraph on 
$S_i$.
\begin{fact}\label{fc:girthmoore}[Moore Bound, \cite{bollobas2004extremal}]
Every $n$-vertex graph $G=(V,E)$ with at least $2n^{1+1/k}$ edges has a cycle 
of length at most $2k$.
\end{fact}
%
%
\paragraph{The Communication Model.}\label{def:com-model}
We use a standard message passing model, the
\congest\ model \cite{Peleg:2000}, where the execution proceeds in synchronous 
rounds and in each round, each node can send a message of size $O(\log n)$ to 
each of its neighbors.
In this model, local computation at each node is for free and the primary 
complexity measure is the number of communication rounds. Each node holds a 
processor with a unique and arbitrary ID of $O(\log n)$ bits.


\begin{definition}[Secret Sharing]\label{def:secret-sharing}
Let $x \in \bit^n$ be a message. The message $x$ is {\em secret shared} to $k$ 
shares by choosing $k$ random strings $x^1,\ldots,x^k \in \bit^n$ conditioned 
on $x = \bigoplus_{j=1}^{k}x^j$. Each $x^j$ is called a share, and notice that 
the joint distribution of any $k-1$ shares is uniform over $(\bit^n)^{k-1}$.
\end{definition}
\section{Technical Overview}
 \subsection{Low Congestion Cycle Covers}\label{sec:cyclehighlevel}

We next give an overview of the construction of low congestion cycle covers of
\Cref{thm:cyclecover_upper}. The proof proof appears in \Cref{sec:cyclecoverall}.

Let $G=(V,E)$ be a 
bridgeless $n$-vertex graph with diameter $D$. 
Our approach is based on constructing a BFS tree $T$ rooted at an arbitrary vertex in the 
graph $G$ and covering the edges by two procedures: the first 
constructs a low congestion cycle cover for the \emph{non}-tree edges and the 
second covers the tree edges.

\paragraph{Covering the Non-Tree Edges.}
Let $E'=E \setminus E(T)$ be the set of non-tree edges. Since the diameter\footnote{The graph $G\setminus T$ might 
be disconnected, when referring to its diameter, we refer to the maximum 
diameter in each connected component of $G\setminus T$.} of $G \setminus T$ 
might be large (\eg 
$\Omega(n)$), to cover the edges of $E'$ by \emph{short} cycles (i.e., of length $O(D)$), one 
must use the 
edges of $T$. A na\"ive approach is to cover 
every 
edge $e=(u,v)$ in 
$E'$ by taking its fundamental cycle in $T$ 
(\ie using the $u$-$v$ path in $T$). Although this yields short cycles, the 
congestion on the tree edges might become $\Omega(n)$. The key 
challenge is to use the edges of $T$ (as we indeed have to) in a 
way that the output cycles would be short without overloading any 
tree edge more than $\widetilde{O}(1)$ times.

%

Our approach is based on using the edges of the tree $T$ only for the purpose of 
connecting
nodes that are somewhat \emph{close} to each other (under some definition of closeness to be described later), in a way that would balance the overload on each tree edge. To realize this approach, we 
define a 
specific way of partitioning the nodes of the tree $T$ 
into \emph{blocks} according to $E'$. In a very rough manner, a block consists of a 
set of nodes that have few incident edges in $E'$.
To define these blocks, we number the nodes based 
on post-order traversal in $T$ and partition them into blocks containing 
nodes with consecutive numbering. The {\em density} of a block $B$ is the 
number 
of edges in $E'$ with an endpoint in $B$. Letting $\densitythreshold$ be some threshold of constant value on the density, the blocks are partitioned such that 
every block is either (1) a singleton block consisting of one node with at least $\densitythreshold$ edges in $E'$ or (2) consists of at least two nodes but has a density bounded by $2\densitythreshold$. 
As a result, the number of 
blocks is not too large 
(say, at most $|E'|/8$).  

To cover the edges of $E'$ by cycles, the algorithm 
considers the contracted graph obtained by contracting all nodes in a 
given block into one supernode and connecting two supernodes $B_1$ and 
$B_2$, if there is an edge in $E'$ whose one endpoint is in $B_1$, and the 
other endpoint is in $B_2$. This graph is in fact a multigraph as it might 
contain self-loops or multi-edges. 
We now use the fact that any $x$-vertex graph with at least $2x$ edges has 
girth $O(\log x)$. 
Since the contracted graph 
contains at most $x=|E'|/8$ nodes and has $|E'|$ edges, its girth is $O(\log n)$. The algorithm
then repeatedly finds (edge-disjoint) short cycles (of length 
$O(\log n)$) in this contracted graph\footnote{That is, it computes a short cycle $C$, omit the edges of $C$ from the contracted graph and repeat.}, until we are left with at most $|E'|/4$ edges. The cycles 
computed in the contracted graph are then translated 
to cycles in the original graph $G$ by using 
the tree paths $\pi(u,v,T)$ between nodes $u,v$ belonging to the same supernode 
(block). We note that this translation might result in cycles that are non-simple, and this is handled later on.

Our key insight is that even though the tree paths connecting two nodes in a given block might be 
long, i.e., of length $\Omega(D)$, we show that every tree edge is ``used'' by at most \emph{two} blocks. 
That is, for each 
edge $e$ of the tree, there are at most 2 blocks such that the tree path 
$\pi(u,v,T)$ of nodes $u,v$ in the block passes through $e$. (If a block has 
only a single node, then it will use no tree edges.) Since the (non-singleton) blocks have 
constant density, we are able to bound the congestion on each tree edge $e$. The 
translation of cycles in the contracted graph to cycles in the original graph
yields $O(D\log n)$-length cycles in the original graph where every edge 
belongs to $O(1)$ cycles.

The above step already covers $1/4$ of the edges in $E'$. We continue 
this process for $\log n$ times until all edges of $E'$ are covered, and thus get 
a $\log n$ factor in the congestion. 

Finally, to make the output cycle simple, we have an additional ``cleanup" 
step  (procedure 
$\SimplifyCycles$) which takes the output collection of non-simple cycles and 
produces a collection of simple ones. 
In this process, some of the edges in the non-simple cycles might be omitted, 
however, we prove that only tree edges might get omitted and all non-tree edges 
remain covered by the simple cycles. This concludes the high level idea of 
covering the non-tree edges.
We note the our blocking definition is quite useful also for distributed 
implementations. The reason is that although the blocks are not independent, in 
the sense that the tree path connecting two nodes in a given block pass through 
other blocks, this independence is very \emph{limited}. The fact that each tree 
edge is used in the tree paths for only \emph{two} blocks allows us also to 
work distributively on all blocks simultaneously (see 
\Cref{sec:distributedconstcovers}). 

\paragraph{Covering the Tree Edges.}
Covering the tree edges turns out to be the harder case where new ideas are 
required. 
Specifically, whereas for the non-tree edges our goal is to find cycles that use 
the tree edge 
as rarely as possible, here we aim to find cycles that cover all tree edges, but still avoid using a particular tree edge in too many 
cycles.

The construction is based on the notion of swap edges.
For every tree edge $e \in T$, define the \emph{swap} edge 
of $e$ by $e'=\swap(e)$ to be an arbitrary edge in $G$ that restores 
the connectivity of $T \setminus \{e\}$. Since the graph $G$ is $2$-edge 
connected such an edge $\swap(e)$ is guaranteed to exist for every $e \in 
T$. Let $e=(u,v)$ (i.e., $u=p(v)$) and $(u',v')=\swap(e)$. Let $s(v)$ be the endpoint 
of $\swap(e)$ that do not belong to $T(u)$ (i.e., the subtree $T$ rooted at 
$u$), thus $v'=s(v)$. 

The algorithm for covering the tree edges is recursive, where in each 
step we split the tree into two edge disjoint subtrees $T_1,T_2$ that are 
balanced in terms of number of edges. To perform a recursive step, we would 
like to break the problem into two independent subproblems, one that covers the 
edges of $T_1$ and the other that covers the edges of $T_2$. However, observe 
that there might be edges 
$(u,v) \in T_1$ where the only cycle that covers them\footnote{Recall that the 
graph $G$ is two edge connected.} passes through $T_2$ (and vice versa). 

Specifically, we will consider all tree edges $(u,v)\in T_1$, whose second endpoint $s(v)$ of their swap edge is in $T_2$. To cover these tree edges, we employ two procedures, one on $T_1$ and the other 
on $T_2$ that together form the desired cycles (for an illustration, see 
\Cref{fig:swapedge,fig:treecoverfigmany}). First, we mark all nodes $v \in T_1$ 
such that their $s(v)$ is in $T_2$. Then, 
we use an Algorithm called $\EdgeDisjointPath$ (\cite{klein1995nearly} and Lemma 
4.3.2 \cite{Peleg:2000}) which solves the following problem: given a rooted 
tree 
$T$ and a set of $2k$ 
marked nodes $M \subseteq V(T)$ for $k\leq n/2$, find a matching 
of these vertices $\langle u_i,u_j \rangle$ 
into pairs such that the tree paths $\pi(u_i,u_j,T)$ connecting the matched 
pairs are \emph{edge-disjoint}.

We employ Algorithm $\EdgeDisjointPath$ on $T_1$ with the marked nodes as 
described above. Then 
for every pair $u_i,u_j \in T_1$ that got matched by Algorithm 
$\EdgeDisjointPath$, we add a virtual edge between $s(u_i)$ and 
$s(u_j)$ in $T_2$. Since this virtual edge is a non-tree edge with both 
endpoints in $T_2$, we have translated the 
dependency between $T_1$ and $T_2$ to covering a \emph{non-tree} edge in $T_2$. At this 
point, we can simply use Algorithm $\NonTreeCover$ on the tree $T_2$ and the 
non-virtual edges. This computes a cycle collection which covers all virtual 
edges $(s(u_i),s(u_j))$. In the final step, we replace each virtual edge 
$(s(u_i),s(u_j))$ with an $s(u_i)$-$s(u_j)$ path that consists of the tree path $\pi(u_i,u_j,T_1)$, and the 
paths between $u_i$ and $s(u_i)$ (as well as the path connecting $u_j$ 
and $s(u_j)$). 

This above description is simplified and avoids many details and complications 
that we had to address in the full algorithm. For instance, in our algorithm, a 
given tree edge might be responsible for the covering of up to $\Theta(D)$ many 
tree edges. This prevents us from using the edge disjoint paths of Algorithm 
$\EdgeDisjointPath$ in a na\"ive manner. In particular, our algorithm has to 
avoid the multiple appearance of a given tree edge on the same cycle as in such 
a case, when making the cycle simple that tree edge might get omitted and will 
no longer be covered. 
See \Cref{sec:cyclecoverall} for the precise 
details of the proof, and see 
\Cref{fig:main-algorithm-non-tree-intro} for a 
summary of our algorithm.
\begin{figure}[!h]
	\begin{boxedminipage}{\textwidth}
\vspace{3mm} \textbf{Algorithm $\GraphCover(G=(V,E))$}
\begin{enumerate}
	\item Construct a BFS tree $T$ of $G$.
	\item Let $E'=E \setminus E(T)$ be the set of
	non-tree
		edges.
	\item Repeat $O(\log n)$ times:
	\begin{enumerate}
		\item Partition the nodes of $T$ with block density $\densitythreshold$ 
		with respect to (uncovered edges) $E'$.
		\item While there are $t$ edges  
		$(u_1,v_1),\ldots,(u_t,v_t) \in \widetilde{E}$ for $t \le \log n$ such 
		that for all $i \in [t-1]$, $v_i$ and $u_{i+1}$ are in the same block 
		and $v_t$ and $u_1$ are in the same block (with respect to the 
		partitioning $\mathcal{B}$):
		\begin{itemize}
			\item Add the cycle 
			$(u_1,v_1) \circ \pi(v_1,u_2) \circ (u_2,v_2) \circ \pi(v_2,u_3) \circ (u_3,v_3) \circ \dots \circ(v_t,u_1)$ to $\cC$.
					\item Remove the covered edges from $E'$.
		\end{itemize}
	\end{enumerate}
	\item $\cC \gets \cC \cup\TreeCover(T)$ (see 
	\Cref{fig:main-algorithm-tree-edges}).
	\item Output $\SimplifyCycles(\cC)$.
\end{enumerate}
	\end{boxedminipage}
	\caption{Procedure for constructing low-congestion covers.}
	\label{fig:main-algorithm-non-tree-intro}
\end{figure}

\subsection{Universally Optimal Cycle Covers}\label{sec:optcovers}
In this section we describe how to transform the construction of \Cref{sec:cyclehighlevel} into an 
universally optimal construction: covering each edge $e$ in $G$ by almost the shortest possible cycle while having almost no overlap between cycles.
Let $C_e$ be the shortest cycle covering $e$ in $G$ and $\OptCCVal=\max_e{|C_e|}$. 
Clearly, there are graphs with diameter $D=\Omega(n)$ and $\OptCCVal=O(1)$. We 
show:
\begingroup
\def\thetheorem{\ref{thm:cyclecover_opt}}
\begin{theorem}[Rephrased]\label{thm:cyclcoveropt}
For any bridgeless graph $G$, one can construct an 
$(\widetilde{O}(\OptCCVal),\widetilde{O}(1))$ cycle cover $\mathcal{C}$. Also, each edge $e \in G$ has a cycle $C'_e$ in 
$\mathcal{C}$ containing $e$ such that $|C'_e|=\widetilde{O}(|C_e|)$. 
\end{theorem}
\addtocounter{theorem}{-1}
\endgroup

We will use the fact that our cycle cover algorithm of \Cref{sec:cyclehighlevel} does not require $G$ to be bridgeless, but rather covers every edge $e$ that appears on \emph{some} cycle in $G$. We call such cycle cover algorithm \emph{nice}.  

Our approach is based on the notion of neighborhood covers (also known as \emph{ball carving}). 
The $t$-\emph{neighborhood cover} \cite{awerbuch1996fast} of the graph $G$ is a 
collection of clusters $\mathcal{N}=\{S_1,\ldots, S_r\}$ in the graph such that (i) every vertex $v$ has a cluster that contains its entire $t$-neighborhood, (ii) the diameter of $G[S_i]$ is $O(t \cdot \log n)$ and (iii) every vertex belongs to $O(\log n)$ clusters in $\mathcal{N}$.
The key observation is that if each edge appears on a cycle of length at most $\OptCCVal$, then there must be a (small diameter) subgraph $G[S_i]$ that fully contains this cycle. 

The algorithm starts by computing an $\OptCCVal$-neighborhood cover which decomposes $G$ into almost-disjoint subgraph $G[S_1],\ldots, G[S_r]$, each with diameter $\widetilde{O}(\OptCCVal)$. 
Next, a $(\dilation_i,\congestion_i)$ cycle cover $\mathcal{C}_i$ is constructed in each subgraph $G_i$ by applying algorithm $\GraphCover$ of \Cref{sec:cyclehighlevel} where $\dilation_i=\widetilde{O}(\Diam(G[S_i])$ and $\congestion_i=\widetilde{O}(1)$.
The final cycle cover $\mathcal{C}$ is the union of all these covers $\mathcal{C}=\bigcup_{i}\mathcal{C}_i$.
Since $\Diam(G[S_i])=\widetilde{O}(\OptCCVal)$, the length of all cycles is $\widetilde{O}(\OptCCVal)$. Turning to congestion, since each vertex appears on $O(\log n)$ many subgraphs, taking the union of all cycles increases the total congestion by only $O(\log n)$ factor. Finally it remains to show that all edges are covered. Since each edge $e$ appears on a cycle $C_e$ in $G$ of length at most $\OptCCVal$, there exists a  cluster, say $S_i$ that contains all the vertices of $C_e$. We have that $e$ appears on a cycle in $G[S_i]$ and hence it is covered by the cycles of $\mathcal{C}_i$.  
To provide a cycle cover that is almost-optimal with respect each edge, we 
repeat the above procedures for $O(\log \OptCCVal)$ many times, in the $i^{th}$ 
application, the algorithm constructs $2^i$-neighborhood cover, applies Alg. 
$\cA$ in each of the resulting clusters and by that covers all edges $e$ with $|C_e|\leq 2^i$. The detailed analysis and pseudocodes is in \Cref{sec:fulloptcovers}. 

\subsection{Application to Resilient Distributed Computation}
Our study of low congestion cycle cover is motivated by applications in distributed computing. We given an overview of our two applications to resilient distributed computation that uses the framework of our cycle cover. Both applications are compilers for distributed algorithms in the standard \congest\ model. In this model, each node can send a message of size $O(\log n)$ to 
each of its neighbors in each rounds (the full definition of the model appears in \Cref{def:com-model}). The full details of the compilers appear in \Cref{sec:applications}.

\paragraph{Byzantine Faults.}
In this setting, there is an adversary that can maliciously modify messages sent over the edges of the graph. The adversary is allowed to do the following.
In each round, he picks a single message $M_e$ passed on the edge $e \in G$ and corrupts it in an arbitrary manner (\ie modifying the sent message, or even completely dropping the message). The recipient of the corrupted message is not notified of the corruption. The adversary is assumed to know the inputs to all the nodes, and the entire history of communications up to the present. 
It then picks which edge to corrupt adaptively using this information.

Our goal is to compile any distributed algorithm $\cA$ into an resilient one $\cA'$ while incurring a small blowup in the number of rounds. The compiled algorithm $\cA'$ has the exact same output as $\cA$ for all nodes even in the presence of such an adversary.
Our compiler assumes a preprocessing phase of the graph, which is \emph{fault-free}, in which the cycle covers are computed. The preprocessing phase computes a $(\dilation_1,\congestion_1)$-cycle covers and a $(\dilation_2,\congestion_2)$-two-edge disjoint  variant using \Cref{lem:twoedgecyclecover} (see \Cref{app:threeedgedisjoint} for details regarding two-edge disjoints cycle cover).

For the simplicity of this overview, we give a description of our compiler assuming that the bandwidth on each edge is $\widetilde{O}(\congestion_2)$. This is the basis for the final compiler that uses the standard bandwidth of $O(\log n)$. We note that this last modification is straightforward in a model without an adversary, e.g., by blowing up the round complexity by a factor of $\widetilde{O}(\congestion_2)$, or by using more efficient scheduling techniques such as \cite{leighton1994packet,Ghaffari15}. However, such transformations fail in the presence of the adversary since two messages that are sent in the same round might be sent in different rounds after this transformation. This allows the adversary to modify both of the messages -- which could not be obtained before the transformations, i.e., in the large bandwidth protocol.

The key idea is to use the three edge-disjoint, low-congestion paths between any neighboring pairs $u$ and $v$ provided by the $(\dilation_2,\congestion_2)$ two-edge disjoint cycle covers. Let $\ell=4\dilation_2$, where $\dilation_2$ is an upper bound on the length these paths. Consider round $i$ of algorithm $\cA$. For every edge $e=(u,v)$, let $M_e$ be the message that $u$ sends to $v$ in round $i$ of algorithm $\cA$. 
Each of these messages $M_e$ is going to be sent using $\ell$ rounds, on the three edge-disjoint routes. The messages will be sent repeatedly on the edge disjoint paths, in a pipeline manner, throughout the $\ell$ rounds. That is, in each of the $\ell$ rounds, node $u$ repeatedly sends the message $M_e$ along the three edge disjoint paths to $v$. Each intermediate node forwards any message received on a path to its successor on that path. The endpoint $v$ recovers the message $M_e$ by taking the majority of the received messages in these $\ell$ rounds. Let $a_1 \leq a_2$ be the lengths of the two edge-disjoint paths connecting $u$ and $v$ (in addition to the edge $(u,v)$). We prove that the fraction of uncorrupted messages received by $v$ is at least
$$
\frac{2\ell-a_1-a_2}{3\ell-a_1-a_2} > \frac{6 \dilation_2}{12 \dilation_2-3} > 1/2.
$$
Thus, regardless of the adversary's strategy, the majority of the messages received by $v$ are correct, allowing $v$ to recover the message.

Our final compiler that works in the \congest\ model with bandwidth $O(\log n)$ is more complex. 
As explained above, using scheduling to reduce congestion might be risky. Our approach compiles each round of algorithm $\cA$ in two phases. The first phase uses the standard $(\dilation_1,\congestion_1)$ cycle cover to reduce the number of ``risky receivers" from $n$ down to $O(\dilation_1+\congestion_1)$. 
The second phase restricts attention to these remaining messages which will be re-send along the three edge-disjoint paths in a similar manner to the description above. The fact that we do not know in advance which messages will be handled in the second phase, poses some obstacles and calls for a very careful scheduling scheme. See \Cref{sec:applications} for the detailed compiler and its analysis.

\paragraph{Eavesdropping.}
In this setting, an adversary eavesdrops on an single (adversarially chosen) edge in each round. The goal is to prevent the adversary from \emph{learning} anything, in the information-theoretic sense, on any of the messages sent throughout the protocol. Here we use the two edge disjoint paths, between neighbors, that the cycle cover provides us in a different way. Instead of repeating the message, we ``secret share'' it.

Consider an edge $(u,v)$ and let $M$ be the message sent on $e$.
The sender $u$ secret shares\footnote{We say $M \in \bit^m$ is {\em secret shared} to $k$ shares by choosing $k$ random strings $M^1,\ldots,M^k \in \bit^m$ conditioned on $M = \bigoplus_{j=1}^{k}M^j$. Each $M^j$ is called a share, and notice that the joint distribution of any $k-1$ shares is uniform over $(\bit^m)^{k-1}$ and thus provides no information on the message $M$.} the message $M$ to $d+1$ random shares $M_1,\ldots,M_{d+1}$ such that $M_1 \oplus \dots \oplus M_{d+1} = M$. The first $\ell$ shares of the message, namely $M_1,\ldots, M_{\ell}$, will be sent on the direct $(u,v)$ edge, in each of the rounds of phase $i$, and the $(d+1)^{th}$ share is sent via the $u$-$v$ path $C_e \setminus \{e\}$. At the end of these $d$ rounds, $v$ receives $d+1$ messages. Since the adversary can learn at most $d$ shares out of the $d+1$ shares, we know that he did not learn anything (in the information-theoretic sense) about the message $M$.
See the full details in \Cref{sec:applications}.

\subsection{Distributed Algorithm for Minor-Closed Graphs}\label{sec:minorclose}


We next turn to consider the distributed construction of low-congestion covers for the family of minor-closed graphs. We will highlight here the main ideas for constructing $(\dilation,\congestion)$ 
cycle covers with
$\dilation=\widetilde{O}(D)$ and $\congestion=\widetilde{O}(1)$ within $r=\widetilde{O}(D)$ rounds. Similarly to \Cref{sec:optcovers}, applying the below construction in each component of the neighborhood cover,  yields a nearly optimal cycle cover with $\dilation=\widetilde{O}(\OptCCVal)$ and $\congestion=\widetilde{O}(1)$, where $\OptCCVal$ is the best dilation of any cycle cover in $G$, regardless of the congestion constraint. 

The distributed output of the cycle cover construction is as follows: each edge $e$ knows the edge IDs of all the edges that are covered by cycles that pass through $e$. 
Let $|E(G)|\leq c \cdot n$ for the universal constant $c$ of the minor closed family of $G$ (see \Cref{fc:minorsparse}). 

The algorithm begins by constructing a BFS tree $T \subseteq G$ in $O(D)$ 
rounds. Here we focus on the covering procedure of the non-tree edges. 
Covering the tree edges is done by a reduction to the non-tree just like in the centralized construction.

The algorithm consists of $O(\log n)$ phases, each takes $O(D)$ rounds. In each phase $i$, we are given a subset $E' \subseteq E \setminus E(T)$ that remains to be covered and the algorithm constructs a cycle cover $\mathcal{C}_i$, that is shown to cover most of the $E'$ edges, as follows:

\paragraph{Step (S1): Tree Decomposition into Subtree Blocks.}
The tree is decomposed into vertex disjoint subtrees, which we call \emph{blocks}. These blocks have different properties compared to those of the algorithm in \Cref{sec:cyclehighlevel}.
The density of a block is 
the number of edges in $E'$ that are incident to nodes of the block. Ideally, we would want the densities of the blocks to be bounded by $\densitythreshold=16 \cdot c$. 
Unfortunately, this cannot be achieved while requiring the blocks to be vertex disjoint subtrees.
Our blocks might have an arbitrarily large density, and this would be handled in the analysis.
  
The tree decomposition works layer by layer from the bottom of the tree 
up the root. The weight of a node $v$, $W(v)$, is 
the number of uncovered non-tree edges incident to $v$. Each node 
$v$ of the layer sends to its parent in $T$, the \emph{residual weight} of its subtree, 
namely, the total weight of all the vertices in its subtree $T(v)$ that are not 
yet assigned to blocks.
A parent $u$ that receives the residual weight from its children does the 
following. Let $W'(u)$ be the sum of the total residual weight of its children 
plus $W(u)$. If 
$W'(u)\geq \densitythreshold$, then $u$ declares a 
block and down-casts its ID to all relevant descendants in its 
subtree (this ID serves as the block-ID). Otherwise, it passes $W'(u)$ to 
its parent. 

\paragraph{Step (S2): Covering Half of the Edges.}
The algorithm constructs a cycle collection that
covers two types of $E'$-edges: (i) edges with both endpoints in the same block and (ii) 
pairs of edges $e_1,e_2 \in E'$ whose endpoints connect the same pair of blocks $B_1,B_2$.
That is, the edges in $E'$ that are not covered are those that connect vertices 
in blocks $B,B'$ and no other edge in $E'$ connects these pair of blocks.

The root of each block is responsible for 
computing these edges in its block, and to compute 
their cycles, as follows. All nodes exchange the block-ID with their 
neighbors. Then, each node sends to the root of its block the block IDs of its 
neighbors in $E'$. 
This allows each root to identity the relevant $E'$ edges incident to its block. The analysis shows that despite the fact that the density of the block might be 
large, this step can be done in $O(D)$ rounds.
Edges with both endpoints in the same block are covered by taking their 
fundamental cycle\footnote{The fundamental cycle of an edge $e=(u,v) \notin T$ is 
the cycle formed by taking $e$ and the $u$-$v$ path in $T$.}
into $\mathcal{C}_i$. 
For the second type, the root arbitrarily matches pairs of $E'$ edges that connect vertices in the same pair of blocks. For each matched pair of edges 
$e=(u,v),e'=(u',v')$ with endpoints in block $B_1$ and $B_2$, the cycle for covering these edges 
$C(e,e')$ defined by
$C(e,e')=\pi(u,u',T)\circ e' \circ 
\pi(v',v,T) \circ e$ (i.e., taking the tree paths in each block). Thus, the
cycles have length $O(D)$.
 (see \Cref{fig:minormatching} for an illustration).
This completes the description of phase $i$.

\paragraph{Covering Argument via Super-Graph.}
We show that most of the $E'$-edges belongs to the two types of edges covered by the algorithm. This 
statement does not hold for general graphs, and exploits the properties minor closed families.
Let $E''$ be the subset of $E'$ edges that are \emph{not} covered in phase $i$. 
We consider the super-graph of $T \cup E''$ obtained by contracting the tree 
edges inside each block. Since the blocks are vertex disjoint, the resulting 
super-graph has one super-node per block and the $E''$ edges connecting these 
super-nodes. By the properties of phase $i$, the super-graph does not contain 
multiple edges or self-loops. The reason is that every self-loop corresponds to 
an edge in $E'$ that connects two nodes inside one block. Multiple edges 
between 
two blocks correspond to two $E'$-edges that connect endpoints in the same pair 
of blocks. Both of these $E'$ edges are covered in phase $i$.
Since the density of each block with respect to $E'$ is at least $\densitythreshold$, the super-graph contains at most $n'=|E'|/(c \cdot 8)$ super-nodes and $|E''|$ edges. As the super-graph belongs to the family of minor-closed as well, we have that $|E''|\leq c\cdot n'$ edges, and thus $|E''|\leq |E'|/8$, as required. 
The key observation for bounding the congestion on the edges is:

\begin{observation}
Let $e=(x,y)$ be a tree edge (where $x$ is closer to the root) and let $B$ be the block of $x$ and $y$.
Letting $B_y=B \cap T(y)$, it holds that $\deg(B_y,E')\leq \densitythreshold$. 
\end{observation}

This observation essentially implies that blocks can be treated as if they have bounded densities, hence taking the tree-paths of blocks into the cycles keeps the congestion bounded.
The distributed algorithm for covering the tree edges essentially mimics the centralized construction of 
\Cref{sec:cyclecoverall}. For the computation of the swap edges distributively we will use the algorithm of Section 4.1 in \cite{ghaffari2016near}.
The full analysis of the algorithm as well as the prcoedure that covers the tree edges, appear in 
\Cref{app:optminor}.

\section{Low Congestion Cycle Cover}\label{sec:cyclecoverall}
We give the formal definition of a cycle cover and prove our main theorem 
regarding low-congestion cycle covers. 

\begin{definition}[Low-Congestion Cycle Cover]\label{def:cycle-cover}
For a given graph $G=(V,E)$, a $(\dilation,\congestion)$ low-congestion 
cycle cover $\cC$ of $G$ is a collection of cycles that cover all 
edges of $G$ such that each cycle $C \in \cC$ is of length at most 
$O(\dilation)$ and each edge appears in at most $O(\congestion)$ cycles in 
$\cC$. That is, for every $e \in E$ it holds that $ 1\le |\{C \in \mathcal{C} : 
e 
\in C\}| \le O(\congestion)$. 
\end{definition}

We also consider partial covers, that cover only a subset of edges $E'$. We 
say that a cycle cover $\mathcal{C}$ is a $(\dilation,\congestion)$ cycle cover 
for $E' \subseteq E$, if all cycles are of length at most $\dilation$, each edge of 
$E'$ appears in at least one of the cycles of $\mathcal{C}$, and no edge in 
$E(G)$ appears in more than $c$ cycles in $\mathcal{C}$. That is, in this 
restricted definition, the covering is with respect to the subset of edges 
$E'$, however, the congestion limitation is with respect to all graph edges.

The main contribution of this section is an existential result regarding cycle 
covers with low congestion. Namely, we show that {\em any graph} that is 2-edge 
connected has a cycle cover where each cycle is at most the diameter of the 
graph (up to $\log n$ factors) and each edge is covered by $O(\log n)$ 
cycles. Moreover, the proof is 
actually constructive, and yields a polynomial time algorithm that computes 
such a cycle cover.

\begingroup
\def\thetheorem{\ref{thm:cyclecover_upper}}
\begin{theorem}[Rephrased]
For every bridgeless $n$-vertex graph $G$ with diameter $D$,
there exists a $(\dilation,\congestion)$-cycle cover with $\dilation=O(D\log 
n)$ and $\congestion=O(\log^3 n)$.
\end{theorem}
\addtocounter{theorem}{-1}
\endgroup


The construction of a $(\dilation,\congestion)$-cycle cover $\cC$ starts by constructing a BFS tree $T$.
The algorithm has two sub-procedures: the first computes a cycle collection $\cC_1$ for covering the \emph{non-tree} edges $E_1=E(G) \setminus E(T)$, the second computes a cycle collection $\cC_2$ for covering the \emph{tree} edges $E_2=E(T)$. We describe each cover 
separately. The pseudo-code for the algorithm is given in 
\Cref{fig:main-algorithm}. The algorithm uses two procedures, 
$\NonTreeCover$ and $\TreeCover$ which are given in 
\Cref{sec:cover-non-tree} and \Cref{sec:cover-tree-edges} respectively.

\begin{figure}[!h]
	\begin{boxedminipage}{\textwidth}
	\vspace{3mm} \textbf{Algorithm $\GraphCover(G=(V,E))$}
	\begin{enumerate}
		\item Construct a BFS tree $T$ of $G$ (with respect to edge set $E$).
		\item Let $E_1=E(G)\setminus E(T)$ be all non-tree edges, and let $E_2=E(T)$ be all 
		tree edges.
		\item $\cC_1 \gets \NonTreeCover(T,E_1)$.
		\item $\cC_2 \gets \TreeCover(T,E_2)$
		\item Output $\cC_1 \cup \cC_2$.
	\end{enumerate}
	\end{boxedminipage}
	\caption{Centralized algorithm for finding a cycle cover of a graph $G$.}
	\label{fig:main-algorithm}
\end{figure}
%
\subsection{Covering Non-Tree Edges}\label{sec:cover-non-tree}
Covering the non-tree edge mainly uses the fact that while the graph has many edges, then the girth 
is small. Specifically, using Fact \ref{fc:girthmoore}, with $k=2\log{n}$ we get 
that the girth of a graph with at least $2n$ edges is at most $4\log{n}$. 
Hence, as long as that the graph has at least $2n$ edges, a cycle of length 
$4\log n$ can be found. We get that all but $2n$ edges in $G$ are covered by 
edge-disjoint cycles of length $O(\log n)$.

In this subsection, we show that the set of edges $E_1$, \ie the set of 
non-tree edges can be covered by a 
$(O(D\log n), \widetilde{O}(1))$-cycle cover denoted $\cC_1$. Actually, what we show is 
slightly more general: if the tree is of depth $D(T)$ the length of the cycles 
is at most $O(D(T)\log n)$. \Cref{lem:non-treecover} will be useful for 
covering 
the tree-edges as well in \Cref{sec:cover-tree-edges}.

\begin{lemma}\label{lem:non-treecover}
Let $G=(V,E)$ be a $n$-vertex graph, let $T \subseteq G$ be a spanning tree of depth 
$D(T)$. Then, there exists 
an $(O(D(T) \log n), O(\log n))$-cycle cover $\cC_1$ for the edges of $E(G)\setminus E(T)$.  
\end{lemma}

An additional useful property of the cover $\cC_1$ is that despite the fact that the length of the cycles in $\cC_1$ is $O(D\log n)$, each cycle is used to cover $O(\log n)$ edges.
\begin{lemma}\label{lem:propnontree}
Each cycle in $\cC_1$ is used to cover $O(\log n)$ edges in $E(G)\setminus E(T)$.
\end{lemma}

The rest of this subsection is devoted to the proof of 
\Cref{lem:non-treecover}. A key 
component in the proof is a 
partitioning of 
the nodes of the tree $T$ into \emph{blocks}. The partitioning is based on a 
numbering of the nodes from $1$ to $n$ and grouping nodes with consecutive 
numbers into blocks under certain restrictions. We define a numbering of the 
nodes 
$$N: V(T) \to \left [|V(T)| \right]$$
by traversing the nodes of the tree in post-order traversal. That is, we let 
$N(u)=i$ if $u$ is the $\ith{i}$ node traversed.
Using this mapping, we proceed to defining a 
partitioning of the nodes into blocks and show some of their useful properties. 

For a block $B$ of nodes and a subset of non-tree edges $E'\subseteq E_1$, the notation $\deg(B,E')$ is the number of edges in $E'$ that 
have an endpoint in the set $B$ (counting multiplicities). We call this the \emph{density} of block $B$ 
with respect to $E'$. 
For a subset of edges $E'$, and a density bound $\densitythreshold$ (which will 
be set to a constant), 
an $(E',\densitythreshold)$-partitioning $\mathcal{B}$ is a partitioning of the 
nodes of the graph into blocks that satisfies the following properties: 
\begin{enumerate}\label{blockdef}
	\item Every block consists of a consecutive subset of nodes (w.r.t.\ their 
	$N(\cdot)$ numbering).
	\item If $\deg(B,E') > \densitythreshold$ then $B$ 
	consists of a single node.
	\item The total number of blocks is at most 
	$4|E'|/\densitythreshold$.
\end{enumerate} 
\begin{claim}
For any $\densitythreshold$ and $E'$, there exists an 
$(E',\densitythreshold)$-partitioning partitioning of the nodes 
of $T$ satisfying the above properties.
\end{claim}
\begin{proof}
This partitioning can be constructed by a greedy algorithm that traverses 
nodes of $T$ in increasing order of their numbering $N(\cdot)$ and groups them 
into blocks while the density of the block does not exceed $\densitythreshold$ 
(see \Cref{fig:main-algorithm-partitioning} for the precise procedure).
\begin{figure}[!h]
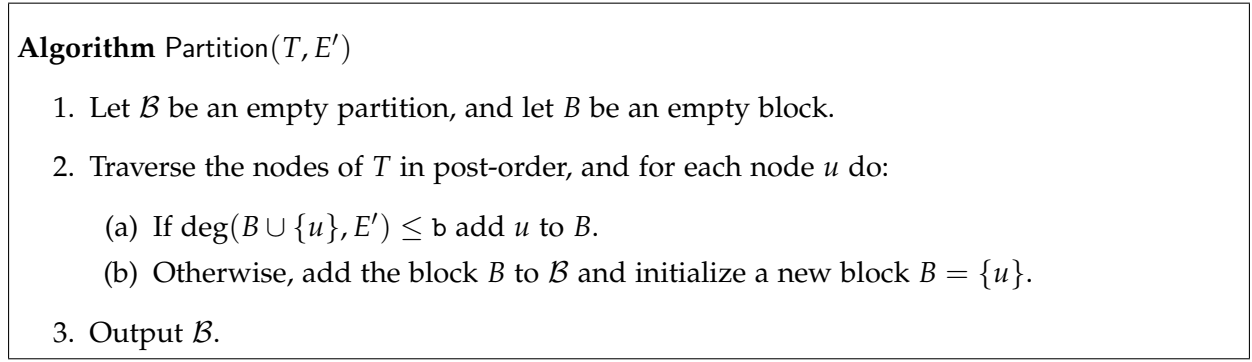

	\begin{boxedminipage}{\textwidth}
\vspace{3mm} \textbf{Algorithm $\Partition(T,E')$}
\begin{enumerate}
	\item Let $\mathcal{B}$ be an empty partition, and let $B$ be an empty 
	block.
	\item Traverse the nodes of $T$ in post-order, and for each node $u$ do:
	\begin{enumerate}
		\item If $\deg(B \cup \{u\},E')  \le 
		\densitythreshold$ add $u$ to $B$.
		\item Otherwise, add the block $B$ to $\mathcal{B}$ and initialize a 
		new block $B=\{u\}$.
	\end{enumerate}
	\item Output $\mathcal{B}$.
\end{enumerate}
	\end{boxedminipage}
	\caption{Partitioning procedure.}
	\label{fig:main-algorithm-partitioning}
\end{figure}
Indeed, properties 1 and 2 are satisfied directly by the construction. For 
property 3, let $t$ be the number of blocks $B$ with $\deg(B,E') \le 
\densitythreshold/2$. For such a block $B$ let $B'$ be the block that comes after $B$. By the construction, we know that $B'$ satisfies $\deg(B,E') + \deg(B',E') >
\densitythreshold$. Let $B_1,\ldots,B_{\ell}$ be the final 
partitioning. Then, we have $t$ pairs of blocks that have density at least 
$\densitythreshold$ and the rest of the $(\ell-t/2)$ blocks that have density 
at least $\densitythreshold/2$. Formally, we have
$$\sum_{i=1}^{\ell}\deg(B_i,E') \ge t\densitythreshold + (\ell - 
t/2)\densitythreshold/2 \ge \ell \densitythreshold/2.$$
On the other hand, since it is a partitioning of $E'$ we have that 
$\sum_{i=1}^{\ell}\deg(B_i,E') = 2|E'|$. Thus, we get that $\ell 
\densitythreshold/2 \le 2|E'|$ and therefore $\ell \le 4|E'|/\densitythreshold$ 
as required.
\end{proof}

Our algorithm for covering the edges of $E_1=E(G)\setminus E(T)$ makes use of this block 
partitioning with $\densitythreshold = 16$. For any two nodes $u,v \in V(T)$, 
The algorithm begins with an empty collection $\cC$ and then 
performs $\log n$ iterations where each iteration works as follows: Let 
$E'
\subseteq E_1$ be the set of uncovered edges (initially $E'=E_1$). 
Then, we partition the nodes of $T$ with respect to $E'$ and 
density parameter $\densitythreshold$. Finally, we search for cycles of length 
$O(\log n)$ between the blocks. If such a cycle exists, we map it to a 
cycle in $G$ by connecting nodes $u,v$ within a block by the path $\pi(u,v)$ in 
the tree $T$. This way a cycle of length $O(\log n)$ between the blocks translates 
to a cycle of length $O(D(T)\log n)$ in the original graph $G$. Denote the 
resulting collection by $\cC$.

We note that the cycles $\cC$ might not be \emph{simple}. This 
might happen if and only if the tree paths $\pi(v_i,u_{i+1})$ and 
$\pi(v_j,u_{j+1})$ intersect for some $j \in [t]$. Notice that the if an edge 
appears more than once in a cycle, then it must be a tree edge. Thus, we can 
transform any non-simple cycle $C$ into a collection of simple cycles that 
cover all edges that appeared only once in $C$ (the formal procedure is given 
at \Cref{fig:main-simplify-cycles}). Since these cycle are 
constructed to cover only non-tree edges, this transformation does not 
damage the covering of the $E_1$ edges. The formal description 
of the algorithm is given in \Cref{fig:main-algorithm-non-tree}.

\begin{figure}[!h]
	\begin{boxedminipage}{\textwidth}
\vspace{3mm} \textbf{Algorithm $\mathsf{NonTreeEdgeCover}(T,E_1)$}
\begin{enumerate}
	\item Initialize a cover $\cC$ as an empty set.
	\item Repeat $O(\log |E_1|)$ times:
	\begin{enumerate}
		\item Let $E' \subseteq E_1$ be the subset of all uncovered 
		edges.
		\item Construct an $(E', \densitythreshold)$-partitioning 
		$\mathcal{B}$ of the nodes of $T$.
		\item While there are $t$ edges  
		$(u_1,v_1),\ldots,(u_t,v_t) \in E'$ for $t \le \log n$ such 
		that for all $i \in [t-1]$, $v_i$ and $u_{i+1}$ are in the same block 
		and $v_t$ and $u_1$ are in the same block (with respect to the 
		partitioning $\mathcal{B}$):
		\begin{itemize}
			\item Add the cycle 
			$(u_1, v_1) \circ \pi(v_1,u_2) \circ (u_2,v_2) \circ \pi(v_2,u_3) \circ (u_3,v_3) \circ \ldots \circ (u_t,v_t) \circ (v_t,u_1)$ to $\cC$.
			\item Remove $(u_1,v_1),\ldots,(u_t,v_t)$ from $E'$.
		\end{itemize}
	\end{enumerate}
	\item Output $\cC' \gets \SimplifyCycles(\cC)$.
\end{enumerate}
	\end{boxedminipage}
	\caption{Procedure for covering non-tree edges.}
	\label{fig:main-algorithm-non-tree}
\end{figure}

\begin{figure}[!h]
	\begin{boxedminipage}{\textwidth}
\vspace{3mm} \textbf{Algorithm $\SimplifyCycles(\cC)$}
\begin{enumerate}
	\item While there is a cycle $C \in \cC$ with a vertex $w \in C$ that 
	appears more than once:
	\begin{enumerate}
		\item Remove $C$ from $\cC$.
		\item Let $C=(v_1,v_2) \circ \dots \circ (v_{k-1},v_{k})$ and define 
		$v_{k+i}=v_{i}$.
		\item Let $i_1,\ldots,i_{\ell}$ be such that $v_{i_j}=w$ for all $j \in 
		[\ell]$, and let $i_{\ell+1}=i_1$.
		\item For all $j \in [\ell]$ let $C_j=(v_{i_j},v_{i_j + 1})\circ  \dots 
		\circ (v_{i_{j+1}-1},v_{i_{j+1}})$, and if $|C_j| \ge 3$, add $C_j$ to 
		$\cC$.
	\end{enumerate}
	\item Output $\cC$.
\end{enumerate}
	\end{boxedminipage}
	\caption{Procedure making all cycles in $\cC$ simple.}
	\label{fig:main-simplify-cycles}
\end{figure}

We proceed with the analysis of the algorithm, and show that it yields the 
desired cycle cover. That is, we show three things: that every cycle has length 
at most $O(D(T)\log n)$, that each edge is covered by at most $O(\log n)$ 
cycles, and that each edge has at least one cycle covering it.

\paragraph{Cycle Length.}
The bound of the cycle length follows directly from the construction. The 
cycles added to the collection are of the form $(u_1, v_1) \circ 
\pi(v_1,u_2) 
\circ (u_2,v_2) \circ \pi(v_2,u_3) \circ (u_3,v_3) \circ \dots \circ 
(v_t,u_1)$, where 
each $\pi(v_i,u_{i+1})$ are paths in the tree $T$ and thus are of length at 
most $2D(T)$. Notice that the simplification process of the cycles can only 
make the cycles shorter. Since $t \le \log n$ we get that the cycle lengths are 
bounded by $O(D(T)\log n)$.

\paragraph{Congestion.}
To bound the congestion of the cycle cover we exploit the structure of the 
partitioning, and the fact that each block in the partition has a low density. 
We begin by showing that by the post-order numbering, all nodes in a 
given subtree have a continuous range of numbers. For every $z \in V$, let 
$\min_N(z)$ be the minimal number of a node in the subtree of $T$ rooted by 
$z$. That is, $\min_N(z)=\min_{u \in T(z)}N(u)$ and similarly let 
$\max_N(z)=\max_{u \in T(z)}N(u)$.
\begin{claim}\label{obs:postordernumber}
For every $z \in V$ and for every $u \in G$ it holds that (1) 
$\max_N(z)=N(z)$ and 
(2) $N(u) \in [\min_N(z), \max_N(z)]$ iff $u \in T(z)$.
\end{claim}
\begin{proof}
The proof is by induction on the depth of $T(z)$. For the base case, we consider 
the leaf nodes $z$, and hence $T(z)$ with $0$-depth, the claim holds vacuously. 
Assume that the claim holds for nodes in level $i+1$ and consider now a node 
$z$ in level $i$. Let $v_{i,1}, \ldots, v_{i,\ell}$ be the children of $z$ 
ordered from left to right. By the post-order traversal, the root $v_{i,j}$ is 
the last vertex visited in $T(v_{i,j})$ and hence 
$N(v_{i,j})=\max_N(v_{i,j})$. Since the traversal of $T(v_{i,j})$ starts right 
after finishing the traversal of $T(v_{i,j-1})$ for every $j\geq 2$, it holds 
that $\min_N(v_{i,j})=N(v_{i,j-1})+1$. Using the induction assumption for 
$v_{i,j}$, we get that all the nodes in $T(z)\setminus \{z\}$ have numbering 
in the range $[\min_N(v_{i,1}),\max_N(v_{i,\ell})]$ and any other node not in 
$T(z)$ is not in this range. Finally, $N(z)=N(v_{i,\ell})+1$ and so the claim 
holds.
\end{proof}

The cycles that we computed contains tree paths $\pi(u,v)$ that connect two 
nodes $u$ and $v$ in the \emph{same} block. Thus, to bound the congestion on a tree edge 
$e \in T$ we need to bound the number of blocks that contain a pair $u,v$ 
such that $\pi(u,v)$ passes through $e$. The next claim shows 
that every edge in the tree is effected by at most 2 
blocks.
\begin{claim}\label{lem:congestedge}
Let $e\in T$ be a tree edge and define $\mathcal{B}(e)=\{ B \in \mathcal{B} 
~\mid~ \exists u,v \in B \mbox{~s.t.~} e \in \pi(u,v)\}$. Then, 
$|\mathcal{B}(e)|\leq 2$ for every $e\in T$.
\end{claim}
\begin{proof}
Let $e=(w,z)$ where $w$ is closer to the root in $T$, and let $u,v$ be two 
nodes in the same block $B$ such that $e \in \pi(u,v)$. Let $\ell$ be the least common ancestor 
of $u$ and $v$ in $T$ (it might be that $\ell \in 
\{u,v\}$), then the tree path between $u$ and $v$ can be written as 
$\pi(u,v)=\pi(u,\ell)\circ \pi(\ell,v)$. Without loss of generality, assume 
that $e \in \pi(\ell,v)$. This implies that $v \in T(z)$ but $u \notin T(z)$. 
Hence, the block of $u$ and $v$ intersects the nodes of $T(z)$. Each block 
consists of a consecutive set of nodes, and by \Cref{obs:postordernumber} also 
$T(z)$ consists of a consecutive set of nodes with numbering in the range 
$[\min_N(z),\max_N(z)]$, thus there are at most two such blocks that intersect 
$e=(w,z)$, \ie blocks $B$ that contains both a vertex $y$ with $N(y)\in 
[\min_N(z),\max_N(z)]$ and a vertex $y'$ with $N(y')\notin  
[\min_N(z),\max_N(z)]$, and the claim follows.
\end{proof}

Finally, we use the above claims to bound the congestion. Consider any 
tree edge $e=(w,z)$ where $w$ is closer to 
the root than $z$. Recall that $T(z)$ be the subtree of $T$ rooted at $z$. Fix an 
iteration $i$ of the algorithm. We characterize all cycles in $\mathcal{C}$ 
that go through this edge.

For any cycle that passes through $e$ there must be a block $B$ and two nodes 
$u,v \in B$ 
such that $e \in \pi(u,v)$. By \Cref{lem:congestedge}, we know that there are 
that at each iteration of the algorithm, there are at most two such blocks $B$ 
that can affect the congestion of $e$.
Moreover, we claim that each such block has density at most 
$\densitythreshold$. Otherwise it would be a block containing a single node, 
say $u$, and thus the path $\pi(u,u)=u$ is empty and cannot contain the edge 
$e$. For each edge in $E'$ we construct a single cycle in $\cC$, and thus for 
each one of the two blocks that affect $e$ the number of pairs $u,v$ such that 
$e \in \pi(u,v)$ is bounded by $\densitythreshold/2$ (each pair $u,v$ has two 
edges in the block $B$ and we know that the total number of edges is bounded by 
$\densitythreshold$).

To summarize the above, we get that for each iteration, that are at most 2 
blocks that can contribute to the congestion of an edge $e$: one block that 
intersects $T(z)$ but has also nodes smaller than $\min_N(z)$ and one block that 
intersects $T(z)$ but has also nodes larger than $\max_N(z)$. Each of these two 
blocks can increase the congestion of $e$ by at most $\densitythreshold/2$. 
Since there are at most $\log n$ iterations, we can bound the total congestion 
by $b\log n$. Notice that if an edge appears $k$ times in a cycle, then this 
congestion bound counts all $k$ appearances. Thus, after the simplification of 
the cycles, the congestion remains unchanged.

\paragraph{Cover.}
We show that each edge in $E_1$ is covered by some cycle and that each cycle is used to cover $O(\log n)$ edges in $E_1$. We begin by showing the covering property of the 
preliminary cycle collection, before the simplification procedure. We later show that 
the covering is preserved even after simplifying the cycles. The idea is that at 
each iteration of the algorithm, the number of uncovered edges is reduced by 
half. Therefore, the $\log |E_1|=O(\log n)$ iterations should suffice for covering all 
edges of $E_1$. In each iteration we partition the nodes into blocks, and we 
search for cycles between the blocks. The point is that if the number of edges 
is large, then when considering the blocks as nodes in a new virtual graph, this 
graph has a large number of edges and thus must have a short cycle. 

In what follows, we formalize the intuition given above. Let $E'_i$ 
be the set $E'$ at the $\ith{i}$ iteration of the algorithm. 
Consider the iteration $i$ with the set of uncovered edge set 
$E'_{i}$. Our goal is to show that $E'_{i+1} \le 
1/2E'_{i}$. By having $\log |E_1|$ iterations, last set will be empty.

Let $\mathcal{B}_i$ be the partitioning performed at iteration $i$ with respect 
to the edge set $E'_{i}$. Define a super-graph $\supergraph$ in 
which each block $B_j \in \mathcal{B}_i$ is represented by a node 
$\widetilde{v}_j$, and there is an edge $(\widetilde{v}_j,\widetilde{v}_{j'})$ 
in $\supergraph$ if there is an edge in $E'_{i}$ between some node 
$u$ in $B_{j}$ and a node $u'$ in $B_{j'}$, \ie 
$$(\widetilde{v}_j,\widetilde{v}_{j'}) \in E(\supergraph) \iff 
E'_{i} 
\cap (B_{j} \times B_{j'}) \neq \emptyset.$$
See \Cref{fig:cyclevirtual} for an illustration. 
\begin{figure}[h!]
\begin{center}
\includegraphics[scale=0.35]{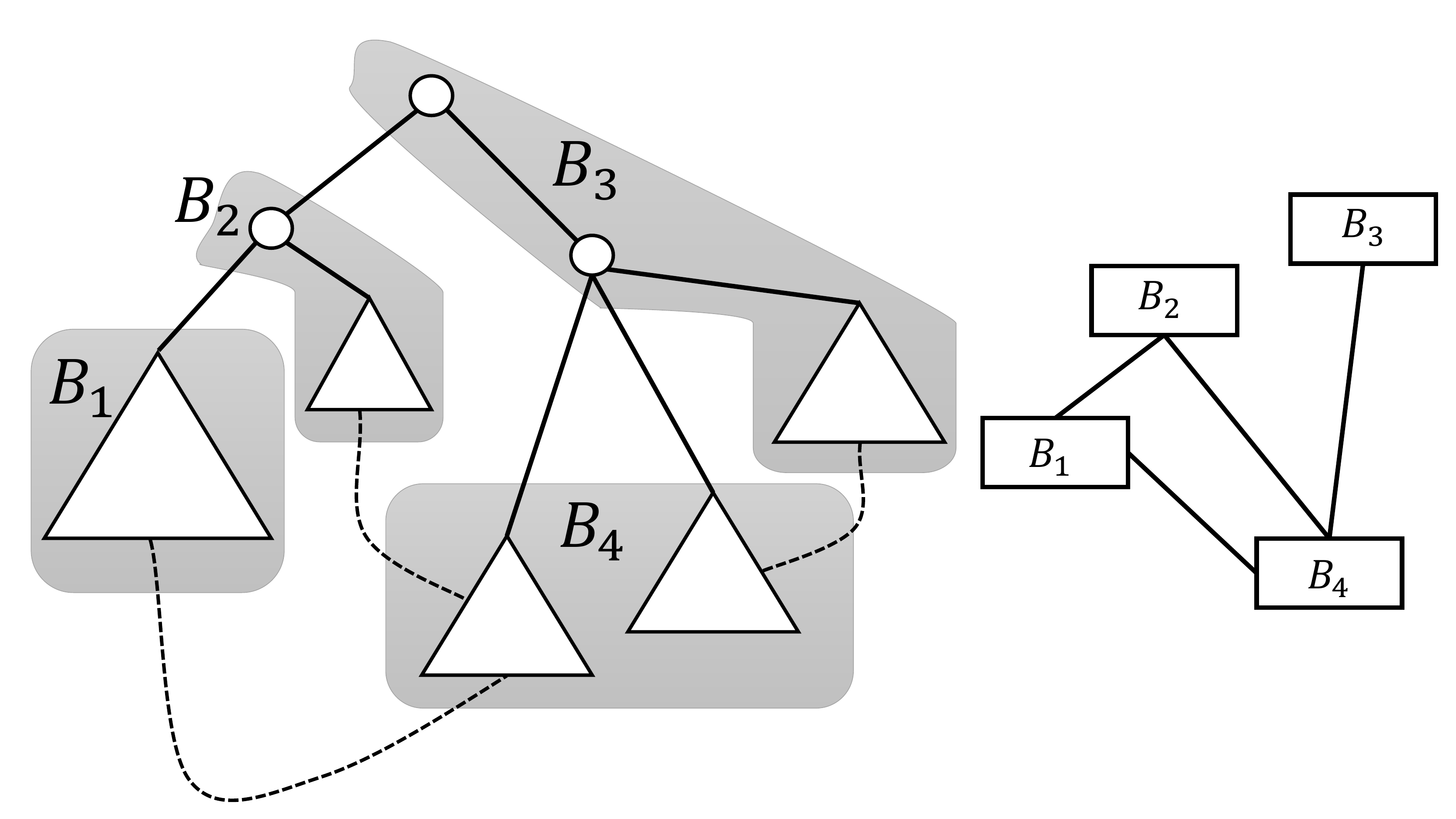}
\caption{Left: Schematic illustration of the block partitioning in the tree 
$T$. Dashed edges are those that remain to be covered after employing Alg. 
$\NonTreeCover$, where each two blocks are connected by exactly one edge. Each 
dashed edge corresponds to super-edges in $\supergraph$. Right: A triangle in 
the super-graph $\supergraph$.
\label{fig:cyclevirtual}}
\end{center}
\end{figure}
The number of nodes in $\supergraph$, which we denote by $n_i$, is the number 
of blocks in the partition and is bounded by 
$n_i \le 4|E'_i|/\densitythreshold$. Let $B(u)$ be the block of the 
node $u$. The algorithm finds cycles of
the form $(u_1, v_1) \circ \pi(v_1,u_2) \circ (u_2,v_2) \circ \pi(v_2,u_3) 
\circ \dots \circ (u_t,v_t) \circ (v_t,u_1)$, which is equivalent to 
finding the cycle $B(u_1),\ldots,B(u_t)$ in the graph $\supergraph$. In 
general, any cycle of length $t$ in $\supergraph$ is mapped to a cycle in $G$ 
of length at most $t \cdot D(T)$. Then, the 
algorithm adds 
the cycle to $\cC$ and removes the edges of the cycle (thus removing them also 
from $\supergraph$). At the end of iteration $i$ the graph 
$\supergraph$ has no cycles of length at most $\log n$. At this point, the next 
set of edges $E'_{i+1}$ is exactly the edges left in $\supergraph$. 
By Fact 
\ref{fc:girthmoore} (and recalling that $\densitythreshold=16$) we get that if 
$\supergraph$ does not have any cycles of length at most $\log n$ then we get 
the following bound on the number of edges:
$$E'_{i+1} \le 2n_i = 8|E'_i|/\densitythreshold = 
|E'_i|/2.$$

Thus, all will be covered by a cycle $C$ before the simplification process. We 
show that the simplification procedure of the cycle maintains the cover 
requirement. This stems from the fact that the only edges that might appear more than once 
in a cycle are tree edges. Thus even if we drop these edges, the non-tree edges remain covered.
It is left to show that this process only drops edges that appear more than once:
\begin{claim}
Let $C$ be a cycle and let $\cC'=\SimplifyCycles(C)$. Then, for 
every edge $e \in C$ that appears at most once in $C$ there is a cycle $C' \in 
\cC'$ such that $e \in C'$.
\end{claim}
\begin{proof}
The procedure $\SimplifyCycles$ works in iterations where in each 
iteration it chooses a vertex $w$ that appears more than once in $C$ and 
partitions the cycle $C$ to consecutive parts, $C_1,\ldots,C_{\ell}$. All edges 
in $C$ appear in some $C_j$. However, $C_j$ might not be a proper cycle since 
it might be the case that $|C_j| \le 2$. Thus, we show that in an edge $e \in 
C$ appeared at most once in a cycle $C$ then it will appear in $C_j$ for some 
$j$ where $|C_j| \ge 3$. We show that this holds for any iteration and thus 
will hold at the end of the process.

We assume without loss of generality that no vertex has two consecutive 
appearances. Denote $e=(v_2,v_3)$ and let $C=(v_1,v_2) \circ (v_2,v_3) \circ 
\dots \circ (v_{k-1},v_k)$ 
for $k \ge 3$. Since $e$ does not appear again in $C$ we know that 
$v_1,v_2,v_3$ are distinct. Thus, if $k=3$ then $C$ will not be split again and 
the claim follows.

Therefore, assume that $k \ge 4$. Since $e$ does not appear again in $C$ we 
know that $v_2,v_3,v_4$ are distinct (it might be the case that $v_1=v_4$). 
Thus, we know that $|\{v_1,v_2,v_3,v_4\}| \ge 3$. Any subsequence begins and 
ends at the same vertex and thus the subsequence $C_j$ that contains $e$ must 
contains all of $v_1,v_2,v_3,v_3$ and thus $|C_j| \ge 3$, and the claim follows.
\end{proof}
Finally, we turn to prove \Cref{lem:propnontree}. The lemma follows by noting that each cycle in $\cC_1$ contains at most $O(\log n)$ non-tree edges. To see this, observe that each cycle computed in the contracted block graph has length $O(\log n)$. Translating these cycles into cycles in $G$ introduces only \emph{tree} edges. We therefore have that each cycle is used to cover $O(\log n)$ non-tree edges.

\subsection{Covering Tree Edges}\label{sec:cover-tree-edges}
Finally, we present Algorithm $\TreeCover$ that computes a cycle cover for the 
tree edges. The algorithm is recursive and uses Algorithm $\NonTreeCover$ as a 
black-box. Formally, we show:
\begin{lemma}\label{lem:treecover}
For every $n$-vertex bridgeless graph $G$ and a tree $T \subseteq G$ of depth $D$, 
there exists a $(D\log n, \log^3 n)$ cycle cover for the 
edges of $T$. 
\end{lemma}
We begin with some notation. Throughout, 
when referring to a tree edge 
$(u,v)\in T$, the node $u$ is closer 
to the root of $T$ than $v$. Let $E(T)=\{e_1,\ldots,e_{n-1}\}$ be an ordering 
of the edges of $T$ in non-decreasing distance from the root.
For every tree edge $e \in T$, recall that the swap edge 
of $e$, denoted by $e'=\swap(e)$, is an arbitrary edge in $G$ that restores 
the connectivity of $T \setminus \{e\}$. Let $e=(u,v)$ (i.e., $u=p(v)$) and $(u',v')=\swap(e)$. Let $s(v)$ be the endpoint 
of $\swap(e)$ that does not belong to $T(u)$ (i.e., the subtree $T$ rooted at 
$u$), thus $v'=s(v)$. Define the 
$v$-$s(v)$ path
$$P_e=\pi(v,u') \circ \swap(e).$$ 
For an illustration see \Cref{fig:swapedge}. 

For the tree $T$, we construct a subset of tree edges denoted by $I(T)$ that we 
are able to cover. These edges are \emph{independent} in the sense that their $P_e$ paths are ``almost" edge disjoint (as will be shown next). The subset $I(T)$ is constructed by going through the edges 
of $T$ in non-decreasing distance from the root. 
At any point, we add $e$ to $I(T)$ only if it is not covered by the $P_{e'}$ 
paths of the $e'$ edges already added.

\begin{claim}
The subset $I(T)$ satisfies:
\begin{itemize}
\item For every $e \in E(T)$, there exists $e' \in I(T)$ such that $e \in e' 
\circ P_{e'}$.
\item For every $e,e' \in I(T)$ such that $e \ne e'$ it holds that $P_{e}$ and 
$P_{e'}$ have no \emph{tree} edge in common (no edge of $T$ is in both paths).
\item For every swap edge $(z,w)$, there exists at most two paths $P_{e},P_{e'}$ for $e,e'\in I(T)$ such that one passes through $(z,w)$ and the other through $(w,z)$. That is, each swap edge appears at most twice on the $P_e$ paths, once in each direction.
\end{itemize}
\end{claim}
\begin{proof}
The first property follows directly from the construction. 
Next, we show that they share no tree edge in common. Assume that there is a 
common edge $(z,w) \in P_e \cap P_{e'} \cap E(T)$. Then, both $e,e'$ must be on 
the path from the root to $z$ on the tree. Without loss of generality, assume that $e'$ is closer to the root than $e$.
We then get that $e \in P_{e'}$, leading to 
contradiction. For the third property, assume towards constriction that both 
$P_e$ and $P_{e'}$ use the same swap edge in the same direction. Again it 
implies that both $e,e'$ are on the path from the root to $z$ on $T$, and the same argument to the previous case yields that $e \in P_{e'}$, thus a contradiction.
\end{proof}

Our cycle cover for the $I(T)$ edges will be shown to cover all the edges of the tree $T$. This is because the  cycle 
that we construct to cover an edge $e \in I(T)$ necessarily 
contains $P_e$.

Algorithm $\TreeCover$ uses the following procedure $\EdgeDisjointPath$, usually used in the context of distributed routing. 

\paragraph{Key Tool: Route Disjoint Matching.}
Algorithm $\EdgeDisjointPath$ solves the following problem 
(\cite{klein1995nearly}, and Lemma 4.3.2 \cite{Peleg:2000}): given a rooted tree $T$ and a set of $2k$ 
marked nodes $M \subseteq V(T)$ for $k\leq n/2$, the goal is to find (by a 
distributed algorithm) a matching of these vertices $\langle w_i,w_j \rangle$ 
into pairs such that the tree paths $\pi(w_i,w_j,T)$ connecting the matched 
pairs are \emph{edge-disjoint}. This matching can be computed distributively in 
$O(\Diam(T))$ rounds by working from the leaf nodes towards the root. In each 
round a node $u$ that received information on more $\ell\geq 2$ unmarked nodes 
in its subtree, match all but at most one into pairs and upcast to its parent 
the ID of at most one unmarked node in its subtree.  It is easy to see that all 
tree paths between matched nodes are indeed edge disjoint. 

We are now ready to explain the cycle cover construction of the tree edges $E(T)$.

\paragraph{Description of Algorithm $\TreeCover$.}
We restrict attention for covering the edges of $I(T)$. The tree edges $I(T)$ will be covered in a specific manner that covers also the edges of $E(T)\setminus I(T)$.
The key idea is to define a collection of (virtual) non-tree edges 
$\widetilde{E}=\{(v,s(v)) : (p(v),v) ~\in I(T)\}$ and covering these 
non-tree edges by enforcing the cycle that covers the non-tree edge $(v,s(v))$ 
to cover the edges $e=(p(v),v)$ as well as the path $P_{e}$. Since every edge $e' \in 
T$ appears in one of the $e \circ P_e$ paths, this will guarantee that all tree 
edges are covered.


Algorithm $\TreeCover$ is recursive and has $O(\log n)$ levels of recursion. In 
each independent level of the recursion we need to solve the following 
sub-problem: Given a tree $T'$, cover by cycles the edges of $I(T')$ along with 
their $P_e$ paths. The key idea is to subdivide this problem 
into two independent and balanced subproblems. To do this, the tree $T'$ gets 
partitioned\footnote{This partitioning procedure is described in 
\Cref{clm:tree-partition}. We note that this partitioning maintains the 
layering structure of $T'$.} into two balanced edge 
disjoint subtrees $T'_1$ and $T'_2$, where $|T'_1|,|T'_2| \leq 2/3 \cdot |T'|$ 
and $E(T'_1)\cup E(T'_2)=E(T')$. 
Some of the tree edges in $T'$ are covered by applying a procedure that computes cycles using the edges of  $T'$, and the remaining ones will be covered recursively in either $T'_1$ or $T'_2$. 
Specifically, the edges of $I(T')$ are partitioned into $4$ types 
depending on the position of their swap edges. For every $x,y \in \{1,2\}$, let
\begin{align*}
E'_{x,y} = \{(u,v) \in E(T'_x) 
\cap I(T')~\mid~ 
 v \in V(T'_x) \mbox{~and~} s(v) \in V(T'_y)\setminus V(T'_x)\}.
\end{align*}
The algorithm computes a cycle cover $\cC_{1,2}$ (resp., $\cC_{2,1}$) for covering the edges of $E'_{1,2},E'_{2,1}$ respectively. The remaining edges $E'_{11}$ and $E'_{22}$ are covered recursively by applying the algorithm on $T'_1$ and $T'_2$ respectively. See Fig. \ref{fig:swapedge} for an illustration.

We now describe how to compute the cycle cover $\cC_{1,2}$ for the edges of $E'_{1,2}$. The edges $E'_{2,1}$ are covered analogously (i.e., by switching the roles of $T'_1$ and $T'_2$). Recall that the tree edges $E'_{1,2}$ are those edges $(p(v),v)$ such $v \in T'_1$ and $s(v) \in T'_2$. The procedure works in $O(\log n)$ phases, each phase $i$ computes three cycle collections $\cC'_{i,1},\cC'_{i,2}$ and $\cC'_{i,3}$ which together covers at least half of the yet uncovered edges of $E'_{1,2}$ (as will be shown in analysis). 

Consider the $i^{th}$ phase where we are given the set of yet uncovered 
edges $X_i \subseteq E'_{1,2}$.
We first mark all the vertices $v$ with $(p(v),v) \in X_i$. Let $M_i$ be this 
set of marked nodes. For ease of description, assume that $M_i$ is even, 
otherwise, we omit one of the marked vertices $w$ (from $M_i$) and take care of 
its edge $(p(w),w)$ in a later phase. We apply Algorithm 
$\EdgeDisjointPath(T'_1,M_{i})$ (see Lemma 4.3.2 \cite{Peleg:2000}) which 
matches the marked vertices $M_i$ into pairs $\Sigma=\{\langle v_1,v_2 \rangle 
\mid v_1,v_2 \in M_i \}$ such that for each pair $\sigma=\langle v_1,v_2 
\rangle$ there is a tree path $\pi(\sigma)=\pi(v_1,v_2,T'_1)$ and all 
the tree paths $\pi(\sigma),\pi(\sigma')$ are edge disjoint for every 
$\sigma,\sigma' \in \Sigma$.

Let $X''_i=\{e=(p(v),v) \in X_i : \exists v' \mbox{~and~} \langle v,v' 
\rangle \in \Sigma, \mbox{~s.t.~} e \in \pi(v,v',T'_1)\}$ be the set of edges 
in $X_i$ that appear in the collection of edge disjoint paths 
$\{\pi(\sigma),\sigma \in \Sigma\}$.
Our goal is to cover all edges in $E''_i=X''_i \cup \{P_e ~\mid~ e \in X''_i\}$ by cycles $\cC_i$. 
To make sure that all edges $E''_i$ are covered, we have to be careful that 
each such edge appears on a given cycle exactly once. Towards this end, we 
define a directed conflict graph $G_{\Sigma}$ whose vertex set are the pairs of 
$\Sigma$, and there is an arc $(\sigma',\sigma) \in A(G_{\Sigma})$ where 
$\sigma=\langle v_1,v_2 \rangle$, $\sigma=\langle v'_1,v'_2 \rangle$, if at 
least one of the following cases holds: Case (I) $e=(p(v_1),v_1)$ on 
$\pi(v_1,v_2,T'_1)$ and the path 
$\pi'=\pi(v'_1,v'_2,T'_1)$ intersects the edges of $P_{e}$; Case (II) $e'=(p(v_2),v_2)$ on $\pi(v_1,v_2,T'_1)$ and the path $\pi'$ intersects the edges of $P_{e'}$.
Intuitively, a cycle that contains both $\pi'$ and $P_e$ is not simple and in particular might not cover all edges on $P_e$. Since the goal of the pair $\sigma=\langle v_1,v_2 \rangle$ is to cover all edges on $P_e$ (for $e \in \pi(v_1,v_2,T'_1)$), the pair $\sigma'$ ``interferes" with $\sigma$.

In the analysis section (\Cref{lem:outdegone}), we show that the outdegree in the graph $G_{\Sigma}$ is bounded by $1$ and hence we can color $G_{\Sigma}$ with $3$ colors. This allows us to partition $\Sigma$ into three color classes $\Sigma_1, \Sigma_2$ and $\Sigma_3$. Each color class $\Sigma_j$ is an independent set in $G_{\Sigma}$ and thus it is ``safe" to cover all these pairs by cycles together. We then compute a cycle cover $\cC_{i,j}$, for each $j \in \{1,2,3\}$. The collection of all these cycles will be shown to cover the edges $E''_i$.

To compute $\cC_{i,j}$ for $j=\{1,2,3\}$, for each matched pair $\langle v_1,v_2 \rangle \in \Sigma_j$, we add to $T'_2$ a virtual edge $\widehat{e}$ between $s(v_1)$ and $s(v_2)$.
Let 
$$\widehat{E}_{i,j}=\{(s(v_1),s(v_2)) ~\mid~ \langle v_1,v_2 \rangle \in \Sigma_j\}.$$
We cover these virtual non-tree edges by cycles using Algorithm $\NonTreeCover$ 
on the tree $T'_2$ with the non-tree edges $\widehat{E}_{i,j}$. Let 
$\cC''_{i,j}$ be an $(O(D\log n),\widehat{O}(\log n))$-cycle cover that is the 
output of Algorithm 
$\NonTreeCover(T'_2,E'_{i,j})$. 
The output cycles of $\cC''_{i,j}$ are not yet cycles in $G$ as they consists 
of 
two types of virtual edges: the edges in $\widehat{E}_{i,j}$ and the edges $\widetilde{E}=\{(v,s(v)) ~\mid~ (p(v),v)\in I(T)\}$. 
First, we translate each cycle $C'' \in \cC''_{i,j}$ into a cycle $C'$ in $G \cup \widetilde{E}$ by replacing each of the virtual edges $\widehat{e}=(s(v_1),s(v_2))\in \widehat{E}_{i,j}$ in $C''$  with the path $P(\widehat{e})=(s(v_1),v_1) \circ \pi(v_1,v_2,T'_1) \circ (v_2,s(v_2))$. 
Then, we replace each virtual edge $(v,s(v))\in \widetilde{E}$ in $C'$ by the $v$-$s(v)$ path $P_{e}$ for $e=(p(v),v)$. This results in cycles $\cC_{i,j}$ in $G$. 

Finally, let $\mathcal{C}_{i}=\mathcal{C}_{i,1}\cup \mathcal{C}_{i,2}\cup \mathcal{C}_{i,3}$
and define $X_{i+1}=X_i \setminus X''_i$ to be the set of edges $e \in X_i$ that are not covered by the paths of $\Sigma$. If in the last phase $\ell=O(\log n)$, the set of marked nodes $M_\ell$ is odd, we omit one of the marked nodes $w \in M_{\ell}$, and cover its tree edge $e=(p(w),w)$ by taking the fundamental cycle of the swap edge $\swap(e)$ into the cycle collection.
The final cycle collection for $E'_{1,2}$ is given by $\cC_{1,2}=\bigcup_{i=1}^\ell \cC_i$. The same is done for the edges $E'_{2,1}$. 
This completes the description of the algorithm. The final collection of cycles is denoted by $\cC_3$. See 
\Cref{fig:main-algorithm-tree-edges} for the full description of the algorithm.
See \Cref{fig:swapedge,fig:treecoverfig,fig:treecoverfigmany} and for illustration.

\begin{figure}[h!]
	\begin{center}
		\includegraphics[scale=0.40]{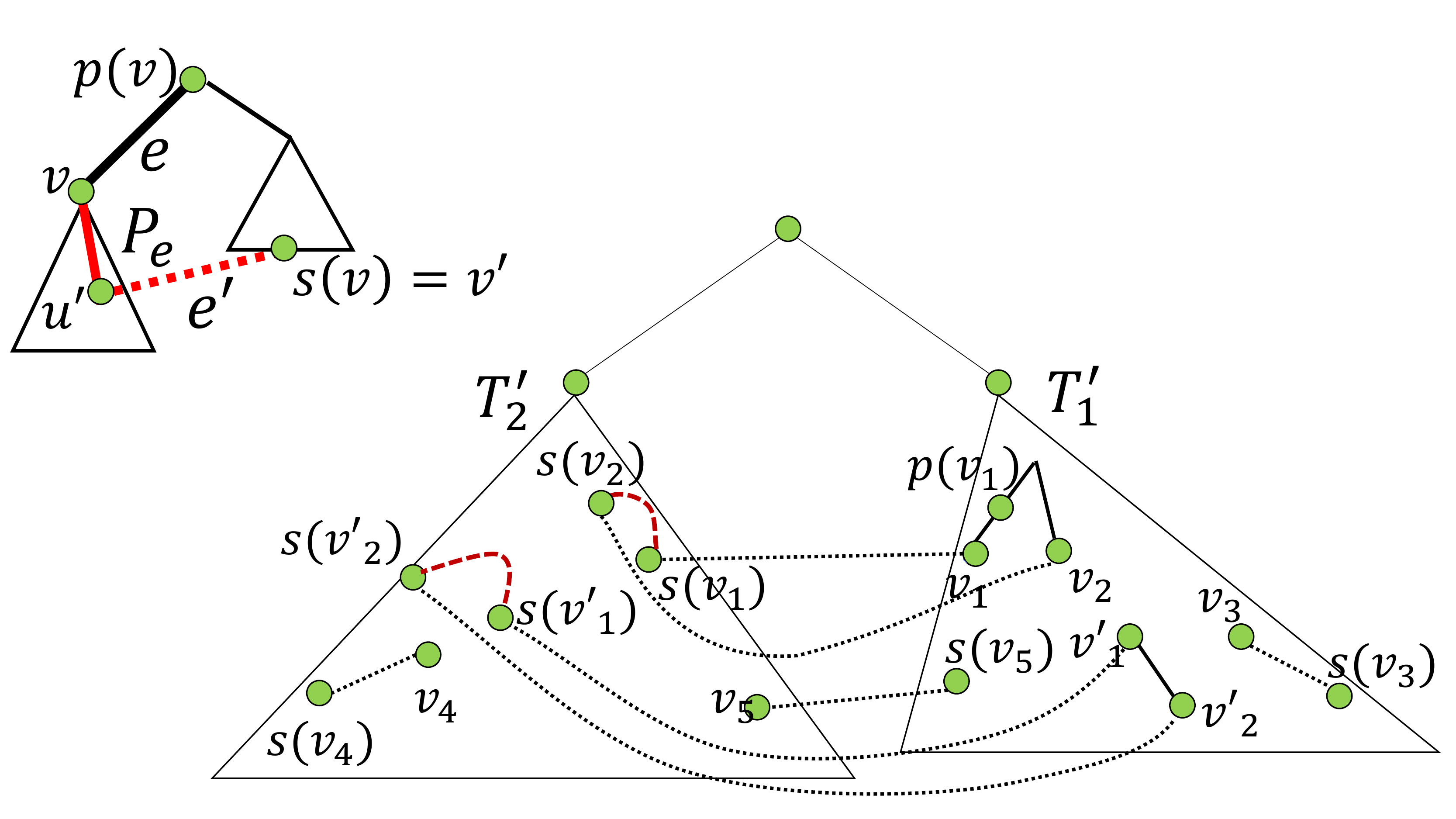}
		 \caption{Left: Illustration of the swap edge $e'=\swap(e)$ and the path $P_e$ 
			for an edge $e \in T$. For each tree edge $e=(u,v)\in T$, we add 
			the  
			auxiliary edge $(v,s(v))$. Right: The tree $T'$ is partitioned into two balanced trees $T'_1$ and $T'_2$. The root vertex in this example belongs to both trees. The edges $\widetilde{E}$ are partitioned into four sets: $E'_{1,1}$ (e.g., the edge $(p(v_3),v_3)$), $E'_{2,2}$ (e.g., the edge $(p(v_4),v_4)$), $E'_{1,2}$ (e.g., the edge $(p(v_1),v_1)$), $E'_{2,1}$ (e.g., the edge $(p(v_5),v_5)$). The algorithm covers the edges of $E'_{1,2}$ by using Algorithm $\EdgeDisjointPath$ to compute a matching and edge disjoint paths in $T'_1$. See the tree paths between $v_1$ and $v_2$ and $v'_1$ and $v'_2$. Based on this matching, we add virtual edges between vertices of $T'_2$, for example the edges $(s(v_1),s(v_2))$ and $(s(v'_1),s(v'_2))$ shown in dashed. The algorithm then applies Algorithm $\NonTreeCover$ to cover these non-tree edges in $T'_2$.
			\label{fig:swapedge}}
	\end{center}
\end{figure}

\begin{figure}[h!]
	\begin{center}
	 \includegraphics[scale=0.40]{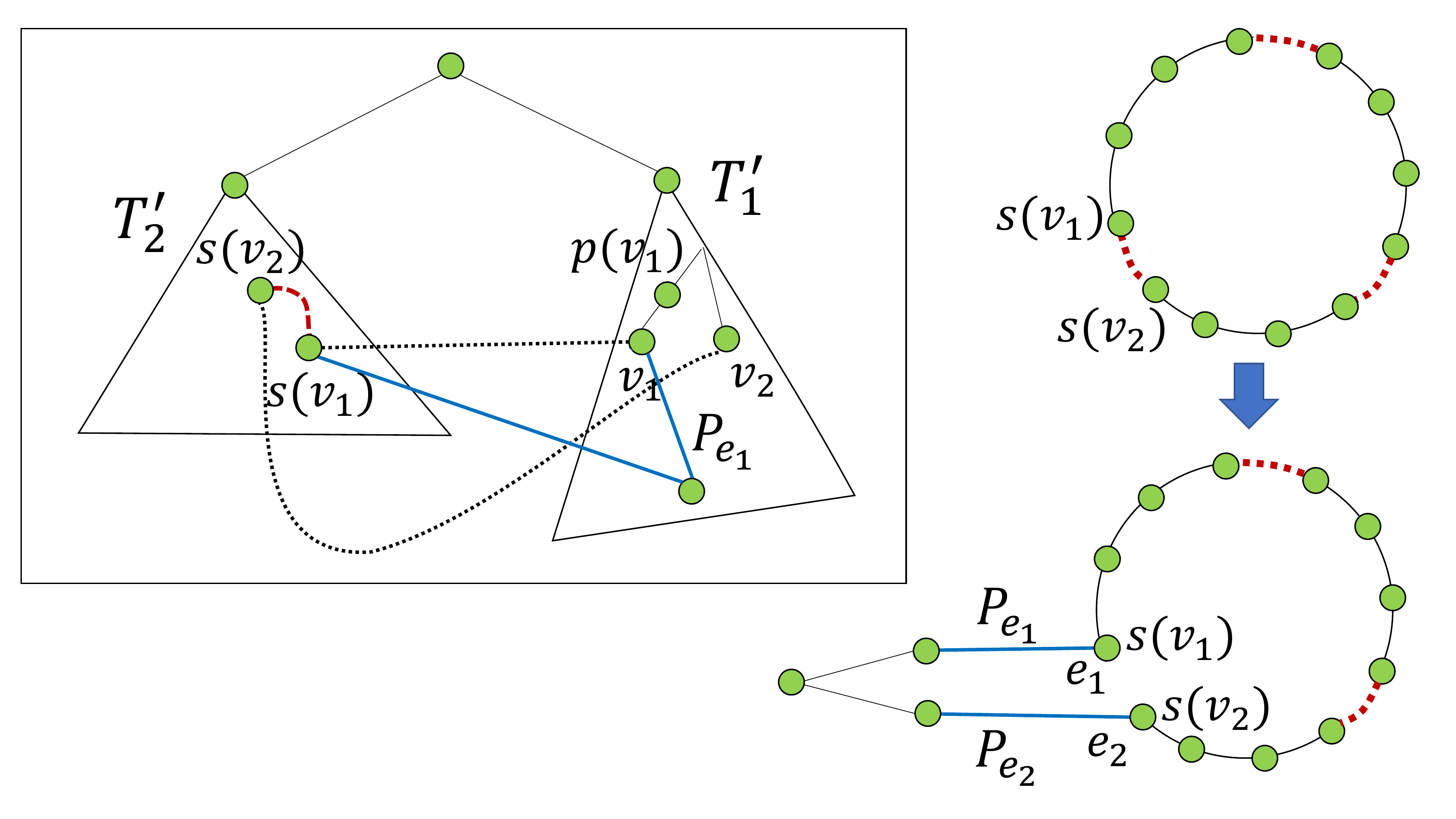}
		\caption{Illustration of replacing a single virtual edge $(s(v_1),s(v_2))$ in a cycle $C'' \in \cC''_{i,j}$ by an $s(v_1)$-$s(v_2)$ path in $G$. 
				\label{fig:treecoverfig}}
	\end{center}
\end{figure}

\begin{figure}[h!]
	\begin{center}
	 	\includegraphics[scale=0.40]{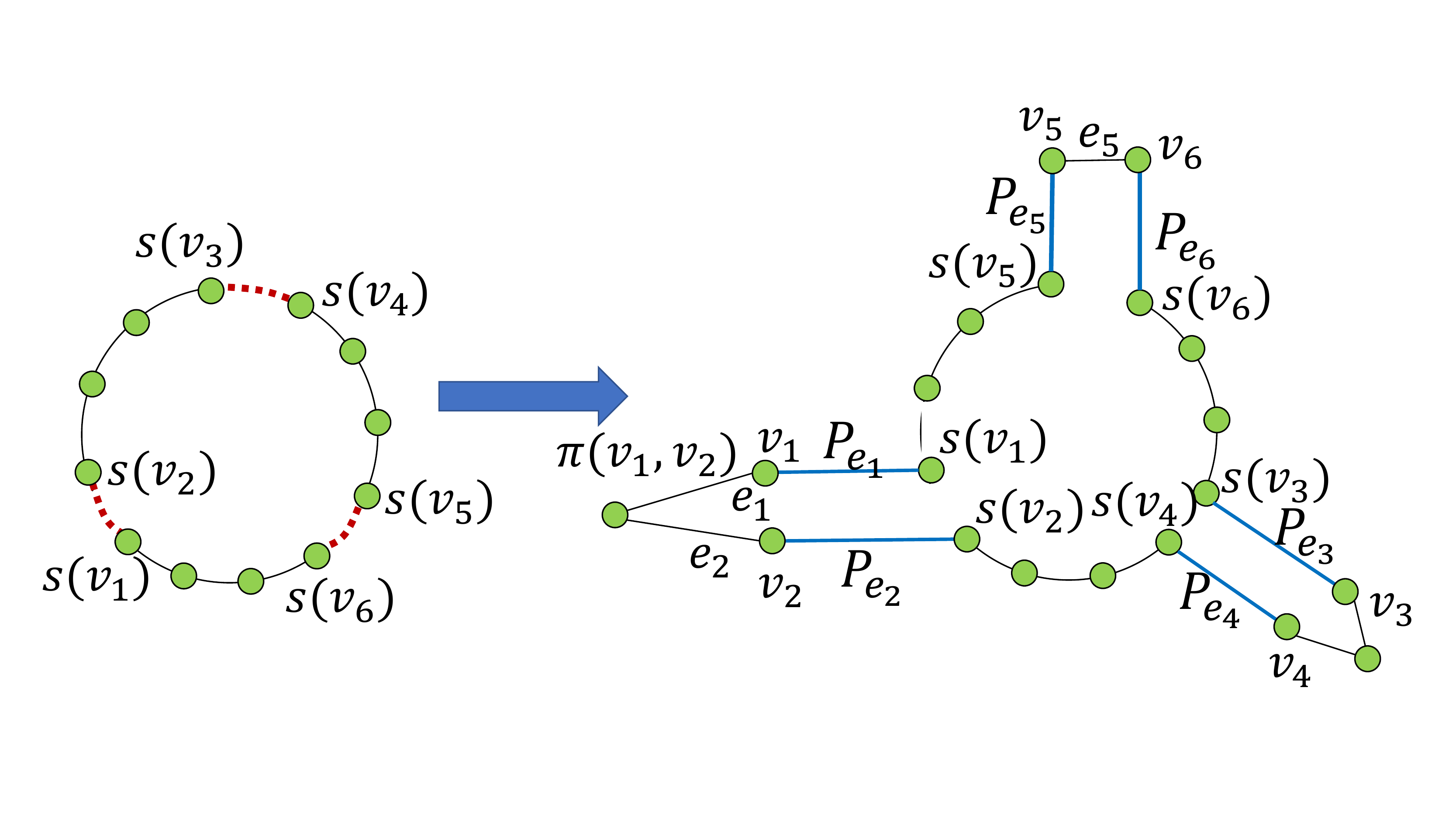}
		\caption{Translating virtual cycles into cycles in $G$. Each cycle contains $O(\log n)$ virtual edges which are are replaced by (almost) edge disjoint paths in $G$. Note that each edge on $e_i \circ P_{e_i}$ appears exactly once since the $P_{e}$ paths are tree-edge disjoint and the pairs $\sigma=\langle u_j,v_j \rangle$ in $\Sigma_{i,j}$ do not interfere with each other.
				\label{fig:treecoverfigmany}}
	\end{center}
\end{figure}

\begin{figure}[!h]
\begin{boxedminipage}{\textwidth}
	\vspace{3mm} \textbf{Algorithm $\TreeCover(T')$}
	\begin{enumerate}
		\item If $|T'|=1$ then output empty collection.
		\item Let $\cC$ be an empty collection.
		\item Partition $T'$ into balanced $T'_1 \cup T'_2$. 
		\item Let $E'$ be an empty set.
		\item For every $(u,v) \in T$ let $(u',v')=\swap((u,v))$ and add a 
		virtual edge $(v,v')$ to $E'$.
		\item For $i=1,...,O(\log n)$:
		\begin{enumerate}
		\item Let $M_{i}$ be all active nodes $v \in V(T'_1)$ s.t. $\swap(v) \in 
		V(T_2)\}$.
		\item Apply $\EdgeDisjointPath(T'_1,M_{i})$ and let $\Sigma=\{\langle v_1,v_2\rangle\}$ be the collection of matched pairs.
		\item Partition $\Sigma$ into $3$ \emph{independent} sets $\Sigma_1,\Sigma_2$ and $\Sigma_3$.
		\item For every $j \in \{1,2,3\}$ compute a cycle cover $\cC_{i,j}$ as follows:
		\begin{enumerate}
		\item For every pair $\langle v_1,v_2 \rangle$ in 
		$\Sigma_j$ add a virtual edge $(s(v_1),s(v_2))$ to $\widehat{E}_{i,j}$.
		\item $\cC''_{i,j} \gets \cC''_{i,j} \cup \NonTreeCover(T_2,\widehat{E}_{i,j})$.
		\item Translate $\cC''_{i,j}$ to cycles $\cC_{i,j}$ in $G$.
		\end{enumerate}
		\item Let $\cC_{i}=\cC_{i,1} \cup \cC_{i,2} \cup \cC_{i,3}$.
		\end{enumerate}
		\item $\cC_1= \bigcup_{i}\cC_{i}$.
		\item Repeat where $T'_1$ and $T'_2$ are switched.
		\item Add $\TreeCover(T'_1)\cup \TreeCover(T'_2)$ to $\cC$.
		\item Output $\mathsf{SimplifyCycles}(\cC \cup \cC_1)$.
	\end{enumerate}
\end{boxedminipage}
\caption{Procedure for covering tree edges.}
\label{fig:main-algorithm-tree-edges}
\end{figure}

We analyze the $\TreeCover$ algorithm and show that it finds short cycles, with 
low congestion and that every edge of $T$ is covered.

\paragraph{Short Cycles.}
By construction, each cycle that we compute using Algorithm $\NonTreeCover$ 
consists of at most $O(\log n)$ non-tree (virtual) edges $\widehat{E}_{i,j}$. The 
algorithm replaces each non-tree edge $\widehat{e}=(v_1,v_2)$ by an $v_1$-$v_2$ 
path in $G$ of length $O(D)$. 
This is done in two steps. First, $\widehat{e}=(v_1,v_2)$ is replaced by a path $P_{\widehat{e}}=(v_1,s(v_1))\circ \pi(v_1,v_2,T) \circ (v_2,s(v_2))$ in $G\cup \widetilde{E}$. Then, each $(v,s(v))$ edge is replaced by the path $P_{(p(v),v)}$ in $G$, which is also of length $O(D)$. Hence, overall the translated path $v_1$-$v_2$ path in $G$ has length $O(D)$. Since there are $O(\log n)$ virtual edges that are replaced on a given cycle, the cycles of $G$ has length $O(D\log n)$.

\paragraph{Cover.}
We start with some auxiliary property used in our algorithm.
\begin{claim}\label{lem:outdegone}
Consider the graph $G_{\Sigma}$ constructed when considering the edges in $E'_{1,2}$. 
The outdegree of each pair $\sigma'=\langle v'_1,v'_2 \rangle\in \Sigma$ in $G_{\Sigma}$ is at most $1$. Therefore, $G_{\Sigma}$ can be colored in $3$ colors. 
\end{claim}
\begin{proof}
Let $\sigma=\langle v_1,v_2 \rangle$ be such that $\sigma'$ interferes with 
$\sigma$ (i.e., $(\sigma',\sigma) \in A(G_{\Sigma})$). Without loss of 
generality, let $e=(p(v_1),v_1)$ be such that $e \in \pi(v_1,v_2,T'_1)$ and 
$\pi(v'_1,v'_2,T'_1)$ intersects the edges of $P_e$.

We first claim that this implies that $e$ appears above the least common ancestor of $v'_1$ and $v'_2$ in $T'_1$, and hence by the properties of our partitioning, also in $T$. Assume towards contradiction otherwise, since $e \circ P_e$ is a path on $T$ (where $e$ is closer to the root) and since $P_e$ intersects $\pi(v'_1,v'_2,T'_1)$, it implies that $e \in \pi(v'_1,v'_2,T'_1)$. 
Since the vertex $v_1$ is marked, we get a contradiction that $v'_1$ got matched with $v'_2$ as the algorithm would have matched $v_1$ with one of them. In particular, we would get that the paths $\pi(\sigma)$ and $\pi(\sigma')$ are \emph{not} edge disjoint, as both contain $e$. Hence, we prove that $e$ is above the LCA of $v'_1$ and $v'_2$. 

Next, assume towards contradiction that there is another pair $\sigma''=\langle 
v''_1,v''_2 \rangle\in \Sigma$ such that $\sigma'$ interferes with $\sigma''$. 
Without loss of generality, let $v''_1$ be such that $e''=(p(v''_1),v''_1)$ in 
on $\pi(v''_1,v''_2,T'_1)$ and $P_{e''}$ intersects with $\pi(v'_1,v'_2,T'_1)$. 
This implies that $e''$ is also above the LCA of $v'_1$ and $v'_2$ in $T'_1$. Since one of the edges of $P_{e''}$ is on $\pi(v'_1,v'_2,T'_1)$ it must be that either $e''$ on $P_{e}$ or vice verca, in contradiction that $e,e'' \in I(T)$. 
\end{proof}
We now claim that each edge $e \in T$ is covered. By the definition of $I(T)\subseteq E(T)$, it is sufficient to show that:
\begin{claim}
For every edge $e \in I(T)$, there exists a cycle $C \in \cC_3$ such that 
$e\circ P_e \subseteq C$. 
\end{claim}
Consider a specific tree edge $e=(p(v),v)$. First, note that since $(v,s(v))$ 
is a non-tree edge, there must be some recursive call with the tree $T'$ such 
that $v \in T'_1$ and $s(v) \in T'_2$ where $T'_1$ and $T'_2$ are the balanced 
partitioning of $T'$. At this point, $(v,s(v))$ is an edge in $E'_{1,2}$. 
We show that in the $\ell=O(\log n)$ phases of the algorithm for covering the 
$E'_{1,2}$, there is a phase in which $e$ is covered. 
\begin{claim}\label{cl:coverdisjoint}
(I) For every $e=(p(v),v) \in E'_{1,2}$ except at most one edge $e^*$, there is 
a phase $i_e$ where Algorithm $\EdgeDisjointPath$ matched $v$ with some $v'$ 
such that $e \in \pi(v,v')$. \\
(II) Each edge $e \neq e^*$ is covered by the cycles computed in phase $i_e$. 
\end{claim}
\begin{proof}
Consider phase $i$ where we cover the edges of $X_i$. Recall that the 
algorithm marks the set of nodes $v$ with $(p(v),v) \in X_i$, resulting in the 
set $M_i$. Let $\Sigma$ be the output pairs of Algorithm 
$\EdgeDisjointPath(T'_1,M_{i})$. We first show that at least half of the edges 
in $X_i$ are covered by the paths of $\Sigma$. 

If $M_i$ is odd, we omit one of the marked nodes and then apply Algorithm 
$\EdgeDisjointPath$ to match the pairs in the even-sized set $M_i$. The key 
observation is that for every matched pair $\langle v_1,v_2 \rangle$, it holds 
that either $(p(v_1),v_1)$ or $(p(v_2),v_2)$ is on $\pi(v_1,v_2,T'_1)$ (or 
both). Hence, at least half of the edges of $X_i$ are on the edge disjoint 
paths $\pi(v_1,v_2,T'_1)$.

We therefore get that after $\ell=c\log n$ phases, we are left with $|M_{\ell}|=O(1)$ at that point if $|M_{\ell}|$ is odd, we omit one vertex $v^*$ such that $e^*=(p(v^*),v^*)$. Claim (I) follows.

We now consider (II) , let $e=(p(v),v)$ and consider phase $i=i_e$ in which $e \in \pi(v,v',T'_1)$ where $v'$ is the matched pair of $v$. We show that all the edges of $e \circ P_e$ are covered by the cycles $\cC_i$ computed in that phase.
By definition, $\langle v,v' \rangle$ belongs to $\Sigma$. By \Cref{lem:outdegone}, $G_{\Sigma}$ can be colored by $3$ colors, let $\Sigma_j \subseteq \Sigma$ be the color class that contains $\langle v,v' \rangle$.

We will show that there exists a cycle $C$ in $\cC_{i,j}$ that covers each edge $e'' \in e \circ P_e$ exactly once.
Recall that the algorithm applies Algorithm $\NonTreeCover$ which computes a 
cycle cover $\cC''_{i,j}$ to cover all the virtual edges $\widehat{E}_{i,j}$ in 
$T'_2$. Also, $(s(v),s(v'))\in \widehat{E}_{i,j}$.

Let $C''$ be the (simple) cycle in $\cC''_{i,j}$ that covers the virtual edge $(s(v),s(v'))$. 
In this cycle $C''$ we have two types of edges: edges in $T'_2$ and virtual edges $(s(v_1),s(v_2))$.  
First, we transform $C''$ into a cycle $C'$ in which each virtual edge $\widehat{e}=(s(v_1),s(v_2))$ is replaced by a path $P(\widehat{e})=(s(v_1),v_1) \circ \pi(v_1,v_2,T'_1) \circ (v_2,s(v_2))$. 
Next, we transform $C'$ into $C \subseteq G$ by replacing each edge $(v_1,s(v_1)) \in \widetilde{E}$ in $C''$ by the $v_1$-$s(v_1)$ path $P_{e_1}$ for $e_1=(p(v_1),v_1)$. 

We now claim that the final cycle $C \subseteq G$, contains each of the edges $e \circ P_e$ exactly once, hence even if $C$ is not simple, making it simple still guarantees that $e \circ P_e$ remain covered.
Since $T'_1$ and $T'_2$ are edge disjoint, we need to restrict attention only two types of $T'_1$ paths that got inserted to $C$: 
 (I) the edge disjoint paths $\Pi_{i,j}=\{\pi(v_1,v_2,T'_1) \mid  \langle 
 v_1,v_2 \rangle \in \Sigma_j\}$ and (II) the $v'$-$s(v')$ paths $P_{e'}$ for 
 every edge $e'=(p(v'),v')$ (appears on $C'$).

We first claim that there is exactly one path $\pi(v_1,v_2,T'_1) \in \Pi_{i,j}$ that contains the edge $e=(p(v),v)$. By the selection of phase $i$, $e \in \pi(v,v',T'_1)$ where $v'$ is the pair of $v$. Since all paths $\Pi_{i,j}$ are edge disjoint, no other path contains $e$. Next, we claim that there is \emph{no} path $\pi \in \Pi_{i,j}$ that passes through an edge $e' \in P_e$. Since $e=(p(v),v) \in \pi(v,v',T'_1)$ and all edges on $P_e$ are \emph{below} $e$ on $T$\footnote{Since our partitioning into $T'_1,T'_2$ maintains the layering structure of $T$, it also holds that $P_e$ is below $e$ on $T'_1$.}, the path $\pi(v,v',T'_1)$ does not contain any $e' \in P_e$.
In addition, since all pairs in $\Sigma_j$ are independent in $G_{\Sigma}$, 
there is no path in $\pi(\sigma') \in \Pi_{i,j}$ that intersects $P_e$ (as in such a case, $\sigma$ interferes with $\langle v,v' \rangle$). We get that $e$ appears 
exactly once on $\Pi_{i,j}$ and no edge from $P_e$ appears on $\Pi_{i,j}$. 
Finally, we consider the second type of paths in $T'_1$, namely, the $P_{e'}$ paths. By 
construction, every $e' \in X_i$ is in $I(T)$ and hence that $P_{e'}$ and $P_{e}$ share 
no tree edge. We get that when replacing the edge $(v,s(v))$ with $P_e$ all edges $e' \in P_e$ appears and non of the tree edges on $P_e$ co-appear on some other $P_{e''}$.
All together, each edge on $e \circ P_e$ appears on the cycle $C$ exactly 
once.  This completes the cover property. 
\end{proof}
Since the edge $e^*$ is covered by taking the fundamental cycle of its swap edge, we get that all edges of $E'_{1,2}$ are covered. Since each edge $(v,s(v))$ belongs to one of these $E'_{1,2}$ sets, the cover property is satisfied. 

\paragraph{Congestion.}
A very convenient property of our partitioning of $T'$ into two trees $T'_1$ 
and $T'_2$ is that this partitioning is closed for LCAs. In particular, for $j 
\in \{1,2\}$ then if $u,v \in T'_j$, the LCA of $u,v$ in $T'$ is also in 
$T'_j$. Note that this is in contrast to blocks of \Cref{sec:cover-non-tree} 
that are not closed to LCAs. 
We begin by proving by induction on $i=\{1,\ldots, O(\log n)\}$ that all the trees $T',T''...$ considered in the same recursion level $i$ are edge disjoint. In the first level, the 
claim holds vacuously as there is only the initial tree $T$. 
Assume it holds up to level $i$ and consider level $i+1$. As each tree $T_j$ in 
level $i-1$ is partitioned into two edge disjoint trees in level $i+1$, the claim holds.

Note that each edge $e=(v,s(v))$ is considered exactly once, i.e., in one recursion call on $T'=T'_1\cup T'_2$ where without loss of generality, $v \in T'_1$ and $s(v) \in T'_2\setminus T'_1$. By \Cref{cl:coverdisjoint}, there is at most one edge $e^* \in E'_{1,2}$, which we cover by taking the fundamental cycle of $\swap(e^*)$ in $T$. 

We first show that the congestion in the collection of all the cycles added in this way is bounded by $O(\log n)$.
To see this, we consider one level $i$ of the recursion and show that each edge 
appears on at most $2$ of the fundamental cycles $\mathcal{F}_i$ added in that 
level. 
Consider an edge $e^*$ that is covered in this way in level $i$ of the recursion. That is the fundamental cycle of $\swap(e^*)$ given by $\pi(v^*,s(v^*))\cup P_{e^*}$ was added to $\mathcal{F}_i$. 
Let $T'$ be such that $T'=T'_1 \cup T'_2$ and $e^*=(p(v^*),v^*)$ is such that $v^* \in T'_1$ and $s(v^*)\in T'_2$. Since both $v$ and $s(v)$ are in $T'$, the tree path $\pi(v^*,s(v^*))\subseteq T'$. As all other trees $T''\neq T'$ in level $i$ of the recursion are edge disjoint, they do not have any edge in common with $\pi(v^*,s(v^*))$.
For the tree $T'$, there are at most two fundamental cycles that we add. One for covering an edge in $E'_{1,2}$ and one for covering an edge in $E'_{2,1}$. 
Since $e^* \in I(T)$, and each edge appears on at most two paths $P_{e},P_{e'}$ for $e,e' \in I(T)$, overall each edge appears at most twice on each of the cycles in $\mathcal{F}_i$ (once in each direction of the edge) and over all the $O(\log n)$ of the recursion, the congestion due to these cycles is $O(\log n)$.

It remains to bound the congestion of all cycles obtained by translating the 
cycles computed using Algorithm 
$\TreeCover$. We do that by showing that the cycle collection $\cC_i$ computed in phase $i$ to cover the edges of $E'_{1,2}$ is an $O(D\log n, \log n)$ cover. Since there are $O(\log n)$ phases and $O(\log n)$ levels of recursion, overall it gives an $O(D\log n, \log^3 n)$ cover.

Since all trees considered in a given recursion level are edge disjoint, we consider one of them: $T'$.
We now focus on phase $i$ of Algorithm $\TreeCover(T')$. In particular, we 
consider the output cycles $\cC''_{i,j}$ for $j \in \{1,2,3\}$ computed by 
Algorithm $\NonTreeCover$ for the edges $\widehat{E}_{1,2}$ and $T'_2$. Each 
edge $e \in T'_2$ appears on $O(\log n)$ cycles of $\cC''_{i,j}$. Each virtual 
edge $\widehat{e}=(s(v_1),s(v_2))$ is replaced by an $s(v_1)$-$s(v_2)$ path 
$P(\widehat{e})=(s(v_1),v_1)\circ 
\pi(v_1,v_2,T'_1)\circ (s(v_1),v_1)$ in $G \cup \widetilde{E}$. Let $\cC'_{i,j}$ be the cycles in $G \cup \widetilde{E}$ obtained from $\cC''_{i,j}$ by replacing the edges of $\widehat{e}\in \widehat{E}_{1,2}$ with the paths $P(\widehat{e})$ in $G \cup \widetilde{E}$. 
Note that every two paths $P(\widehat{e})$ and $P(\widehat{e}')$ are edge disjoint for every $\widehat{e},\widehat{e}' \in \widehat{E}_{1,2}$.
The edges $(s(v_1),v_1)$ of $\widetilde{E}$ gets used only in tree $T'$ in that 
recursion level. Hence, each edge $(v_1,s(v_1))$ appears on $O(\log n)$ cycles 
$C'$ in $G \cup \widetilde{E}$.

Since the paths $\pi(v_1,v_2,T'_1)$ are edge disjoint, each edge $e' \in \pi(v_1,v_2,T'_1)$ appears on at most $O(\log n)$ cycles $C'$ in $G \cup \widetilde{E}$ (i.e., on the cycles translated from $C'' \in \cC''_{i,j}$ that contains the edge $\widehat{e}=(s(v_1),s(v_2))$). Up to this point we get that each virtual edge $(v,s(v))\in \widetilde{E}$ appears on $O(\log n)$ cycles of $\cC'_{i,j}$. Finally, when replacing $(v,s(v))$ with the paths $P_{(p(v),v)}$, the congestion in $G$ is increased by factor of at most $2$ as every two path $P_e$ and $P_{e'}$ for $e,e' \in I(T)$, are nearly edge disjoint (each edge $(z,w)$ appears on at most twice of these paths, one time in each direction). We get that the cycle collection $\cC_i$ is an $O(D\log n,\log n)$ cover, as desired.

Finally, we conclude by observing that our cycle cover algorithm $\GraphCover$ does not require $G$ to bridgeless by rather covers by a cycle, every edge that appears on some cycle in $G$.
\begin{observation}\label{obs:nice}
Algorithm $\GraphCover$ covers every edge $e$ that appears on some cycle in $G$, hence it is \emph{nice}.
\end{observation}
\begin{proof}
Every non-tree edge is clearly an edge the appears on a cycle, an Alg. $\NonTreeCover$ indeed covers all non-tree edges. In addition, every tree edge that appears on a cycle, has a swap edge. Since Alg. $\TreeCover$ covers all swap edges while guaranteeing that their corresponding tree edges get covered as well, the observation follows.
\end{proof}

\subsection{Universally Optimal Cycle Covers}\label{sec:fulloptcovers}
For each edge $e=(u,v) \in E(G)$, let $C_e$ be the shortest cycle in $G$ that 
contains $e$ and let $\OptCCVal=\max_{e \in G} |C_e|$. Clearly, for every 
$(\dilation,\congestion)$-cycle cover $\mathcal{C}$ of $G$, it must hold that each cycle length is at least $\OptCCVal$, even when there is \emph{no} constraint on the congestion $\congestion$. 
%
%

An algorithm $\cA$ for constructing 
cycle covers is \emph{nice} if it does not require $G$ to be bridgeless, but rather covers by cycles all $G$-edges that lie on a 
cycle in $G$. In particular, Algorithm $\GraphCover$ of \Cref{sec:cyclecoverall} is nice (see \Cref{obs:nice}). 

Consider a nice algorithm $\cA$ that constructs an $(f_{\cA}(D), 
\congestion)$-cycle 
cover. We describe Alg. $\cAOPT$ for constructing an
$(\dilation', \congestion')$-cycle cover with $\dilation'=f_{\cA}(\widetilde{O}(\OptCCVal))$ and $\congestion'=\widetilde{O}(\congestion)$. 
Taking $\cA$ to be the $\GraphCover$ algorithm of \Cref{sec:cyclecoverall} yields the theorem. 
Alg. $\cAOPT$ is based on the notion of neighborhood-covers \cite{awerbuch1996fast} (also noted by \emph{ball carving}). 
\begin{definition}[Neighborhood Cover, \cite{awerbuch1996fast}]\label{def:neighborcover}
A $(k,t,q)$ $t$-\emph{neighborhood cover} of an $n$-vertex graph $G=(V,E)$ is 
a 
collection of subsets of $V$ (denoted as \emph{clusters}), $\mathcal{S}=\{S_1,\ldots, S_r\}$ with the 
following properties:
\begin{itemize}
\item
For every vertex $v$, there exists a cluster $S_i$ such that $\Gamma_t(V) \subseteq S_i$. 
\item 
The diameter of each induced sugraph $G[S_i]$ 
is at most $\widetilde{O}(t)$.
\item 
Each vertex belongs to at most $\widetilde{O}(1)$ clusters.
\end{itemize}
\end{definition}
Alg. $\cAOPT$ first constructs an $\OptCCVal$-neighborhood cover $\mathcal{N}=\{S_1,\ldots, S_r\}$. Thus the diameter $D_i$ of each subgraph $G[S_i]$ is $\widetilde{O}(\OptCCVal)$. Next, Algorithm $\cA$ is applied in each subgraph 
$G[S_i]$ simultaneously, computing a
$(f_{\cA}(D_i), \congestion)$-cycle cover $\mathcal{C}_i$ for every $i \in 
\{1,\ldots, r\}$. The output cover is $\mathcal{C}^*=\bigcup_{i=1}^r 
\mathcal{C}_i$. 

The key observation is that since each edge $e$ lies on a cycle $C_e$ of length 
$O(\OptCCVal)$ in $G$, there exists a subgraph $G[S_i]$ that contains the entire cycle $C_e$. Since Alg. $\cA$ is nice, $e$ gets covered in the cycle collection $\mathcal{C}_i$ computed by Alg. $\cA$ in $G[S_i]$. 
As the diameter of 
$G[S_i]$ is $\widetilde{O}(\OptCCVal)$, the length of all cycles is bounded by
$f_{\cA}(\OptCCVal)$. 
Finally, since each edge appears on $\widetilde{O}(1)$ clusters, we have that each edge appears on $\widetilde{O}(\congestion)$ cycles in the output cycle cover $\mathcal{C}^*$. 

To provide a cycle cover that is almost-optimal with respect each individual edge, we 
repeat the above procedures for $O(\log \OptCCVal)$ many times. In the $i^{th}$ 
application, the algorithm constructs $2^i$-neighborhood cover and applies Alg. 
$\cA$ in each of the resulting clusters. The output cycle cover is the union of 
all cycle covers computed in these applications.

An edge $e$ that lies on a cycle of length $\ell=|C_e|$ in $G$, will be a covered by the cycles computed in the $\lceil \log \ell \rceil$ iteration. Since the cycles computed in that iteration are computed in clusters of $2^{\lceil \log \ell \rceil}$-neighborhood cover, $e$ will be covered by a cycle of length $\widetilde{O}(|C_e|)$. 
Finally, the $O(\log \OptCCVal)$ repetitions increases the congestion by factor of at most $O(\log n)$, the claim follows. 
We now provide a detailed analysis of Algorithm $\cAOPT$ for proving \Cref{thm:cyclcoveropt}.

\paragraph{Edge cover.}
We show that $\mathcal{C}^*$ is a cover. 
Consider an edge $e=(u,v)$. By the definition of the neighborhood cover there 
exist an $i \in [r]$ such that $\Gamma_{\OptCCVal}(u) \subseteq S_i$. 
Since $\cA$ is nice, each edge in $G[S_i]$ that 
belongs to some cycle in $G[S_i]$ is covered by the output cycle cover of Algorithm 
$\cA$. As $e$ belongs to a cycle $C_e$ in $G$ of length at most 
$\OptCCVal$, it 
holds that $C_e \subseteq G[S_i]$, thus $e$ is covered by 
a cycle in $\mathcal{C}_i$.

\paragraph{Cycle length.}
By the construction of the neighborhood cover, 
the strong diameter $D_i$ of each $G[S_i]$ is $\widetilde{O}(\OptCCVal)$. Thus, when we run 
$\cA$ on $G[S_i]$, it returns a cycle cover where each cycle is of 
length at most
$f_{\cA}(D_i)=f_{\cA}(\widetilde{O}(\OptCCVal))$. 

\paragraph{Congestion.}
Since $\mathcal{C}_i$ is an $(f(D_i), \congestion)$ cover, each edge $e \in G[S_i]$, appears on at most $\congestion$ cycles in $\mathcal{C}_i$ for every 
$i \in \{1,\ldots, r\}$. Since each vertex $v$ appears in at most $q=O(\log n)$ 
different clusters $S_j \in \mathcal{S}$, overall, we get that each edge 
appears on $O(\congestion \cdot \log n)$ cycles in $\mathcal{C}^*$. 
In \Cref{fig:algorithm-opt,fig:algorithm-opt-edge}, we describe the pseudocodes when taking $\cA$ to be algorithm $\GraphCover$ of \Cref{sec:cyclecoverall} which constructs $(\widetilde{O}(D),\widetilde{O}(1))$ cycle covers.
\begin{figure}[!h]
\begin{boxedminipage}{\textwidth}
\vspace{3mm} \textbf{Algorithm $\OptimalCycleCover$:}
\begin{enumerate}
	\item Construct a $(k,t,q)$ neighborhood cover $S_1,\ldots, S_r$ for 
	$k=O(\log n)$, $t=\OptCCVal$ and $q=\widetilde{O}(1)$.
	\item For each $i \in [r]$ run $\GraphCover$ on $G[S_i]$ to get 
	$\mathcal{C}_i$. 
	\item Output $\mathcal{C}^*=\bigcup_{i=1}^r 
	\mathcal{C}_i$.
\end{enumerate}
\end{boxedminipage}
\caption{Description of the nearly optimal algorithm.}
\label{fig:algorithm-opt}
\end{figure}

\begin{figure}[!h]
\begin{boxedminipage}{\textwidth}
\vspace{3mm} \textbf{Algorithm $\OptimalEdgeCycleCover$:}
\begin{enumerate}
	\item For $i = 1 \ldots \lceil \log \OptCCVal\rceil$:
	\begin{enumerate}
		\item Construct an $2^i$-neighborhood 
		cover 
		$\mathcal{N}_i=\{S_{i,1},\ldots, S_{i,r_i}\}$.
		\item For each $j \in [r_i]$ run $\GraphCover$ on $G[S_{i,j}]$ to get 
		$\mathcal{C}_{i,j}$. 
	\end{enumerate}
	\item Output $\mathcal{C}^*=\bigcup_{i \in [\lceil \log \OptCCVal\rceil],j 
	\in [r]}
	\mathcal{C}_{i,j}$.
\end{enumerate}
\end{boxedminipage}
\caption{Description of the nearly optimal algorithm with respect to each edge.}
\label{fig:algorithm-opt-edge}
\end{figure}
See  \Cref{sec:distributedconstcovers} for the distributed construction of universally optimal cycle covers in the \congest\ model.

\subsection{Two-Edge Disjoint Cycle Covers}\label{app:threeedgedisjoint}
%
 %
A $(\dilation,\congestion)$-two-edge disjoint cycle cover $\mathcal{C}$ is a collection of cycles such that each edge appears on at least two \emph{edge disjoint} cycles, each cycle is of length at most $\dilation$ and each appears on at most $\congestion$ cycles. Using our cycle cover theorem, in this section we show the following generalization;
\begingroup
\def\thetheorem{\ref{lem:twoedgecyclecover}}
\begin{theorem}[Rephrased]
	Let $G$ be a $3$-edge connected graph, then there exists a construction of $(d,c)$ two-edge disjoint cycle cover $\mathcal{C}$ with $d=\widetilde{O}(D^3)$ and $c=\widetilde{O}(D^2)$.
\end{theorem}
\addtocounter{theorem}{-1}
\endgroup
\begin{proof}
The construction is based on the sampling approach \cite{dinitz2011fault}. The algorithm consists of $O(D^2 \log n)$ iterations or independent \emph{experiments}. In each experiment $i$, we sample each edge $e \in E(G)$ into $G_i$ with probability $p=(1-1/3D)$. We then apply Algorithm $\OptimalEdgeCycleCover$ (see \Cref{fig:algorithm-opt-edge}) in the graph $G_i$. The output of this algorithm is a cycle collection $\mathcal{C}_i$ that covers each edge $e \in G_i$ by a cycle of length at most $O(\log n \cdot |C_e|)$ where $C_e$ is the shortest cycle in $G_i$ that covers $e$ (if such exists). The final cycle collection is $\mathcal{C}=\bigcup_i \mathcal{C}_i$. 

Next,  for each edge $e=(u,v)$, we define the subgraph $G_e=\{C \in  \mathcal{C} ~\mid~ e \in C\}$ containing all cycles in $\mathcal{C}$ that cover $e$ in all these experiments\footnote{It is in fact sufficient to pick from each $\mathcal{C}_i$, the shortest cycle that covers $e$, for every $i$.}. The $3$-edge disjoint cycles between $u$ and $v$ are obtained by computing max-flow between $u$ and $v$ in $G_e$. 
	
We next prove the correctness of this procedure and begin by showing the w.h.p.\ the $u$-$v$ cut in $G_e$ is at least $3$ for every $e=(u,v)$. This would imply by Menger theorem that the max-flow computation indeed finds $3$ edge disjoint paths between $u$ and $v$. To prove this claim we show that for every pair of two edge $e_1,e_2 \in E$, $G_e \setminus \{e_1,e_2\}$ contains a $u$-$v$ path of length at most $O(D\log n)$  (hence, the min-cut between $u$ and $v$ is at least $3$). 

Fix such a triplet $\langle e,e_1,e_2 \rangle$ and an experiment $i$. 
We will bound the probability of the following event $\mathcal{E}_i$: $G_i$ does not contain $e_1,e_2$ but contains \emph{all} the edges on $P \cup \{e\}$, where $P$ is the $u$-$v$ shortest path in $G\setminus\{e,e_1,e_2\}$. Since all edges are sampled independently into $G_i$ with probability $p$, the probability that $\mathcal{E}_i$ happens  is $p^{|P|+1}\cdot (1-p)^2\leq p^{3D+1} \cdot 1/9D^2$. Hence, w.h.p., there exists an experiment $j$ in which the event $\mathcal{E}_j$ holds. Since $e=(u,v)$ is covered by a path of length $O(D)$ in $G_j$ and $G_j$ does not contain $e_1,e_2$, we get that the cycle $C'_e$ that covers $e=(u,v)$ in $\mathcal{C}_j$ is of length $O(D\log n)$. Overall, the path $C'_e \setminus \{e\}$ is free from $\{e_1,e_2\}$. Since the $u$-$v$ cut in $G_e$ is at least $3$, by Menger theorem we get that $G_e$ contains $3$ edge disjoint $u$-$v$ paths of length at most $|V(G_e)|=\widetilde{O}(D^3)$.

We next turn to consider the congestion. Since the final cover is a union of $O(D^2\log n)$ cycle covers, and  each individual cover $\mathcal{C}_i$ has congestion of $\widetilde{O}(1)$, the total congestion is bounded by $\widetilde{O}(D^2)$. 
%
%
\end{proof}
By the proof of \Cref{lem:twoedgecyclecover}, the construction of two-edge disjoint cycle covers is reduced to $\widetilde{O}(D^2)$ applications of cycle cover constructions. Using the distributed construction of cycle cover of \Cref{sec:distributedconstcovers}, we get an $\widetilde{O}(n\cdot D^2)$ algorithm for constructing  the two-edge disjoint covers.
\section{Resilient Distributed Computation}\label{sec:applications}
Our study of low congestion cycle cover is motivated by applications to distributed computing. In this section, we describe two applications to resilient distributed computation that use the framework of our cycle covers. Both applications provide compilers (or simulation) for distributed algorithms in the standard \congest\ model. In this model, each node can send a message of size $O(\log n)$ to 
each of its neighbors in each rounds (see \Cref{def:com-model} for the full definition). The first compiler transforms any algorithm to be resilient to Byzantine faults. The second one compiles any algorithm to be secure against an eavesdropper.

\subsection{Byzantine Faults}
\paragraph{The Model.}
We consider an adversary that can maliciously modify messages sent over the edges of the graph. The adversary is allowed to do the following.
In each round, he picks a single message $M_e$ passed on the edge $e \in G$ and corrupts it in an arbitrary manner (\ie modifying the sent message, or even completely dropping the message). The recipient of the corrupted message is not notified of the corruption. The adversary is assumed to know the inputs to all the nodes, and the entire history of the communications up to the present. 
It then picks which edge to corrupt adaptively using this information.

The goal is to compile any distributed algorithm $\cA$ into an resilient algorithm $\cA'$. The compiled algorithm $\cA'$ has the exact same output as $\cA$ for all nodes even in the presence of such an adversary. The compiler works round-by-round, and after compiling round $i$ of algorithm $\cA$, all nodes will be able to recover the original messages sent in algorithm $\cA$ in round $i$.

Our compiler assumes a preprocessing phase of the graph, which is \emph{fault-free}, in which the cycle covers are computed and are given in a distributed manner. Alternatively, if the topology of the network $G$ is
known to all nodes as assumed in many previous works, then there is no need for the preprocessing phase. For instance, in  \cite{hoza2016adversarial} it assumed that nodes know the entire graph, which allows the nodes to route messages over a sparser spanning subgraph.

\paragraph{Preprocessing.}
The preprocessing phase computes a $(\dilation_1,\congestion_1)$-cycle covers $\mathcal{C}_1$ and a $(\dilation_2,\congestion_2)$-two-edge disjoint variant $\mathcal{C}_2$, see \Cref{app:threeedgedisjoint}. These covers are known in a distributed manner, where each edge $e=(u,v)$ knows the cycles that cover it, and the other cycles that go through it. 

\paragraph{The Compiler.}
To simplify the presentation, we first describe the compiler under the assumption that the bandwidth on each edge is $\widetilde{O}(\congestion_2)$. We then reduce the bandwidth to $O(\congestion_1+\dilation_1)$, and finally present the final compiler with bandwidth of $O(\log n)$ (i.e., the standard \congest\ model). 
Such transformations are usually straightforward in the fault-free setting (e.g., by simply blowing-up every round by a factor of $\widetilde{O}(\congestion_2)$ rounds). In our setting, it becomes quite tricky. To see this, assume that there are two messages $M_1,M_2$ that are sent in the \emph{same} round in the large bandwidth protocol. In such a case, the adversary can corrupt only one of these messages. When applying a scheduler to reduce the bandwidth, the messages $M_1$ and $M_2$ might get sent in different rounds, allowing the adversary to corrupt them both! 

Throughout, we fix a round $i$ in algorithm $\cA$, and for each edge $e=(u,v)$, let $M_e$ be the message sent on $e$ in round $i$ of algorithm $\cA$. 

\paragraph{Warming up, Compiler (A) with Bandwidth $\widetilde{O}(\congestion_2)$:} The compiler works by exploiting the three edge-disjoints paths between every neighboring pair $u$ and $v$, as provided by the $(\dilation_2,\congestion_2)$ two-edge disjoint cycle cover. 

Specifically, each of the $M_e$ messages, for every $e\in G$, is sent along the three edge-disjoint $u$-$v$ paths, repeatedly for $\ell$ rounds in a pipeline manner, where $\ell=4\dilation_2$. That is, for $\ell$ rounds, the node $u$ repeatedly sends the message $M_e$ to $v$ via the three edge disjoint paths. Each intermediate node on these paths simply forwards the message it has received to its successor on that path.
The endpoint $v$ computes the message $M_e$ by taking the majority of the messages obtained in these $\ell$ rounds. 

We claim that the majority message $M'_e$ recovered by each $v$ is the correct message $M_e$. Let $a_1 \leq a_2$ be the lengths of the two edge-disjoint paths connecting $u$ and $v$ (in addition to the edge $(u,v)$). By definition, $a_1,a_2 \leq \dilation_2$. The endpoint $v$ received in total $3\ell-a_1-a_2$ messages from $u$ during this phase: $\ell$ messages are received from the direct edge $(u,v)$, $\ell-a_1$ messages received on the second $u$-$v$ path, and $\ell-a_2$ messages on the third $u$-$v$ path.

Since the adversary can corrupt at most one message per round, and since the paths are edge-disjoint, in a given round the adversary could corrupt at most one message sent on the three edge-disjoint paths. Hence, in $\ell$ rounds, the adversary could corrupt at most $\ell$ of the received messages in total. Thus, the fraction of uncorrupted messages is at least
$$
\frac{2\ell-a_1-a_2}{3\ell-a_1-a_2} > \frac{6 \dilation_2}{12\dilation_2-3} > 1/2.
$$
Therefore, the strict majority of messages received by $v$ which establishes the correctness of the compiler. Notice that each edge can get at most $\congestion_2$ messages in a given round, since it appears on  $\congestion_2$ many paths. Since the edge bandwidth is $\widetilde{O}(\congestion_2)$, all these messages can go through in a single round. Round $i$ is then complied within $O(\dilation_2)$ rounds.

\paragraph{Intermediate Compiler (B), Bandwidth $O(\dilation_1+\congestion_1)$:}
We will have two phases. In the first phase, all but $O(\dilation_1)$ of the messages will be correctly recovered. The second phase will take care of these remaining messages using the ideas of compiler (A).

%
The first phase contains two subphases, each of  $\dilation_1$ rounds. 
In the first subphase, each node $u$ sends the message $M_e$ along the edge $(u,v)$ in each of these rounds. In addition $u$ sends $M_e$ along the path $C_e\setminus \{e\}$, where $C_e$ is the cycle that covers $e$ in the cycle cover $\mathcal{C}_1$. At the end of these $\dilation_1$ rounds, $v$ should receive $\dilation_1+1$ messages. 
Observe that the adversary cannot modify \emph{all} the messages received by $v$: if he modifies the single message sent on the path $C_e\setminus \{e\}$, then he cannot modify one of the messages sent directly on the edge $(u,v)$. If $v$ received from $u$ a collection of $\dilation_1+1$ identical messages $M'$, then 
it is assured that this is the correct message and $M'=M_e$. In the complementary case, the message $M_e$ is considered to be \emph{suspicious}, and it will be handled in the second phase. 

The key point is that while some messages cannot be recovered, almost all messages will be transmitted successfully. Since there are only $\dilation_1$ rounds, the adversary can corrupt at most $\dilation_1$ messages. Since each edge $e$ appears on at most $\congestion_1$ many cycles, in a given round it can receive $\congestion_1$ many messages, and since the bandwidth of the edge is $O(\dilation_1+\congestion_1)$, all these messages can be sent in a \emph{single} round. 

The second subphase of the first round is devoted for feedback:
the receiving nodes $v$ notify their neighbors $u$ whether they successfully received their message in the first subphase. This is done in a similar manner to the first subphase. I.e., the feedback message from $v$ to $u$ is sent directly on the $e$ edge, in each of the $\dilation_1$ rounds. In addition, it is sent once along the edge-disjoint path $C_e \setminus \{e\}$. 
Only an endpoint $u$ that has received $\dilation_1+1$ positive acknowledgment messages from its neighbor $u$ can be assured that the $M_e$ message has been received successfully. 

Overall, at the end of this phase, we are left with only $2\dilation_1$ 
suspicious messages to be handled.  Importantly, the senders endpoints of these 
suspicious messages are aware of that fact (based on the above feedback 
procedure), and will become active in the second phase that is described next.

The second phase applies compiler (A) but only for a subset of $O(\dilation_1)$ many messages. This allows us to use an improved bandwidth of $O(\dilation_1)$. The number of rounds of the second phase is still bounded by $O(\dilation_2)$.

\paragraph{Final Compiler (C), Bandwidth $O(\log n)$:}
Throughout, we distinguish between two types for sending a message $M_e$: the direct type where $M_e$ is sent along the edge $e=(u,v)$, and the indirect type where $M_e$ is sent along a $u$-$v$ path (which is not $e$). 

We will have two phases as in compiler (B). In the first phase, the messages $M_e$ are sent directly on the $e$ edges in \emph{every} round of the phase. To route the indirect messages along the cycles of the $(\dilation_1,\congestion_1)$ cycle cover, we apply the random delay approach of \cite{leighton1994packet}, which takes $O(\dilation_1 + \congestion_1)$ rounds. 
%
Since the direct messages are sent in each of the rounds of the first phase, no matter how the other messages are sent, the adversary still cannot modify all the $M_e$ messages for a given pair $u,v$. As a result, at the end of the first phase (including the feedback procedure), we are left with $T=O(\dilation_1 + \congestion_1)$ suspicious messages, that would be handled in the second phase.

Implementing the second phase with the random delay approach is too risky.
The reason is that the correctness of the second phase is based on having the 
correct majority for each edge $M_e$. Thus, we have to make sure that each 
message $M_e$ is sent on the \emph{three} edge disjoint paths in the 
\emph{exact} same round. Obtaining this coordination is not so trivial. To see 
why it is crucial, consider a scenario where the scheduler sends the message 
$M_e$ along the path $P_1$ in round $j$, and along the path $P_2$ in round 
$j+1$, where $P_1$ and $P_2$ are the two edge disjoint paths between $u$ and 
$v$ (aside from $e$). In such a case, the adversary can in fact \emph{corrupt} 
both of this messages, and our majority approach will fail. 

To handle this, we use the fact that there are only $T$ 
suspicious 
messages to be sent and we handle them in the second phase one by one. To 
define the order in which these messages are handled, the sender endpoint of 
each suspicious message picks a random ID in $[1,T^2]$. Since $T=\Omega(\log 
n)$, each suspicious message gets a unique ID, with high probability. 
We will now have $T^2$ many subphases, each consists of $\ell=4\dilation_2$ many rounds. 
The $i^{th}$ subphase will take care of the suspicious message whose random ID is $i$. 
In each subphase, we implement compiler (A) with the imporved bandwidth of $O(\log n)$, as each subphase takes care of (at most) one message $M_e$. 

To make sure that intermediate vertices along the edge disjoint paths will know where to route the message, we add the information on the
source and destination $u,v$ to each of the sent messages $M_e$. We note that 
since the adversary might create fake messages by its own, it might be the case 
that vertices receive messages even if they are not the true recipient of the 
suspicious message that is handled in this particular subphase. 
Indeed, the endpoint $v$ does not know the random ID of the message $M_e$ (as 
this was chosen by its neighbor $u$) and thus it does not know when to expect 
the messages from $u$ to arrive. However, by the same majority argument, a 
vertex should receive a majority of messages in at most one subphase -- and 
every vertex that receives a majority of messages (which are identical) can 
indeed deduce that this is the correct message. 

Overall the round complexity of the first phase is $O(\dilation_1+\congestion_1)$ and of the second phase is $O((\dilation_1+\congestion_1)^2\cdot \dilation_2)$. This completes the proof of \Cref{thm:bcomp}.

\subsection{Eavesdropping}
\paragraph{Model.}
In this setting, we consider an adversary that on each round eavesdrops on one of the graph edges chosen in an arbitrary manner. Our goal is to prevent the adversary from \emph{learning} anything, in the information-theoretic sense, on any of the messages sent throughout the protocol.
We show how to use the low congestion cycle-cover to provide a compiler that can take any $r$-round distributed algorithm $\cA$ in the \congest\ model and turn it into another algorithm $\cA'$ that is secure against an eavesdropper, while incurring an overhead of $\widetilde{O}(D)$ in the round complexity. We note that if we settle for computational assumptions then there is a simple solution. One can encrypt the message using a public-key encryption scheme and then send the encrypted message using the public-key of the destination node. Thus, the main goal is to achieve unconditional security.

\paragraph{Preprocessing.}
The preprocessing phase computes a $(\dilation_1,\congestion_1)$-cycle cover. At the end of the phase, each node knows the cycles in participates in.

\paragraph{Our Compiler.}
The compiler works round-by-round. Fix a round $i$ in algorithm $\cA$, such a round is fully specified by the collection of messages sent on the edges at this round. Consider an edge $e=(u,v)$ and let $M=M_e$ be the message sent on $e$ in round this round. The secure algorithm simulates round $i$ within $\dilation_1$ rounds. At the end of these $\dilation_1$ rounds, $v$ will receive the message $M$ while the eavesdropper learns noting on $M$. 

The sender $u$ secret shares the message $M$ to $\dilation_1+1$ random shares $M_1,\ldots,M_{d+1}$ such that $M_1 \oplus \dots \oplus M_{d+1} = M$ (see \Cref{def:secret-sharing}). The first $\dilation_1$ shares of the message, namely $M_1,\ldots, M_{\ell}$, will be sent on the direct $(u,v)$ edge, in each of the rounds of phase $i$, and the $(\dilation_1+1)^{th}$ share is sent via the $u$-$v$ path $C_e \setminus \{e\}$. At the end of these $\dilation_1$ rounds, $v$ receives $d+1$ messages.

We next claim that the adversary did not learn anything (in the information-theoretic sense) about the message $M$, for any edge $(u,v)$. To show this, it suffices to show that the adversary learns at most $\dilation_1$ shares out of the total $\dilation_1+1$ shares of the message $M$. First consider the case that there is a round $j$ (in phase $i$) where the adversary did not eavesdropping on the edge $(u,v)$. In such a case, it doesn't know the $j^{th}$ share of $M$ and hence cannot know $M$. Otherwise, the adversary was eavesdropping the edge $(u,v)$ during the entire phase. This implies that it did not eavesdrop on none of the edges of $C_e \setminus \{e\}$ and hence did not learn the $(\dilation_1+1)^{th}$ share $M_{d+1}$.

The total number of rounds is $O(\dilation_1)$ using a bandwidth of $\congestion_1$. Next, we show how to schedule the messages to get a bandwidth of $O(\log n)$.

\paragraph{Scheduling Messages.}
The scheduling scheme here is similar to the scheduling of the pre-phase described before. We send a direct message on the each $(u,v)$ in all rounds. Then, the additional message send via the cycle is sent using the scheduling scheme of \cite{leighton1994packet}. The adversary still one always miss at least one share, and thus security holds. The total number of rounds as a result of this scheduling scheme is $O(\dilation_1 + \congestion_1)$. This completes the proof of \Cref{thm:ecomp}.
\section{Distributed Construction for Minor-Closed 
Graphs}\label{app:optminor}
The following fact about the sparsity of minor-closed graphs is essential in our algorithm:
\begin{fact}\label{fc:minorsparse}\cite{mader1967homomorphieeigenschaften,thomason1984extremal}
	Every (non-trivial) minor-closed family of graphs has bounded density. In 
	particular, graphs with an $h$-vertex forbidden minor have at most $O(h 
	\sqrt{\log h} 
	\cdot 
	n)$ edges.
\end{fact}
%
Recall that $\OptCCVal(G)=\max_e |C_e|$ where $C_e$ is the shortest cycle containing 
$e$. We show:
\begingroup
\def\thetheorem{\ref{thm:optimal-minor-closed}}
\begin{theorem}[Rephrased]\label{lem:mccycle}
For every minor-closed graph $G$,  one can construct in 
$\widetilde{O}(\OptCCVal(G))$ rounds, an 
$(\widetilde{O}(\OptCCVal(G)),\widetilde{O}(1))$ cycle cover $\cC$ that covers 
every edge $e \in G$ that lies on some cycle in $G$. If the vertices do not 
know $\OptCCVal(G)$, then the round complexity is $\widetilde{O}(D)$ (which is 
optimal, see \Cref{fig:flower}).
\end{theorem}
\addtocounter{theorem}{-1}
\endgroup

Due to \Cref{lem:distmodcover}, it is sufficient to present an 
$\widetilde{O}(D)$-round algorithm that constructs an 
$(\widetilde{O}(D),\widetilde{O}(1))$ cycle cover $\cC'$ where $D$ is the  
diameter of $G$. Then by applying this algorithm in each cluster of a 
$\OptCCVal(G)$-neighborhood cover, we get the desired nearly optimal 
covers. 

For a minor-closed graph $G$ in family $\mathcal{F}$, let 
$\density(\mathcal{F})$ be the upper bound density of every $G' \in 
\mathcal{F}$. That is, every $n$-vertex graph $G' \in \mathcal{F}$ has 
$\density(\mathcal{F})\cdot n$ edges, where $\density(\mathcal{F})$ is a 
constant that depends only on the family $\mathcal{F}$ (and not $n$). In the 
description of the algorithm, we assume that $\density(\mathcal{F})$ is known 
to all nodes.\footnote{We observe that algorithm can be easily modified to work 
even without $\density(\mathcal{F})$ being known.}
The algorithm starts by constructing a BFS tree $T \subseteq G$ in $O(D)$ 
rounds. Let $E'$ be the subset of non-tree edges. The main procedure is 
$\NonTreeMinorClosed$ that constructs a cover $\mathcal{C}'$ for all the 
non-tree edges $E'$.
We first describe how this cover can be constructed in $\widetilde{O}(D)$ rounds. 
\subsection{Covering Non-Tree Edges}
The algorithm has $\ell=O(\log n)$ phases, in each phase $i$, it is given a subset $E'_i$ of non-tree edges to be covered and constructs a cycle collection $\mathcal{C}_i$ that covers constant fraction of the edges in $E'_i$. At the end, the collection $\mathcal{C}'=\bigcup_{i}\mathcal{C}_i$ covers all non-tree edges.

We now focus on phase $i$ and explain how to construct $\mathcal{C}_i$. 
Similar to Alg.\ $\NonTreeCover$, our approach is based on partitioning the 
vertices into blocks -- a set of nodes that have few incident edges in $E'_i$ . 
Unlike Alg.\ $\NonTreeCover$, here each block is a subtree of $T$ and every two 
blocks are \emph{vertex-disjoint}. This plays an essential role as it allows us 
to contract each block into a super-node and get a contracted graph 
$G'_i$ that is still minor-closed and hence must be sparse. The fact that the 
contracted graph is sparse would imply that most the of non-tree edges connect 
many vertices between the same block or between two blocks. This allows 
us to cover these edges efficiently. 
We now describe phase $i$ of Alg.\ $\NonTreeMinorClosed$ in details. 

\paragraph{Step (S1): Tree Decomposition into Vertex Disjoint Subtree Blocks.}
A block $B$ is a subset of vertices such that $T(B)$ is a connected subtree of $T$. For a subset of edges $E'$, the density of the block, $\deg(B,E')$, is the number of edges in $E'$ that have an endpoint in the set $B$. 
Define 
\begin{align}\label{eq:density}
\densitythreshold=16 \cdot \density(\mathcal{F})~.
\end{align}
The algorithm works from the bottom of the tree up the root, in $O(\depth(T))$ rounds. 
Let $W\colon V \to V$ be a weight function where $w(v)=\deg(v,E'_i)$. 
In round $i\geq 1$, each vertex $v$ in layer $\depth(T)-i+1$ sends to its 
parent the \emph{residual weight} of its subtree, namely, the total weight of 
all the vertices in its subtree $T(v)$, that are not yet assigned to blocks. 
Every vertex $v$ that receives the residual weight from its children does the 
following:
Let $W'(v)$ be the sum of the total residual weight plus its own weight. If $W'(v)\geq \densitythreshold$, then $v$ declares a block and down-cast the its ID (i.e., that serves as the block-ID), to all relevant descendants in its subtree. Otherwise, it passes $W'(v)$ to its parent. 
Let $\mathcal{B}_i$ the output block decomposition.

\paragraph{Step (S2): Covering Half of the Edges.}
For ease of description, we orient every non-tree edge $e=(u,v)\in E'_i$ from its higher-ID endpoint to the lower-ID endpoint. Every vertex will be responsible for its outgoing edges in $E'_i$. 
At that point, every vertex $v$ knows its block-ID. 

\paragraph{Non-tree edges inside the same block:}
All nodes exchange their block-ID with their neighbors. Nodes that are incident to non-tree edges $e$ with both endpoints at the same block, mark the fundamental cycle of these edges. That is, the cycle that covers each such edge $e$ is given by $C(e)=\pi(u,v,T)\circ e$. All these cycles are added to $\mathcal{C}_i$. 

\paragraph{Non-tree edges between different blocks:}
Each vertex $v$ in block $B$, sends to the root of its block, all its outgoing edges along with the block-IDs of each of its outgoing $E'_i$-neighbors. In the analysis, we will show that despite the fact that the density of the block might be large, this step can be done in $O(\depth(T))$ rounds.

The root of each block $B$ receives all the outgoing edges of its block vertices and the block IDs of the other endpoints. It then partitions the edges into $|\mathcal{B}_i|$ subsets $E_i(B,B')$ for
every block $B'  \in \mathcal{B}_i \setminus \{B\}$.
Fix a pair $B \neq B'$ and assume that $E_i(B,B')$ has even-size, otherwise, omit at most one edge to make it even. To cover the edges in $E_i(B,B')$, the leader arbitrarily matches these edges into pairs $\langle e,e'\rangle$ and  
notifies the matching of all the edges to the vertices in its block. For every matched pair $\langle e,e'\rangle \in E_i(B,B')$, we have $e=(u,v)$, $e'=(u',v')$ such that $u,u' \in B$ and $v,v' \in B'$. Letting $u$ be of higher ID than $u$, the vertex $u$ is responsible for that pair. First, $u$ sends a message to its endpoint $v \in B'$ and notifies it regarding its pairing with the edge $(u',v')$. The other endpoint $v'$ notifies it to the leader of its block $B'$. The cycle $C(e,e')$ covering these edges is defined by:
$$C(e,e')=\pi(u,u',T(B))\circ e' \circ \pi(v',v) \circ e~.$$

The final cycle collection $\mathcal{C}_i$ contains the set of all $C(e,e')$ cycles of each matched pair $\langle e,e'\rangle \in E_i(B,B')$ for every $B,B' \in \mathcal{B}_i$. All the matched edges are removed from $E'_i$. Note that to make the $E_i(B,B')$ sets even, the algorithm omitted at most one edge from each such sets and all these edges are precisely those that remained to be handled in the next phase, namely, the edges $E'_{i+1}$. This completes the description of phase $i$.  See \Cref{fig:minormatching} for an illustration. 
 

\begin{figure}[h!]
\begin{center}
\includegraphics[scale=0.4]{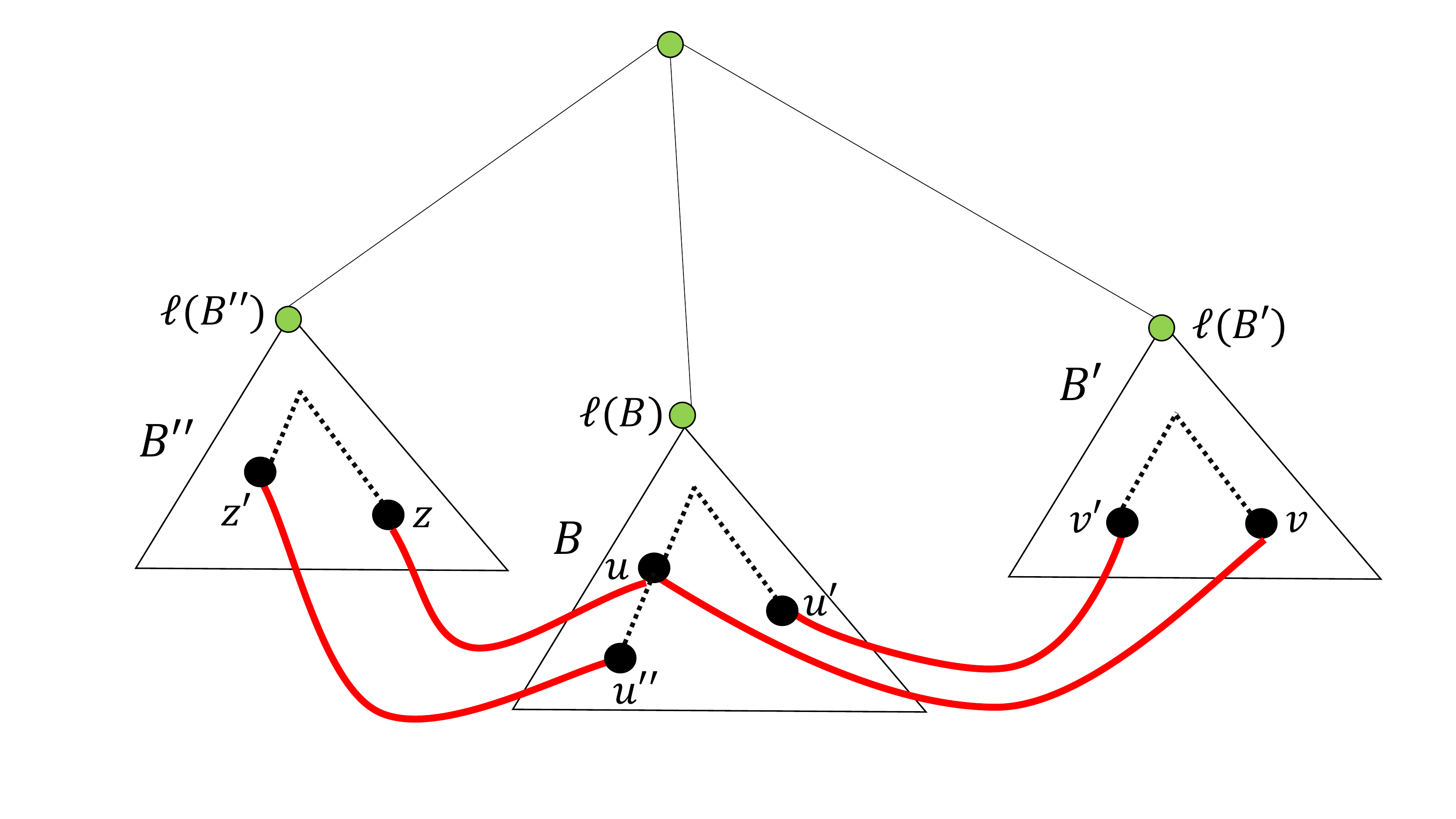}
\caption{Illustration of phase $i$ of Alg. $\NonTreeMinorClosed$. Red thick edges correspond to non-tree edges in $E'_i$. Dashed edges are internal tree paths. The edges $e=(u,v)$ and $e'=(u',v')$ are matched and their corresponding cycle $C(e,e')$ uses the tree paths in each block. Also, the edge $(u,z)$ is matched with the edge $(u'',z')$ where their cycle uses a tree path in the block $B$ as well. Overall, since the density of each block is bounded by a constant, the total congestion is also $O(1)$.
\label{fig:minormatching}}
\end{center}
\end{figure}

We proceed by analyzing Alg. $\NonTreeMinorClosed$.

\paragraph{Round Complexity and Message Complexity.}
Note that unlike the cycle cover algorithm of \cite{parteryogev17}, the blocks of Alg. $\NonTreeMinorClosed$ might have arbitrary large density. The key claim for bounding the rounding complexity and the congestion of the cycles is the following:
\begin{claim}\label{obs:keyminorclosed}
Let $e=(x,y)$ be a tree edge (where $x$ is closer to the root) and let $B$ be the block of $x$ and $y$.
Letting $B_y=B \cap T(y)$, it holds that $\deg(B_y,E'_i)\leq \densitythreshold$. 
\end{claim}
\begin{proof}
By the construction of the blocks, $\deg(B_y,E'_i)\leq \densitythreshold$ as otherwise $y$ would declare $B_y$ as a block, in contradiction that $x$ and $y$ are in the same block. 
\end{proof}

\begin{claim}\label{cl:minorcloseround}
Algorithm  $\NonTreeMinorClosed$ has round complexity of $O(\depth(T))$. 
\end{claim}
\begin{proof}
The decomposition clearly takes $O(\depth(T))$ rounds so we consider the second phase.
The algorithm starts by letting each vertex $v$ send to its root the 
information all on the edges $\deg(v,E'_i)$. We now show that this can be done 
in $O(\depth(T))$ rounds by observing that despite the fact that the total 
density of a block might be \emph{large}, the total number of messages that 
passes through a given tree edge is \emph{small}. 
Consider an edge $e=(x,y) \in T$ in block $B$, we will prove that the total number of messages that go through that edge is bounded by $O(\densitythreshold)$. 
Since all the messages that go through the edge $e$ towards the root of $B$ originated from vertices in $B_y=V(T(y))\cap B$, the above claim follows by \Cref{obs:keyminorclosed}.

All non-tree edges $e=(u,v)$ that have both endpoints in $B$ mark their fundamental cycle in $T$. By the definition of the block, this fundamental cycle is in the tree of $B$. The marking is done by sending the ID of the non-tree edge to all the edges on the fundamental cycle. This is done by letting one endpoint $u$ send the edge ID of $e=(u,v)$ to the root and back to $v$. By the same argument as above, each edge $e\in T$ receives $O(\densitythreshold)$ such messages and hence this can be done is $O(\depth(T))$ rounds.

Next, the root of each block $B$ partitions the $E'_i$ edges of its block members into $|\mathcal{B}_i|$ subsets $E_i(B,B')$ and each edge $e=(u,v)$ receives the ID of a matched edge $e'=(u',v')$ such that both $e$ and $e'$ connect vertices in the same pair of blocks. By applying the same argument only in the reverse direction, we again get that only  $O(\densitythreshold)$ messages pass on each edge\footnote{That is, an edge $e''=(x,y)\in T$ only sends information back to vertices in $B \cap V(T(y))$.}. Finally, marking the edges on all cycles $C(e,e')$ is done in $O(\depth(T))$ rounds as well, using same arguments.
\end{proof}

Since the construction of private trees employs Alg.\ $\NonTreeMinorClosed$ on 
many subgraphs of $G$ simultaneously, it is also important to bound the number 
of messages that go through a single edge $e \in G$ throughout the entire 
execution of Alg.\ $\NonTreeMinorClosed$. The next lemma essentially enables us 
to employ Alg.\ $\NonTreeMinorClosed$ on many subgraphs at once, at almost the 
round complexity as that of a single application, e.g., by 
using the random delay approach of \cite{Ghaffari15}.   

\begin{claim}\label{cl:congestionminorclose}
Alg.\ $\NonTreeMinorClosed$ passes $\widetilde{O}(1)$ messages on each edge $e 
\in G$. 
\end{claim}
\begin{proof}
Alg.\ $\NonTreeMinorClosed$ consists of $O(\log n)$ phases. We show that in 
each phase $i$, a total of $O(\log n)$ messages pass on each edge $e$. 
The first step of the phase is to decompose the tree $T$ into blocks. By 
working from leafs towards the root, on each tree edge $(u,p(u))$, $u$ sends to 
$p(u)$ the residual weight in its subtree and $p(u)$ sends to $u$ a message 
containing its block ID. Hence, overall, on each tree edge, the algorithm 
passes $O(1)$ messages. Then, each vertex sends to its neighbors its block ID 
and non-tree edges within the same block are covered by taking their 
fundamental cycle (inside the block). Every vertex $v$ sends to its block 
leader the identities of all its edges in $E'_i$ including the block-ID of the 
other endpoint. This information is passed on the subtree of each block. 
By the proof of \Cref{cl:minorcloseround}, one each edge, there are total of $O(\densitythreshold)$ messages (this bounds holds overall the $O(\depth(T))$ rounds of the algorithm). 
\end{proof}

\paragraph{Cover Analysis.} 
The analysis of Algorithm $\NonTreeMinorClosed$ exploits two properties of 
minor-closed graphs: (i) being closed under edge contraction and (ii) being 
sparse (see Fact \ref{fc:minorsparse}). 

Let $E'_{i+1}$ be set all the edges that are not covered in phase $i$. We will show that $|E'_{i+1}|\leq |E'_i|/2$. That is, we will show that at least half of the edges in $E'_i$ are covered by the cycles of $\mathcal{C}_i$. 
Consider the subgraph $G'_i=T \cup E'_{i+1}$, clearly, $G'_i \subseteq G$ is a minor-closed graph as well. We now compute a new graph $\widetilde{G}_i$ by contracting all the tree edges, $E(T)$, in $G'_i$. 
Note that this contraction is only of the sake of the analysis, and it is not part of the algorithm. 
Since the blocks correspond to vertex-disjoint trees, the resulting contracted graph has $|\mathcal{B}_i|$ nodes and all its edges correspond to the non-tree edges $E'_{i+1}$. We slightly abuse notation by denoting the super-node of block $B$ in $\widetilde{G}_i$ by $B$. 
By the explanation above, each edge in $\widetilde{G}_i$ is of multiplicity at most $2$. This is because for every pair $B,B'$, at most one edge in $E_i(B,B')$ is added to $E_{i+1}$ and also at most one edge of $E_i(B',B)$ is added to $E_{i+1}$. 
Let $\widetilde{G}'_i$ be the simple graph analogue of the contracted graph $\widetilde{G}_i$, i.e., removing multiplicities of edges. Since $\widetilde{G}_i$ is minor-close with $|\mathcal{B}_i|$ nodes, it has at most $\density(\mathcal{F}) \cdot |\mathcal{B}_i|$ edges. 
Since the weight of each block is at least $\densitythreshold$ and blocks are vertex disjoint, we have that $|\mathcal{B}_i|\leq 2E'/\densitythreshold$.
We have:
\begin{eqnarray*}
|E'_{i+1}|=|E(\widetilde{G}_i)|\leq 2\cdot |E(\widetilde{G}'_i)|
\leq 2\density(\mathcal{F})\cdot |\mathcal{B}_i| \\ \leq 
8\density(\mathcal{F})\cdot |E'_i|/\densitythreshold \leq |E'_i|/2~,
\end{eqnarray*}
where the last inequality follows by Equation \ref{eq:density}.

\paragraph{Length and Congestion Analysis.}
Clearly, all cycles of the form $C(e,e')$ or $C(e)$ have length $O(\depth(T))$. It is also easy to see, that by definition, each cycle is used to cover at most two non-tree edges. 
We now claim that each edge appears on $O(1)$ cycles of $\mathcal{C}_i$. 

Since each non-tree edge appears on at most two cycles, it is sufficient to bound the congestion on the tree edges. Let $e''=(x,y)$ be a tree edge in $T$ (where $x$ is closer to the root) and let $\mathcal{C}(e'')$ be the collection of all the cycles that go through $e''$.  Let $B$ be the unique block to which $x,y$ belong. By construction, the edge $e''$ appears only on cycles $C(e)$ or $C(e,e')$ where the edge $e$ is incident to a vertex in $B_y=T(y)\cap B$. By \Cref{obs:keyminorclosed}, $\deg(B_y,E'_i)\leq \densitythreshold$, and hence $e''$ appears on $O(\densitythreshold)$ cycles are required. 

\subsection{Covering Tree Edges}
The distributed covering of tree edges is given by Alg.\ $\DistTreeCover$. In the high level, this algorithm 
reduces the problem of 
covering tree edge to the problem of covering \emph{non-tree} edges at the cost 
of $O(\depth(T))$ rounds. Hence, by applying the same reduction to the non-tree 
setting and using Alg.\ $\NonTreeMinorClosed$, we will get an 
$(\widetilde{O}(D),\widetilde{O}(1))$ cycle cover $\cC''$ for the tree edges of 
$T$. Also here, each cycle $C \in \cC''$ might be used to 
cover $O(D)$ tree edges. 

\paragraph{Description of Algorithm $\DistTreeCover$}
Algorithm $\DistTreeCover$ essentially mimics the centralized construction of 
\Cref{sec:cyclecoverall}. Let $p(v)$ be the parent of $v$ in the BFS tree $T$. A non-tree edge $e'=(u',v')$ is a swap edge for the tree edge $e=(p(v),v)$ if $e \in \pi(u',v')$, let $s(v)=v'$ by the endpoint of $e'$ that is not in $T(v)$. 
By using the algorithm of Section 4.1 in \cite{ghaffari2016near}, we can make 
every node $v$ know $s(v)$ in $O(D)$ rounds.

A key part in the algorithm of \Cref{sec:cover-tree-edges} is the definition of the 
path $P_{e}=\pi(v,u')\circ (u',s(v))$ for every tree edge $e=(p(v),v)$. By 
computing swap edges using Section 4.1 in \cite{ghaffari2016near} all the edges 
of each $P_{e}$ get marked. 

\paragraph{Computing the set $I(T) \subseteq E(T)$.}
We next describe how to compute a maximal collection of tree edges $I=\{e_i\}$ 
whose paths $P_{e_i}$ are edge disjoint and in addition for each edge $e_j \in 
E(T)\setminus I$ there exists an edge $e_i \in T'$ such that $e_j \in P_{e_i}$. 
To achieve this, we start working on the root towards the leaf. In every round 
$i\in \{1,\ldots,D\}$, we consider only \emph{active} edges in layer $i$ in 
$T$.  Initially, all edges are active. An edge becomes inactive in a given 
round if it receives an inactivation message in any previous round. Each active 
edge in layer $i$, say $e_j$, initiates an inactivation message on its path 
$P_{e_j}$. An inactivation message of an edge $e_j$ propagates on the path 
$P_{e_j}$ round by round, making all the corresponding edges on it to become 
inactive.

Note that the paths $P_{e_{j}}$ and $P_{e_{j'}}$ for two edges $e_j$ 
and $e_{j'}$ in the same layer of the BFS tree, are edge disjoint and hence 
inactivation messages from different edges on the same layer do not interfere 
each other.  We get that an edge in layer $i$ active in round $i$ only if it 
did not receive any prior inactivation message from any of its BFS ancestors. 
In addition, any edge that receives an inactivation message necessarily appears 
on a path of an active edge. It is easy to see that within $D$ rounds, all 
active edges $I$ on $T$ satisfy the desired properties (\ie their $P_{e_{i}}$ 
paths cover the remaining $T$ edges and these paths are edge disjoint).
 
\paragraph{Distributed Implementation of Algorithm $\TreeCover$.}
First, we mark all the edges on the $P_e$ paths for every $e \in I(T)$. As 
every node $v$ with $e=(p(v),v)$ know its swap edge, it can send information 
along $P_e$ and mark the edges on the path. Since each edge appears on the most 
two $P_e$ paths, this can be done simultaneously for all $e \in I(T)$. 

From this point on we follow the steps of Algorithm $\TreeCover$. The 
partitioning of \Cref{clm:tree-partition} can be done in $O(D)$ rounds as it 
only required nodes to count the number of nodes in their subtree. We define 
the ID of each tree $T'_1,T'_2$ to be the maximum edge ID in the tree (as the 
trees are edge disjoint, this is indeed an identifier for the tree). 
By passing information on the $P_e$ paths, each node $v$ can learn the tree ID of its swap endpoint $s(v)$. This allows to partition the edges of $T'$ into $E'_{x,y}$ for $x,y \in \{1,2\}$. 
Consider now the $i^{th}$ phase in the computation of cycle cover $\cC_{1,2}$ for the edges $E'_{1,2}$.

Applying Algorithm $\EdgeDisjointPath$ can be done in $O(D)$ round. At the end, 
each node $v_j$ knows its matched pair $v'_j$ and the edges on the tree path 
$\pi(v_j,v'_j,T'_1)$ are marked. 
Let $\Sigma$ be the matched pairs. We now the virtual conflict graph 
$G_{\Sigma}$. Each pair 
$\langle v_j,v'_j\rangle \in \Sigma$ is simulated by the node of higher ID, say, $v_j$. We say that $v_j$ is the \emph{leader} of the pair $\langle v_j,v'_j\rangle \in \Sigma$.  
Next, each node $v$ that got matched with $v'$ activates the edges on its path $P_e \cap E(T'_1)$ for $e=(p(v),v)$. Since the $\pi$ edges of the matched pairs are marked as well, every edge $e' \in \pi(v_k,v'_k,T'_1)$ that belongs to an active path $P_e$ sends the ID of the edge $e$ to the leader of the pair $\langle v_k,v'_k\rangle$. By \Cref{lem:outdegone}, every pair $\sigma'$ interferes with at most one other pair and hence there is no congestion and a single message is sent along the edge-disjoint paths $\pi(v_j,v'_j,T'_1)$ for every $\langle v_j,v'_j\rangle\in \Sigma$. Overall, we get the the construction of the virtual graph can be done in $O(D)$ rounds. 

We next claim that all leaders of two neighboring pairs $\sigma,\sigma' \in G_{\Sigma}$ can exchange $O(\log n)$ bits of information using $O(D)$ rounds. Hence, any $r$-round algorithm for the graph $G_{\Sigma}$ can be simulated in $T'_1$ in $O(r \cdot D)$ rounds. 
To see this, consider two neighbors $\sigma=\langle x,y\rangle,\sigma'=\langle x',y'\rangle$ where $\sigma'$ interferes $\sigma$. Without loss of generality, assume that the leader $x'$ of $\sigma'$ wants to send a message to the leader $x$ of $\sigma$. First, $x'$ sends the message on the path $\pi(x',y',T'_1)$. The edge $e' \in \pi(x',y',T'_1) \cap P_e$ for $e=(p(x),x)$ that receives this message sends it to the leader $x$ along the path $P_e$. Since we only send messages along edge disjoint paths, there is no congestion and can be done in $O(D)$ rounds. 

Since the graph $G_{\Sigma}$ has arboricity $O(1)$, it can be colored with 
$O(1)$ colors and $O(\log n)$ rounds using the algorithm of 
\cite{barenboim2010sublogarithmic}. By the above, simulating this algorithm in 
$G$ takes $O(D\log n)$ rounds.
We then consider each color class at a time where at step $j$ we consider 
$\Sigma_{i,j}$. 
For every $\sigma=\langle x,y \rangle$, $x$ sends the ID of $s(y)$ to $s(x)$ along the $P_e$ path for $e=(p(x),x)$. In the same manner, $y$ sends the ID of $s(x)$ to $s(y)$. This allows each node in $T'_2$ know its virtual edge. At that point we run Algorithm $\NonTreeMinorClosed$ to cover the virtual edges. Each virtual edge is later replaced with a true path in $G$ in a straightforward manner. 

\paragraph{Analysis of Algorithm $\DistTreeCover$.}
\begin{claim}\label{cl:disttreeedge}
Algorithm $\DistTreeCover$ computes a $(\widetilde{O}(D),\widetilde{O}(1))$ cycle cover $\cC_2$ for the tree edges $E(T)$ and has round complexity of $\widetilde{O}(D)$.
\end{claim}
\begin{proof}
The correctness follows the same line of arguments as in the centralized construction (see the Analysis of \Cref{sec:cover-tree-edges}), only the here we use Algorithm $\NonTreeMinorClosed$. 
Each cycle computed by Algorithm $\NonTreeMinorClosed$ has length $\widetilde{O}(D)$ and the cycle covers $\widetilde{O}(1)$ non-tree edges. In our case, each non-tree edge is virtual and replaced by a path of length $O(D)$ hence the final cycle has still length $\widetilde{O}(D)$. With respect to congestion, we have $O(\log n)$ levels of recursion, and in each level when working on the subtree $T'$ we have $O(\log n)$ applications of Algorithm $\NonTreeMinorClosed$ which computes cycles with congestion $\widetilde{O}(1)$. The total congestion is then bounded by $\widetilde{O}(1)$. 

We proceed with round complexity.
The algorithm has $O(\log n)$ levels of recursion. In each level we work on edge disjoint trees simultaneously. Consider a tree $T'$. The partitioning into $T'_1,T'_2$ takes $O(D)$ rounds. We now have $O(\log n)$ phases. We show that each phase takes $\widetilde{O}(D)$ rounds, which is the round complexity of Algorithm $\NonTreeMinorClosed$. In particular, In phase $i$ we have the following procedures. Applying Algorithm $\EdgeDisjointPath$ in $T'_1,T'_2$ takes $O(D)$ rounds. The computation of the conflict graph $G_{\Sigma}$ takes $O(D)$ rounds as well and coloring it using the coloring algorithm for low-arboricity graphs of \cite{barenboim2010sublogarithmic} takes $O(D\log n)$ rounds. Then we apply Algorithm $\NonTreeMinorClosed$ which takes $\widetilde{O}(D)$ rounds. Translating the cycles into cycles in $G$ takes $\widetilde{O}(D)$ rounds. 
\end{proof}

%
%
%
%
Summing over all the $O(\log n)$ phases, each (tree) edge appears on $O(\log n \densitythreshold)=O(\log n)$ cycles of the final cycle collection $\mathcal{C}'=\bigcup_u \mathcal{C}_i$. 
We therefore have:
\begin{lemma}\label{lem:minorclosesumdist}
For every bridgeless minor-closed graph $G$, a tree $T \subseteq G$ of diameter $D$, there exists:
(i) a $O(\depth(T))$ round algorithm that constructs an $(O(\depth(T)),O(\log n))$ cycle collection $\cC$ that covers all non-tree edges. Each cycle in $\cC$ is used to cover at most \emph{two} non-tree edges in $E(G)\setminus E(T)$. In addition, the algorithm passes $\widetilde{O}(1)$ messages on each edge $e$ over the entire execution; (ii) an $\widetilde{O}(D)$ round algorithm that constructs an $(\widetilde{O}(D),\widetilde{O}(1))$ cycle collection $\cC$ that covers all edges in $G$.
\end{lemma}
%
%
%
\begin{figure}[h!]
\begin{center}
\includegraphics[scale=0.35]{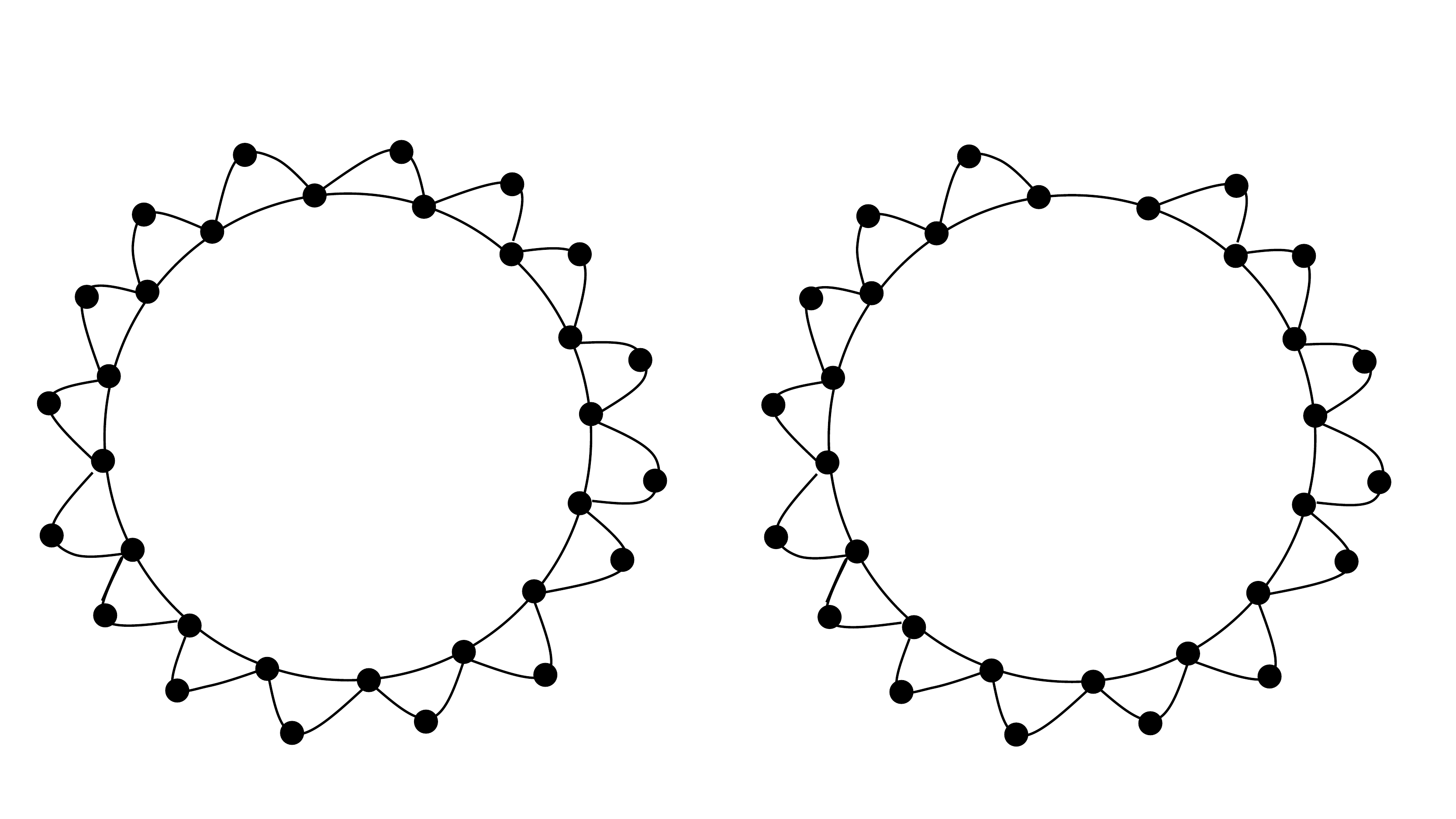}
\caption{Let $G$ be the left graph, then $\OptCCVal(G)=3$, and let $G'$ be the 
right graph then $\OptCCVal(G)=n$. Without knowledge the value of 
$\OptCCVal(G)$ a vertex that is at 
$n/2$ distance from the missing edges on $G'$, cannot distinguish in $n/100$ 
rounds if it is in $G$ or $G'$.}
\label{fig:flower}
\end{center}
\end{figure}

\bibliographystyle{alpha}
\bibliography{crypto}

\newcommand{\etalchar}[1]{$^{#1}$}
\begin{thebibliography}{BOGW88}

\bibitem[ABCP96]{awerbuch1996fast}
Baruch Awerbuch, Bonnie Berger, Lenore Cowen, and David Peleg.
\newblock Fast distributed network decompositions and covers.
\newblock {\em Journal of Parallel and Distributed Computing}, 39(2):105--114,
  1996.

\bibitem[ABCP98]{awerbuch1998near}
Baruch Awerbuch, Bonnie Berger, Lenore Cowen, and David Peleg.
\newblock Near-linear time construction of sparse neighborhood covers.
\newblock {\em SIAM Journal on Computing}, 28(1):263--277, 1998.

\bibitem[Awe85]{awerbuch1985complexity}
Baruch Awerbuch.
\newblock Complexity of network synchronization.
\newblock {\em Journal of the ACM (JACM)}, 32(4):804--823, 1985.

\bibitem[BDP97]{berman1997reliable}
Piotr Berman, Krzysztof Diks, and Andrzej Pelc.
\newblock Reliable broadcasting in logarithmic time with byzantine link
  failures.
\newblock {\em Journal of Algorithms}, 22(2):199--211, 1997.

\bibitem[BE10]{barenboim2010sublogarithmic}
Leonid Barenboim and Michael Elkin.
\newblock Sublogarithmic distributed mis algorithm for sparse graphs using
  nash-williams decomposition.
\newblock {\em Distributed Computing}, 22(5-6):363--379, 2010.

\bibitem[BEH{\etalchar{+}}10]{baier2010length}
Georg Baier, Thomas Erlebach, Alexander Hall, Ekkehard K{\"o}hler, Petr Kolman,
  Ond{\v{r}}ej Pangr{\'a}c, Heiko Schilling, and Martin Skutella.
\newblock Length-bounded cuts and flows.
\newblock {\em ACM Transactions on Algorithms (TALG)}, 7(1):4, 2010.

\bibitem[BH94]{bagchi1994information}
Anindo Bagchi and S.~Louis Hakimi.
\newblock Information dissemination in distributed systems with faulty units.
\newblock {\em IEEE Transactions on Computers}, 43(6):698--710, 1994.

\bibitem[BM05]{blaser2005approximating}
Markus Bl{\"a}ser and Bodo Manthey.
\newblock Approximating maximum weight cycle covers in directed graphs with
  weights zero and one.
\newblock {\em Algorithmica}, 42(2):121--139, 2005.

\bibitem[BOGW88]{ben1988completeness}
Michael Ben-Or, Shafi Goldwasser, and Avi Wigderson.
\newblock Completeness theorems for non-cryptographic fault-tolerant
  distributed computation.
\newblock In {\em Proceedings of the twentieth annual ACM symposium on Theory
  of computing}, pages 1--10. ACM, 1988.

\bibitem[Bol04]{bollobas2004extremal}
B{\'e}la Bollob{\'a}s.
\newblock {\em Extremal graph theory}.
\newblock Courier Corporation, 2004.

\bibitem[BP93]{blough1993optimal}
Douglas~M Blough and Andrzej Pelc.
\newblock Optimal communication in networks with randomly distributed byzantine
  faults.
\newblock {\em Networks}, 23(8):691--701, 1993.

\bibitem[CHGH18]{censor2018making}
Keren Censor-Hillel, Ran Gelles, and Bernhard Haeupler.
\newblock Making asynchronous distributed computations robust to channel noise.
\newblock In {\em LIPIcs-Leibniz International Proceedings in Informatics},
  volume~94. Schloss Dagstuhl-Leibniz-Zentrum fuer Informatik, 2018.

\bibitem[CHT17]{censor2017fast}
Keren Censor-Hillel and Tariq Toukan.
\newblock On fast and robust information spreading in the vertex-congest model.
\newblock {\em Theoretical Computer Science}, 2017.

\bibitem[DK11]{dinitz2011fault}
Michael Dinitz and Robert Krauthgamer.
\newblock Fault-tolerant spanners: better and simpler.
\newblock In {\em Proceedings of the 30th annual ACM SIGACT-SIGOPS symposium on
  Principles of distributed computing}, pages 169--178. ACM, 2011.

\bibitem[DPPU88]{DworkPPU88}
Cynthia Dwork, David Peleg, Nicholas Pippenger, and Eli Upfal.
\newblock Fault tolerance in networks of bounded degree.
\newblock {\em {SIAM} J. Comput.}, 17(5):975--988, 1988.

\bibitem[EJ73]{edmonds1973matching}
Jack Edmonds and Ellis~L Johnson.
\newblock Matching, euler tours and the chinese postman.
\newblock {\em Mathematical programming}, 5(1):88--124, 1973.

\bibitem[EN17]{elkin2017efficient}
Michael Elkin and Ofer Neiman.
\newblock Efficient algorithms for constructing very sparse spanners and
  emulators.
\newblock In {\em Proceedings of the Twenty-Eighth Annual ACM-SIAM Symposium on
  Discrete Algorithms}, pages 652--669. Society for Industrial and Applied
  Mathematics, 2017.

\bibitem[Fan92]{fan1992integer}
Genghua Fan.
\newblock Integer flows and cycle covers.
\newblock {\em Journal of Combinatorial Theory, Series B}, 54(1):113--122,
  1992.

\bibitem[Fis83]{fischer1983consensus}
Michael~J Fischer.
\newblock The consensus problem in unreliable distributed systems (a brief
  survey).
\newblock In {\em International Conference on Fundamentals of Computation
  Theory}, pages 127--140. Springer, 1983.

\bibitem[FLP85]{fischer1985impossibility}
Michael~J Fischer, Nancy~A Lynch, and Michael~S Paterson.
\newblock Impossibility of distributed consensus with one faulty process.
\newblock {\em Journal of the ACM (JACM)}, 32(2):374--382, 1985.

\bibitem[G{\"a}r99]{gartner1999fundamentals}
Felix~C G{\"a}rtner.
\newblock Fundamentals of fault-tolerant distributed computing in asynchronous
  environments.
\newblock {\em ACM Computing Surveys (CSUR)}, 31(1):1--26, 1999.

\bibitem[Gel17]{Gelles17}
Ran Gelles.
\newblock Coding for interactive communication: {A} survey.
\newblock {\em Foundations and Trends in Theoretical Computer Science},
  13(1-2):1--157, 2017.

\bibitem[GH16]{ghaffari2016distributed}
Mohsen Ghaffari and Bernhard Haeupler.
\newblock Distributed algorithms for planar networks ii: Low-congestion
  shortcuts, {MST}, and min-cut.
\newblock In {\em Proceedings of the twenty-seventh annual ACM-SIAM symposium
  on Discrete algorithms}, pages 202--219. SIAM, 2016.

\bibitem[Gha15a]{ghaffari2015distributed}
Mohsen Ghaffari.
\newblock Distributed broadcast revisited: Towards universal optimality.
\newblock In {\em International Colloquium on Automata, Languages, and
  Programming}, pages 638--649. Springer, 2015.

\bibitem[Gha15b]{Ghaffari15}
Mohsen Ghaffari.
\newblock Near-optimal scheduling of distributed algorithms.
\newblock In {\em Proceedings of the 2015 {ACM} Symposium on Principles of
  Distributed Computing, {PODC}}, pages 3--12, 2015.

\bibitem[GP16]{ghaffari2016near}
Mohsen Ghaffari and Merav Parter.
\newblock Near-optimal distributed algorithms for fault-tolerant tree
  structures.
\newblock In {\em Proceedings of the 28th ACM Symposium on Parallelism in
  Algorithms and Architectures}, pages 387--396. ACM, 2016.

\bibitem[GP17]{ghaffari2017near}
Mohsen Ghaffari and Merav Parter.
\newblock Near-optimal distributed dfs in planar graphs.
\newblock In {\em 31st International Symposium on Distributed Computing (DISC
  2017)}, volume~91, page~21. Schloss Dagstuhl-Leibniz-Zentrum fuer Informatik,
  2017.

\bibitem[Gua62]{guan1962graphic}
Meigu Guan.
\newblock Graphic programming using odd and even points.
\newblock {\em Chinese Math.}, 1:237--277, 1962.

\bibitem[HHW18]{haeupler2018round}
Bernhard Haeupler, D~Ellis Hershkowitz, and David Wajc.
\newblock Round-and message-optimal distributed part-wise aggregation.
\newblock {\em arXiv preprint arXiv:1801.05127}, 2018.

\bibitem[HIZ16a]{haeupler2016low}
Bernhard Haeupler, Taisuke Izumi, and Goran Zuzic.
\newblock Low-congestion shortcuts without embedding.
\newblock In {\em Proceedings of the 2016 ACM Symposium on Principles of
  Distributed Computing}, pages 451--460. ACM, 2016.

\bibitem[HIZ16b]{haeupler2016near}
Bernhard Haeupler, Taisuke Izumi, and Goran Zuzic.
\newblock Near-optimal low-congestion shortcuts on bounded parameter graphs.
\newblock In {\em International Symposium on Distributed Computing}, pages
  158--172. Springer, 2016.

\bibitem[HL18]{haeupler2018faster}
Bernhard Haeupler and Jason Li.
\newblock Faster distributed shortest path approximations via shortcuts.
\newblock {\em arXiv preprint arXiv:1802.03671}, 2018.

\bibitem[HLZ18]{haeupler2018minor}
Bernhard Haeupler, Jason Li, and Goran Zuzic.
\newblock Minor excluded network families admit fast distributed algorithms.
\newblock {\em arXiv preprint arXiv:1801.06237}, 2018.

\bibitem[HO01]{hochbaum2001bounded}
Dorit~S Hochbaum and Eli~V Olinick.
\newblock The bounded cycle-cover problem.
\newblock {\em INFORMS Journal on Computing}, 13(2):104--119, 2001.

\bibitem[HS16]{hoza2016adversarial}
William~M Hoza and Leonard~J Schulman.
\newblock The adversarial noise threshold for distributed protocols.
\newblock In {\em Proceedings of the twenty-seventh annual ACM-SIAM symposium
  on Discrete algorithms}, pages 240--258. Society for Industrial and Applied
  Mathematics, 2016.

\bibitem[IMM05]{immorlica2005cycle}
Nicole Immorlica, Mohammad Mahdian, and Vahab~S Mirrokni.
\newblock Cycle cover with short cycles.
\newblock In {\em Annual Symposium on Theoretical Aspects of Computer Science},
  pages 641--653. Springer, 2005.

\bibitem[IR78]{itai1978covering}
Alon Itai and Michael Rodeh.
\newblock Covering a graph by circuits.
\newblock In {\em International Colloquium on Automata, Languages, and
  Programming}, pages 289--299. Springer, 1978.

\bibitem[KKP01]{kranakis2001fault}
Evangelos Kranakis, Danny Krizanc, and Andrzej Pelc.
\newblock Fault-tolerant broadcasting in radio networks.
\newblock {\em Journal of Algorithms}, 39(1):47--67, 2001.

\bibitem[KN16]{khachay2016approximability}
Michael Khachay and Katherine Neznakhina.
\newblock Approximability of the minimum-weight k-size cycle cover problem.
\newblock {\em Journal of Global Optimization}, 66(1):65--82, 2016.

\bibitem[KNY05]{krivelevich2005approximation}
Michael Krivelevich, Zeev Nutov, and Raphael Yuster.
\newblock Approximation algorithms for cycle packing problems.
\newblock In {\em Proceedings of the sixteenth annual ACM-SIAM symposium on
  Discrete algorithms}, pages 556--561. Society for Industrial and Applied
  Mathematics, 2005.

\bibitem[KR95]{klein1995nearly}
Philip Klein and R~Ravi.
\newblock A nearly best-possible approximation algorithm for node-weighted
  steiner trees.
\newblock {\em Journal of Algorithms}, 19(1):104--115, 1995.

\bibitem[KR01]{keidar2001cost}
Idit Keidar and Sergio Rajsbaum.
\newblock On the cost of fault-tolerant consensus when there are no faults:
  preliminary version.
\newblock {\em ACM SIGACT News}, 32(2):45--63, 2001.

\bibitem[Li18]{li2018distributed}
Jason Li.
\newblock Distributed treewidth computation.
\newblock {\em arXiv preprint arXiv:1805.10708}, 2018.

\bibitem[LMR94]{leighton1994packet}
Frank~Thomson Leighton, Bruce~M Maggs, and Satish~B Rao.
\newblock Packet routing and job-shop scheduling ino (congestion+ dilation)
  steps.
\newblock {\em Combinatorica}, 14(2):167--186, 1994.

\bibitem[LMR18]{leviplanar18}
Reut Levi, Moti Medina, and Dana Ron.
\newblock Property testing of planarity in the {CONGEST} model.
\newblock In {\em Proceedings of the 2018 {ACM} Symposium on Principles of
  Distributed Computing, {PODC} 2018, Egham, United Kingdom, July 23-27, 2018},
  pages 347--356, 2018.

\bibitem[LZKS13]{leblanc2013resilient}
Heath~J LeBlanc, Haotian Zhang, Xenofon Koutsoukos, and Shreyas Sundaram.
\newblock Resilient asymptotic consensus in robust networks.
\newblock {\em IEEE Journal on Selected Areas in Communications},
  31(4):766--781, 2013.

\bibitem[Mad67]{mader1967homomorphieeigenschaften}
Wolfgang Mader.
\newblock Homomorphieeigenschaften und mittlere kantendichte von graphen.
\newblock {\em Mathematische Annalen}, 174(4):265--268, 1967.

\bibitem[Man09]{manthey2009minimum}
Bodo Manthey.
\newblock Minimum-weight cycle covers and their approximability.
\newblock {\em Discrete Applied Mathematics}, 157(7):1470--1480, 2009.

\bibitem[Pel96]{pelc1996fault}
Andrzej Pelc.
\newblock Fault-tolerant broadcasting and gossiping in communication networks.
\newblock {\em Networks: An International Journal}, 28(3):143--156, 1996.

\bibitem[Pel00]{Peleg:2000}
David Peleg.
\newblock {\em Distributed Computing: A Locality-sensitive Approach}.
\newblock SIAM, 2000.

\bibitem[PP05]{pelc2005broadcasting}
Andrzej Pelc and David Peleg.
\newblock Broadcasting with locally bounded byzantine faults.
\newblock {\em Information Processing Letters}, 93(3):109--115, 2005.

\bibitem[PS89]{peleg1989time}
David Peleg and Alejandro~A Sch{\"a}ffer.
\newblock Time bounds on fault-tolerant broadcasting.
\newblock {\em Networks}, 19(7):803--822, 1989.

\bibitem[PY17]{parteryogev17}
Merav Parter and Eylon Yogev.
\newblock Distributed computing made secure: {A} graph theoretic approach.
\newblock {\em CoRR}, abs/1712.01139, 2017.
\newblock To Appear in SODA '19.

\bibitem[Sey79]{Seymour79}
P.~D. Seymour.
\newblock Sums of circuits.
\newblock {\em Graph theory and related topics}, pages 341–--355, 1979.

\bibitem[Sze73]{szekeres1973polyhedral}
George Szekeres.
\newblock Polyhedral decompositions of cubic graphs.
\newblock {\em Bulletin of the Australian Mathematical Society}, 8(3):367--387,
  1973.

\bibitem[Tho84]{thomason1984extremal}
Andrew Thomason.
\newblock An extremal function for contractions of graphs.
\newblock In {\em Mathematical Proceedings of the Cambridge Philosophical
  Society}, volume~95, pages 261--265. Cambridge University Press, 1984.

\bibitem[Tho97]{thomassen1997complexity}
Carsten Thomassen.
\newblock On the complexity of finding a minimum cycle cover of a graph.
\newblock {\em SIAM Journal on Computing}, 26(3):675--677, 1997.

\end{thebibliography}
\appendix
\section{Balanced Partitioning of a Tree}\label{clm:tree-partition}
We show that every rooted tree $T$ can be partitioned into two edge-disjoint 
rooted trees $T_1$ and $T_2$ such that (I) $E(T_1)\cup E(T_2)=E(T)$ and (II) 
$V(T_1), V(T_2)\leq 2/3\cdot N$ where $N=|T|$. 
In addition, this partitioning maintains the layering structure of $T$ as will 
be described later.
To compute this partitioning, define the weight $w(v)$ of each vertex $v$ in 
$T$ to be the number of vertices in its subtree $T(v)$. First, consider the 
case, where there is a vertex $v^*$ with weight $w(v^*) \in [1/3 N,2/3 N]$. 
In such a case, define $T_1=T_{v^*}$ and $T_2=T \setminus E(T(v^*))$. 
By definition, both $T_1$ and $T_2$ are trees, all edges of $T$ are covered and 
$|T_1|,|T_2|\in [1/3 N, 2/3 N]$.

Else, if no such balanced vertex exists, there must be a vertex $v^*$ such that 
$w(v^*)\geq 2/3 N$ but for each of its children in $T$, $u_i$, it holds that 
$w(u_i)\leq 1/3 N$. In such a case, we consider the children of $v^*$ from left 
to right $u_1,\ldots, u_k$ and sum up their weights until we get to a value in 
the range $[1/3N,2/3N]$. Formally, let $\ell \in \{1,\ldots,k\}$ be the minimal 
index satisfying that $\sum_{i=1}^\ell w(u_i)\in [1/3 N,2/3N]$. Since each 
$w(u_i)\leq 1/3 N$, such an index $\ell$ exists. 
We then set $T_1=\bigcup_{i=1}^\ell \left(T(u_i) \cup \{(u_i,v^*)\}\right)$ and 
$T_2=T\setminus \bigcup_{i=1}^\ell V(T(u_i))$.
By construction, all edges of $T$ are covered by $T_1$ and $T_2$. In addition, 
by definition, $|T_1|\in [1/3N,2/3 N]$ and hence also $T_2 \in [1/3 N,2/3 N]$.

Finally, we pick the roots $r_1,r_2$ of $T_1,T_2$ (respectively) to be the 
vertices the are close-most to the root $r$ in $T$. We then get for $u,v \in 
T_1$, that if $u$ is closer to the root than $v$ in $T$, then also $u$ is 
closer to the root $r_1$ than $v$ in $T_1$.

\section{Distributed Constructions for General Graphs}
\subsection{Low Congestion Covers}\label{sec:distributedconstcovers}
In this section, we show how the covers of 
\Cref{thm:cyclecover_upper,thm:cyclcoveropt} with existentially optimal bounds can be constructed using 
$\widetilde{O}(n)$ rounds in the distributed setting.
\begin{lemma}\label{lem:preproc-cover}
For every bridgeless  $n$-vertex graph a $(D\log n,\log^3 n)$ cycle 
cover can be computed distributively in $\widetilde{O}(n)$ rounds of 
pre-processing.
\end{lemma}
\begin{proof}
Compute a BFS tree $T$ and consider the set of non-tree edges $E'$. Let $E_0=E'$. As long that number of edges $E_i$ to be covered in $E'$ is at least $O(\log^c n \cdot n)$, we do as follows in phase $i$. 
Let $\Delta_i=|E_i|/n$. We partition the edges of $E_i$ into $\ell_i=\Delta_i/(c\cdot \log n)$ edge-disjoint subgraphs by letting each edge in $E_i$ pick a number in $[1,\ell_i]$ uniformly at random. We have that w.h.p. each subgraph $E_{i,j}$ contains $\Theta(n\log n)$ edges of $E_i$. 

At the point, we work on each subgraph $E_{i,j}$ independently. We compute a 
BFS tree $T_{i,j}$ in each $E_{i,j}$ (using only communication on $E_{i,j}$ 
edges). We then collect all edges of $E_{i,j}$ to the root by pipelining these 
edges on $T_{i,j}$. At that point, each root of $T_{i,j}$ can partition all but 
$2n$ edges of $E_{i,j}$ into edge disjoint cycles of length $O(\log n)$. The 
root also pass these cycle information to the relevant edges using the 
communication on $T_{i,j}$. Note that since the $E_{i,j}$ subgraphs are 
disjoint, this can be done simultaneously for all  subgraphs $E_{i,j}$. At the 
end of that phase, we are left with $2n\cdot \ell_i=O(|E_i|/\log n)$ uncovered 
edges $E_{i+1}$ to be handled in the next phase. Overall, after $O(\log n/\log 
\log n)$ phases, we are left with $O(n\log n)$ uncovered edges. At the point, 
we can pipeline these edges to the root of the BFS tree, along with the $n-1$ 
edges of the BFS tree and let the root compute it locally as explained in 
\Cref{sec:cyclecoverall}. The lemma follows.
\end{proof}

\paragraph{Preprocessing algorithm for universally optimal covers.}
\begin{lemma}\label{lem:distmodcover}
Every distributed \emph{nice} algorithm $\cA$ 
that given a bridgeless graph $G$ with diameter $D$, constructs an $(f_{\cA}(D), \congestion)$ cycle cover $\cC$ within $r_{\A}(D)$ rounds can be transformed 
into an algorithm $\A'$ that constructs an 
$(f_{\cA}(\widetilde{O}(\OptCCVal(G))), \widetilde{O}(\congestion))$ cover $\cC'$ for $G$, within $r_{\A}(\OptCCVal(G))$ rounds. 
\end{lemma}
\begin{proof}
Algorithm $\cA'$ first employs \Cref{lem:distncimpl} to construct an $t$-neighborhood cover $\cN$ with for $t=\OptCCVal$ within $\widetilde{O}(\OptCCVal)$ rounds. Then, it applies Alg. $\cA$ on each subgraph $G[S_i]$ resulting in a cycle collection $\cC_i$. Since each vertex belongs to $\widetilde{O}(1)$ clusters, Algorithm $\cA$ can be applied on all graphs $G[S_i]$ simultaneously using $\widetilde{O}(r_{\A}(\OptCCVal))$ rounds, in total. The final cycle cover is $\cC=\bigcup_i \cC_i$. Since the diameter of each subgraph $G[S_i]$ is $\widetilde{O}(\OptCCVal)$, $\cC_i$ is an $(f_{\cA}(\widetilde{O}(\OptCCVal)), \congestion)$ cycle cover for the edges of $G[S_i]$ (i.e., covering the edges that lie on some cycle on $G[S_i]$). We have that $\cC$ is an $(f_{\cA}(\widetilde{O}(\OptCCVal)), \widetilde{O}(\congestion))$ cycle cover for $G$. 

To see that each edge $e$ is indeed covered, note that each edge $e$ lies on some cycle $C_e$ in $G$ of length at most $\OptCCVal$. By the properties of the neighborhood cover, w.h.p., there is a cluster $S_{i} \in \cN$ that contains all the vertices of $C_e$ and hence $e$ is an edge that lies on a cycle in the subgraph $G[S_i]$. Since the algorithm $\cA$ is nice, the edge $e$ is covered in the cycles of $\cC_i$. 

\end{proof}
By combining \Cref{lem:preproc-cover} with \Cref{lem:distmodcover}, we have:
\begin{lemma}\label{lem:preproc-cover}
For every bridgeless $n$-vertex graph a $(\widetilde{O}(\OptCCVal(G)),\widetilde{O}(1))$ cycle 
cover can be computed distributively in $\widetilde{O}(n)$ rounds of 
preprocessing.
\end{lemma}


\subsection{Neighborhood Covers}\label{sec:distncover}
In this section we describe how to construct neighborhood cover in the 
\congest\ model. As far as we know, previous explicit constructions for 
neighborhood cover (such as \cite{awerbuch1996fast}) are in the \local\ model 
and use large messages. For the definition of $(k,t,q)$ neighborhood cover, see 
\Cref{def:neighborcover}.
For ease of presentation, we construct a slightly weaker notion where the 
diameter of each cluster is $O(k \cdot t \cdot \log n)$ rather than $O(k \cdot 
t)$ as in \Cref{def:neighborcover} (this weaker notion suffices for our 
construction). This construction is implicit in the recent 
spanner construction of \cite{elkin2017efficient}. 

\begin{lemma}\label{lem:distncimpl}
For every integer $t$, and every $n$-vertex graph $G=(V,E)$, one can construct 
in $O(k \cdot t \cdot \log n)$ rounds, an $(k,t,q)$ neighborhood cover with $k=2\log n$, $q=O(\log 
n)$ and the strong diameter of each cluster is $O(t \cdot k \cdot \log n)$, w.h.p. In addition, there are $\widetilde{O}(1)$ messages that go through each edge $e \in G$ over the entire execution of the algorithm.
\end{lemma}
We first describe how using the spanner construction of 
\cite{elkin2017efficient}, we get a neighborhood cover that succeeds with 
constant probability. That is, we show that using the algorithm of \cite{elkin2017efficient} 
for constructing a $(k \cdot t)$-spanner, one can get in $O(k t \log n)$ rounds, a collection of 
subsets $\mathcal{S}=\{S_1,\ldots, S_n\}$ such that (I) the diameter of each 
$G[S_i]$ is $O(k \cdot t \cdot \log n)$, (II) w.h.p., each vertex belongs to 
$O(k\cdot n^{1/k})$ sets and (III) for each vertex $v$, there is a {\em 
constant}
probability that there exists $S_i$ that contains its entire $t$-neighborhood. 
Repeating this procedure for $O(\log n)$ many times yields the final cover.

We now describe the phase $i=\{1,\ldots, \Theta(\log n)\}$ where we construct a collection of $n$ sets $S_{i,u_1},\ldots, S_{i,u_n}$ that satisfy (I-III). 
Each vertex $u \in V$ samples a radius $r_u$ from the exponential 
distribution\footnote{Recall the exponential distribution with parameter 
$\beta$ where $f(x)=\beta\cdot e^{-\beta \cdot x}$ for $x\geq 0$ and $0$ 
otherwise.} with parameter $\beta=\ln(c\cdot n)/(3k \cdot t)$. Each vertex $u$ 
starts to broadcast\footnote{Having $u$ start at round $-r_u$ is not part of 
\cite{elkin2017efficient}. We introduced this modification to guarantee that 
the total of messages that are sent on each edge is at most $\widetilde{O}(1)$.}
its messages in round $-\lceil r_u \rceil$.
For a vertex $w$ that received (at least one) message in round $i$ for the \emph{first} time, let 
$m_{u_j}(w)=r_{u_j}-\dist(w,u_j,G)$ for every message originated from $u_j$ and 
received at $w$ in round $i$. Let $m(w)=\max_{u_j}m_{u_j}(w)$ and $u^* \in 
\Gamma(w)$ be such that $m(w)=m_{u^*}(w)$. Then, $w$ does the following: (i) 
store $m(w)$ and the neighbor $p_{u^*}(w)$ and (ii) sends the message $\langle 
w, 
m(w)-1 \rangle$ to all its neighbors $z \in \Gamma(w)\setminus \{p_{u_j}(w)\}$ 
in round $i+1$. 

For every vertex $u$, let $S_{i,u}=\{ w ~\mid~ m_u(w)\geq m(w)-1\}$. 
The final neighborhood cover is given by $\mathcal{S}=\bigcup_{i}\bigcup_{u} S_{i,u}$.

We show that the output collection of sets are indeed neighborhood cover. 
Fix a phase $i$, we claim the following about the output sets 
$S_{i,u_1},\ldots, S_{i,u_n}$.
\begin{claim}\label{lem:ncpropconstant}
\begin{description}
\item{(I)} Each $S_{i,u}$ is connected with diameter $O(k \cdot t \log n)$ with high probability.
\item{(II)} For every $i$, every vertex $w$ appears in $O(\log n \cdot (cn)^{1/(kt)})$ sets $S_{i,u}$ with high probability.
\item{(III)} For every vertex $w$, there exists $S_{i,u}$ such that $\Gamma_t(w)\subseteq S_{i,u}$, with constant probability.
\end{description}
\end{claim}
\begin{proof}
For ease of notation let $S_{i,u}=S_u$. 
For every $u$ and $w$, let $p_u(w)$ be the neighbor of $w$ that lies on the shortest path from $w$ to $u$, from which $w$ received the message about $u$ (breaking ties based on IDs). 

To show that each set $S_u$ is connected, it is sufficient to show that if $w \in S_u$ then also $p_u(w) \in S_u$. 
The proof is as Claim 5 in \cite{elkin2017efficient}. In particular, since $w$ 
and $w'=p_u(w)$ are neighbors, it holds that $m(w) \geq m(w')-1$ and hence $w' 
\in S_u$. In addition, by Claim 3 in \cite{elkin2017efficient} (and plugging 
our value of $\beta$) it holds that for every $u$ w.h.p.\ $r_u \leq O(k \cdot t 
\log n)$. 

We proceed with Claim (II). 
For each $w$ and $u$, let $X_{w,u}\in \{0,1\}$ be the random variable indicating that $w \in S_u$.
Let $Q_w=\sum_{u}X_{w,u}$ be the random variable of the number of sets to which 
$w$ belongs.
In Lemma 2 of \cite{elkin2017efficient} they show that for any $1 \le z \le n$ 
it holds that
\begin{align*}
\Pr[Q_w \ge z] \le (1-e^{-\beta})^{z-1}.
\end{align*}
Plugging in $z=c'\log n \cdot (cn)^{1/(kt)}+1 = c'\log n \cdot e^{\beta} + 1$ 
we get
\begin{align*}
\Pr[Q_w \ge z] \le (1-e^{-\beta})^{z-1} \le (1-e^{-\beta})^{c'\log n \cdot 
e^{\beta}} \le 1/n^{c'}.
\end{align*}
Taking a union on all nodes $w$ in the graph claim (II) follows.
%

Finally, consider claim (III), and consider the more strict event in which the 
entire $t$-neighborhood of $w$ belongs to $S_{u^*}$ where $u^*$ is the vertex 
that attains $m(w)=m_{u^*}(w)$ (breaking ties based on IDs). Let $Y_w$ be an 
indicator variable for this event.
We show that $Y_w=0$ with probability of at most constant. 
Consider $u^*$ as above, we bound the probability that there is a vertex $y \in 
\Gamma_t(w)$ that does not belong to $S_{u^*}$. We therefore have:
\begin{equation}\label{eq:nconeside}
m(y)> m_{u^*}(y)-1=r_{u^*}-\dist(u^*,y,G)-1\geq r_{u^*}-\dist(u^*,w,G)-t-1=m(w)-t-1~,
\end{equation}
and in the same manner, $m(w)\geq m(y)-t-1$. Therefore, $m(w) \in 
[m(y)-t-1,m(y)+t+1]$. That is, given that $m(w) \ge m(y)-t-1$ the probability 
that also $m(w) \le m(y)+t+1$ is at most $1-e^{-3\beta \cdot t}=1-\ln (cn)/k 
\leq c'$ for $k=2\log n$. 
The claim follows.
\end{proof}

We are now ready to complete the proof of \Cref{lem:distncimpl}.
\begin{proof}
It is easy to see that each phase can be implemented in $O(k \cdot t \log n)$ rounds.
In the distributed implementation of \cite{elkin2017efficient} (Sec.\ 2.1.1), 
the algorithm might pass $\widetilde{O}(t)$ messages on a given edge. For our 
purposes (e.g., distributed construction of private tree) it is important that 
on each edge the algorithm sends a total of $\widetilde{O}(1)$ messages. By 
letting each node $u$ start at round $-r(u)$, we make sure that the message 
from the node $w$ that maximizes $r_w-\dist(w,u,G)$ arrives first to $u$ and 
hence there is no need to send any other messages from other centers on that 
edge. 

By \Cref{lem:ncpropconstant}, w.h.p., all subsets have small diameter and 
bounded overlap. In addition, for every vertex $w$, with constant probability, 
the entire $t$-neighborhood of $w$ is covered by some of the $S_{i,u}$ sets.
Since we repeat this process for $O(\log n)$ times, w.h.p., there exists a set 
that covers $\Gamma_t(w)$. By applying the union bounded overall sets, we get 
that w.h.p. all vertices are covered, the Lemma follows.
\end{proof}


\end{document}